%% file: main.tex
\documentclass[preprint,12pt]{elsarticle}
\journal{Journal of Pure and Applied Algebra}

\include{macros}

\begin{document}

\begin{frontmatter}

\title{Interacting Hopf Algebras}

\author[ens]{Filippo Bonchi}
\ead{filippo.bonchi@ens-lyon.fr}
\author[south]{Pawe{\l} Soboci\'{n}ski}
\ead{ps@ecs.soton.ac.uk}
\author[nijmegen]{Fabio Zanasi}
\ead{fzanasi@cs.ru.nl}

\address[ens]{ENS de Lyon, Universit\'{e} de Lyon, CNRS, INRIA, France}
\address[south]{University of Southampton, United Kingdom}
\address[nijmegen]{Radboud University of Nijmegen, Netherlands}



\begin{abstract}
We introduce the theory $\IBR$ of interacting Hopf algebras, parametrised over a principal ideal domain $\PID$. The axioms of $\IBR$ are derived using Lack's approach to composing PROPs: they feature two Hopf algebra and two Frobenius algebra structures on four different monoid-comonoid pairs. This construction is instrumental in showing that $\IBR$ is isomorphic to the PROP of linear relations (i.e. subspaces) over the field of fractions of $\PID$.
\end{abstract}
\end{frontmatter}

\section{Introduction}

\input{./source/intro}

\section{Background}\label{sec:background}

\input{./source/Background}

\section{Hopf Algebras: the Theory of Matrices}\label{sec:theorymatr}

\input{./source/0_ABR}

\section{Generalising Distributive Laws by Pullback and Pushout}\label{sec:distrLawPullback}

\input{./source/distrlawspbpo}

\section{Interacting Hopf Algebras I: the Theories of Spans and Cospans of Matrices}\label{sec:ibrw}

\input{./source/1_IBR}

\subsection{Compact Closed Structure of $\IBRw$}\label{sec:cc}

\input{./source/2_CompactClosed}

\subsection{$\IBRw$: the theory of spans of $\PID$-matrices}\label{sec:completeness}

\input{./source/25_KernelLemma}

\paragraph{Circuits of Invertible Matrices}

\input{./source/4_InvMatrices}

\paragraph{Computing Kernels in $\IBRw$}

\input{./source/5_ProofMatrixEqualKernel}


\input{./source/6_ProofCompleteness}

\input{./source/8_cubeibwface}
\subsection{$\IBRb$: the theory of cospans of $\PID$-matrices}\label{sec:IBRbCospan}

\input{./source/9_cubeibbface}

\section{Interacting Hopf Algebras II: the Theory of Linear Subspaces}\label{sec:cubetop}

\input{./source/10_cubetopface}

\subsection{The Cube: Bottom Face}\label{sec:cubebottom}

\input{./source/cube2}

\subsection{The Cube: Rear Faces}\label{sec:cubeback}

\input{./source/7_cubebackwardfaces}

\subsection{The Cube Rebuilt}\label{sec:cuberebuilt}

\input{./source/11_cuberebuilt}

\section{Example: Interacting Hopf Algebras for Rational Subspaces}\label{sec:instances}

\input{./source/12_cubeinstances}

\bibliographystyle{elsarticle-num}
\bibliography{catBib3}

\newpage

\appendix

\section{The Frobenius Laws in $\IBRw$}\label{AppFrob}

\input{./source/A_AppendixFrobLaws}

\section{Derived Laws of $\IBRw$}\label{AppDerLaws}

\input{./source/C_AppendixLaws}

\section{Shaping the Compact Closed Structure of $\IBRw$}\label{AppCC}

\input{./source/B_AppendixCC}

\section{Derived Laws of $\IBR$}\label{AppDerLawsIH}

\input{./source/D_AppendixLawsIB}

\end{document}

%% file: macros.tex
\usepackage{amsmath}
\usepackage{graphicx,epsfig,color}
\usepackage{xypic}
\xyoption{rotate}
\input xy
\xyoption{all}
\usepackage{verbatim}
\usepackage{amssymb}
\usepackage{eucal}
\usepackage{enumerate}
\usepackage[shortlabels]{enumitem}

\usepackage{multicol}
\usepackage{manfnt}
\usepackage[hyperindex,breaklinks]{hyperref}

\usepackage{xspace}
\usepackage{amsfonts}
\usepackage{mathrsfs}
\usepackage[all,2cell]{xy}
\UseAllTwocells
\SilentMatrices

\usepackage{bbm}
\DeclareMathOperator{\Mat}{\mathsf{Mat}}

\mathcode`:="003A  
\mathcode`;="003B  
\mathcode`?="003F  
\mathcode`|="026A  
\mathcode`<="4268  
\mathcode`>="5269  

\mathchardef\ls="213C    
\mathchardef\gr="213E    
\mathchardef\uparrow="0222  
\mathchardef\downarrow="0223  

\newenvironment{Iff-RL}{\textbf{($\Rightarrow$)} }{\bigskip}
\newenvironment{Iff-LR}{\textbf{($\Leftarrow$)} }{}

\newcommand{\m}{\mathit}
\def \: {\ensuremath{\colon}}
\newcommand{\mb}{\mathbb}

\newcommand\id{\m{id}}

\newcommand\funF{\mathcal{F}}

\newcommand\T{\mathcal{T}}

\def \catC {\mathbb{C}}
\def \Set  {\mathbf{Set}}

\def \N {\mathbb{N}} 
\def \Z {\mathbb{Z}} 
\def \Q {\mathbb{Q}} 
\def \T {\mathbb{T}} 

\def \PROP {\mathbf{PROP}} 
\def \PRO {\mathbf{PRO}} 
\def \Perm {\mathbb{P}} 
\def \F {\mathbb{F}}
\def \Fop {\F^{\m{op}}}

\def \Vect {\Mat{\Z}} 

\def \AB {\mathbb{AB}}

\newcommand{\Span}[1]{\mathsf{Span}(#1)}
\newcommand{\Cospan}[1]{\mathsf{Cospan}(#1)}

\def \wmon {\mf{Mon}^w}

\def \bcom {\mf{Com}^b}

\def \tns {\oplus}

\newcommand{\eql}[1]{\ \overset{\text{#1}}=\ }

\def \SV {\mathbb{SV}} 
\newcommand{\sem}[1]{\mathcal{Sem}_{#1}} 

\newcommand{\initVect}{!} 

\DeclareMathSymbol{\reversedExclMark}{\mathord}{operators}{"3C}
\newcommand{\finVect}{\reversedExclMark} 

\newcommand\Gmult{\lower5pt\hbox{$\includegraphics[width=20pt]{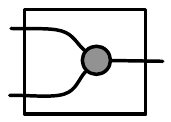}$}}
\newcommand\Gcomult{\lower5pt\hbox{$\includegraphics[width=20pt]{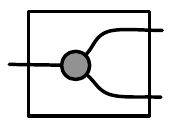}$}}
\newcommand\Gunit{\lower5pt\hbox{$\includegraphics[width=16pt]{graffles/Gunit.pdf}$}}
\newcommand\Gcounit{\lower5pt\hbox{$\includegraphics[width=16pt]{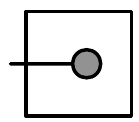}$}}
\newcommand\twoGcounit{\lower5pt\hbox{$\includegraphics[width=16pt]{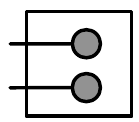}$}}

\newcommand\Bmult{\lower5pt\hbox{$\includegraphics[width=20pt]{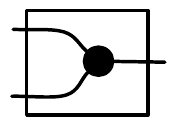}$}}
\newcommand\Bcomult{\lower5pt\hbox{$\includegraphics[width=20pt]{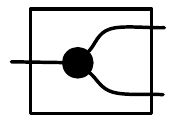}$}}
\newcommand\Bunit{\lower5pt\hbox{$\includegraphics[width=16pt]{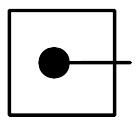}$}}
\newcommand\Bcounit{\lower5pt\hbox{$\includegraphics[width=16pt]{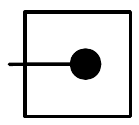}$}}

\newcommand\Wmult{\lower5pt\hbox{$\includegraphics[width=20pt]{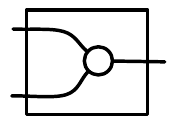}$}}
\newcommand\Wcomult{\lower5pt\hbox{$\includegraphics[width=20pt]{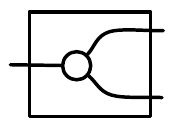}$}}
\newcommand\Wunit{\lower5pt\hbox{$\includegraphics[width=16pt]{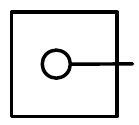}$}}
\newcommand\Wcounit{\lower5pt\hbox{$\includegraphics[width=16pt]{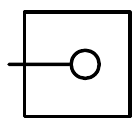}$}}

\newcommand\ASepUno{\lower3pt\hbox{$\includegraphics[width=28pt]{graffles/ASep1.pdf}$}}
\newcommand\ASepDue{\lower3pt\hbox{$\includegraphics[width=24pt]{graffles/ASep2.pdf}$}}
\newcommand\Idnet{\lower3pt\hbox{$\includegraphics[width=20pt]{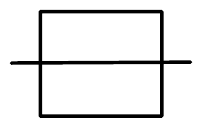}$}}
\newcommand\symNet{\lower3pt\hbox{$\includegraphics[width=20pt]{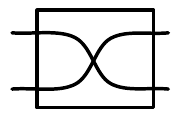}$}}

\newcommand{\pullbackcorner}[1][d]{\save*!/#1+1.0pc/#1:(1,-1)@^{|-}\restore}
\newcommand{\pushoutcorner}[1][d]{\save*!/#1-1.0pc/#1:(-1,1)@^{|-}\restore}
\newcommand\BlackX{\lower5pt\hbox{\includegraphics[width=30pt]{graffles/BlackX.pdf}}}
\newcommand\BSep{\lower5pt\hbox{\includegraphics[width=35pt]{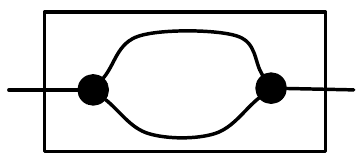}}}
\newcommand\WSep{\lower5pt\hbox{\includegraphics[width=35pt]{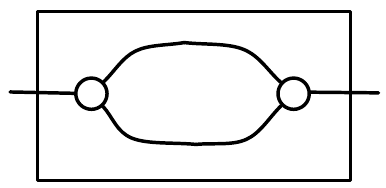}}}

\newcommand{\matrixOneOne}{\left(%
                \begin{array}{c}
                 \!\! 1 \!\!\\
                  \!\! 1 \!\!
                \end{array}\right)}

\newcommand{\matrixOneOneFlat}{\left(%
                \begin{array}{cc}
                \!\!  1 \!& \! 1 \!\!
                \end{array}\right)}

\newcommand{\matrixOne}{\left(%
                \begin{array}{c}
                \!\! 1 \!\!
                \end{array}\right)}

\newcommand{\matrixZero}{\left(%
                \begin{array}{c}
                \!\!  0 \!\!
                \end{array}\right)}

\newcommand{\matrixNull}{\left(%
                \begin{array}{c}
                \!\!
                \end{array}\right)}

\newcommand{\tr}[1]{\xrightarrow{#1}}    
\newcommand{\tl}[1]{\xleftarrow{#1}}    

\newcommand{\op}{op}

\def \PID {\mathsf{R}}
\def \PROPR {\mathbb{R}} 

\newcommand \HA[1] {\mathbb{HA}_{\scriptscriptstyle #1}}
\newcommand \IH[1] {\mathbb{IH}_{\scriptscriptstyle #1}}
\def \ABR {\HA{\PID}}
\def \IBR {\IH{\PID}} 
\def \ABRop {\ABR^{\m{op}}} %

\def \frPID {\mathsf{k}} 
\def \Mon {\mathbb{M}}
\def \Com {\mathbb{C}}
\def \wmon {\Mon}
\def \bcom {\Com}
\def \To {\Rightarrow}
\def \VectR {\Mat{\PID}}
\newcommand \SVH[1] {\mathbb{SV}_{\scriptscriptstyle #1}}
\def \SVR {\SVH{\frPID}}
\DeclareMathOperator{\RModule}{\mathsf{Mod}}
\DeclareMathOperator{\FRModule}{\mathsf{FMod}}
\DeclareMathOperator{\Kernel}{Ker}
\DeclareMathOperator{\Cokernel}{Coker}
\newcommand \RMod[1] {\RModule{#1}} 
\newcommand \FRMod[1] {\FRModule{#1}} 
\newcommand \Ker[1] {\Kernel(#1)} 
\newcommand \Coker[1] {\Cokernel(#1)} 
\def \VectRop {\Mat{\PID}^{\m{op}}}
\def \VectSpFr {\FRMod{\frPID}}
\def \poi {\,\ensuremath{;}\,}
\newcommand \coc[1] {{#1}^{\star}}
\newcommand \refl[1] {{#1}^{R}}
\def \df {\ \ensuremath{:\!\!=}\ }
\def \dfop {\ \ensuremath{=\!\!:}\ }
\def \minus {\ensuremath{\!-\!\!}}
\newcommand \restr[2] {#1_{\upharpoonright #2}} 
\newcommand \tra {\ensuremath{\mathcal{T}}} 
\newcommand \pn {\ensuremath{\mathcal{N}}} 

\newcommand\scalar{\!\lower5pt\hbox{$\includegraphics[width=22pt]{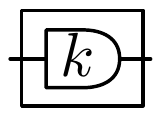}$}\!}
\newcommand\coscalar{\!\lower5pt\hbox{$\includegraphics[width=22pt]{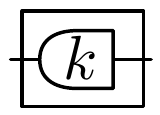}$}\!}
\newcommand\unitscalar{\lower4pt\hbox{$\includegraphics[width=22pt]{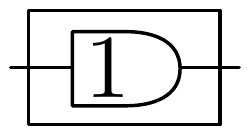}$}}
\newcommand\scalarminusone{\lower8pt\hbox{$\includegraphics[width=30pt]{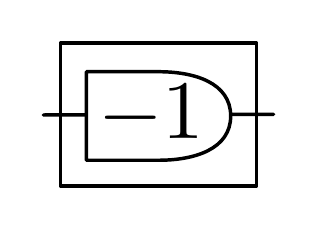}$}}
\newcommand\antipode{\lower4pt\hbox{$\includegraphics[width=20pt]{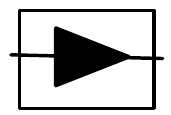}$}}
\newcommand\antipodeop{\lower3pt\hbox{$\includegraphics[width=22pt]{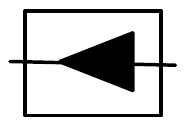}$}}
\newcommand\antipodesquare{\lower3pt\hbox{$\includegraphics[width=22pt]{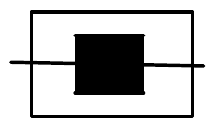}$}}
\newcommand\circuitAdots{\lower6pt\hbox{$\includegraphics[width=30pt]{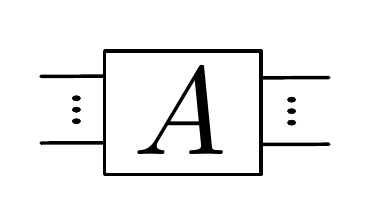}$}}

\newcommand\wcounitn{\lower5pt\hbox{$\includegraphics[width=25pt]{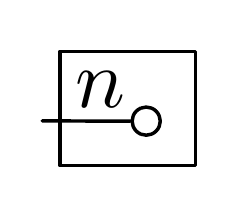}$}}
\newcommand\bcounitn{\lower5pt\hbox{$\includegraphics[width=25pt]{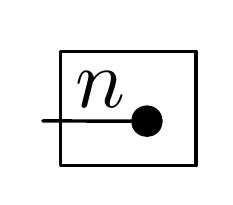}$}}
\newcommand\lccn{\lower5pt\hbox{$\includegraphics[width=25pt]{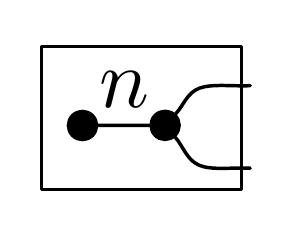}$}}
\newcommand\rccn{\lower5pt\hbox{$\includegraphics[width=25pt]{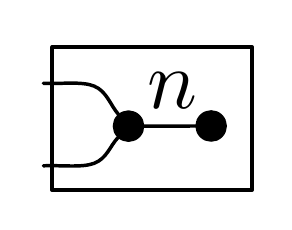}$}}
\newcommand\idncircuit{\lower5pt\hbox{$\includegraphics[width=25pt]{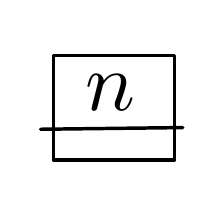}$}}
\newcommand\circuitrbcounits{\lower5pt\hbox{$\includegraphics[width=25pt]{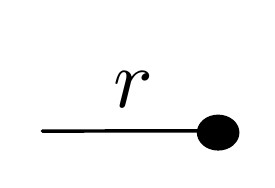}$}}
\newcommand\lccB{\lower5pt\hbox{$\includegraphics[width=25pt]{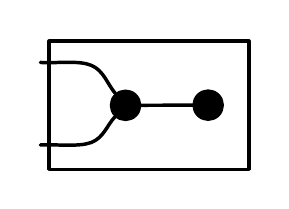}$}}
\newcommand\rccB{\lower5pt\hbox{$\includegraphics[width=25pt]{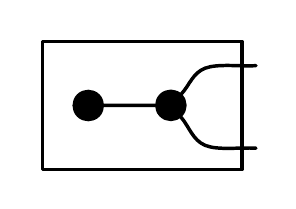}$}}
\newcommand\IdBcounitc{\lower5pt\hbox{$\includegraphics[width=20pt]{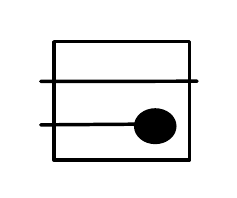}$}}
\newcommand\BcounitId{\lower5pt\hbox{$\includegraphics[width=20pt]{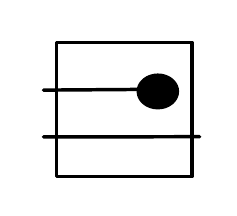}$}}
\newcommand\symNetTwoOne{\lower7pt\hbox{$\includegraphics[width=25pt]{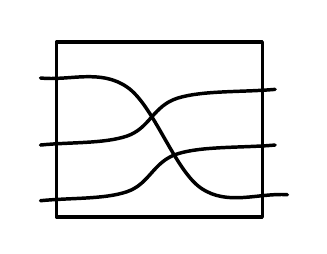}$}}
\newcommand\nscalar{\!\!\lower5pt\hbox{$\includegraphics[width=35pt]{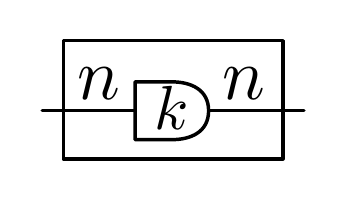}$}\!\!}
\newcommand\Wmultstar{\!\lower5pt\hbox{$\includegraphics[width=20pt]{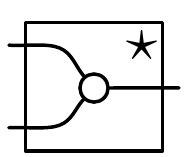}$}\!}
\newcommand\scalarstar{\!\!\lower7pt\hbox{$\includegraphics[width=35pt]{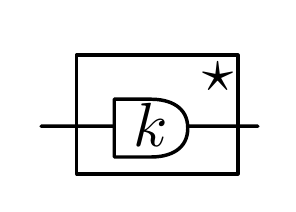}$}\!\!}
\newcommand\twoBcounit{\lower5pt\hbox{$\includegraphics[width=15pt]{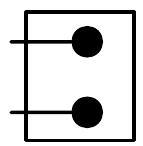}$}}

\newcommand\Wcomm{\!\lower8pt\hbox{$\includegraphics[width=33pt]{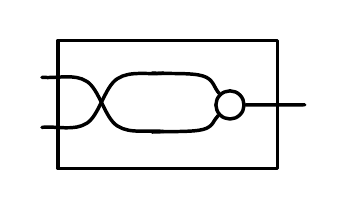}$}\!}
\newcommand\delay{\!\lower6pt\hbox{$\includegraphics[width=25pt]{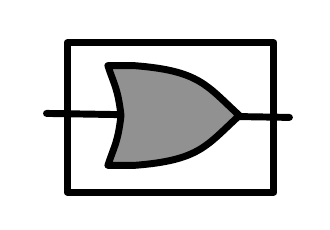}$}\!}
\newcommand\circuitX{\!\lower6pt\hbox{$\includegraphics[width=25pt]{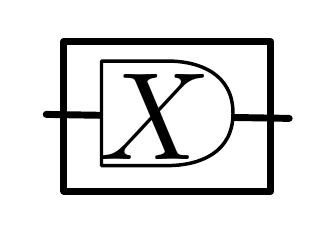}$}\!}
\newcommand\idcircuit{\lower3pt\hbox{$\includegraphics[width=20pt,height=13pt]{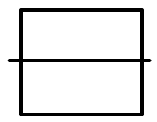}$}}
\newcommand\antipodeantipode{\lower7pt\hbox{$\includegraphics[width=30pt]{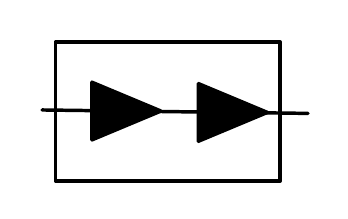}$}}
\newcommand\antipodeopantipodeop{\lower8pt\hbox{$\includegraphics[width=30pt]{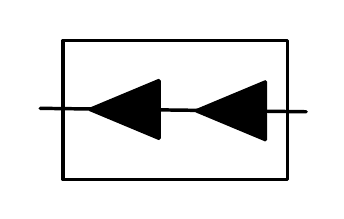}$}}
\newcommand\BcomultSingleAntipode{\lower9pt\hbox{$\includegraphics[width=21pt]{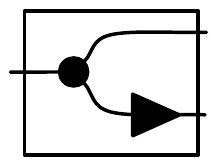}$}}

\newcommand{\vlist}[1]{\mathbf{#1}}
\newcommand{\xx}{\vlist{x}}
\newcommand{\yy}{\vlist{y}}
\newcommand{\zz}{\vlist{z}}
\newcommand{\uu}{\vlist{u}}
\newcommand{\vv}{\vlist{v}}
\newcommand{\ww}{\vlist{w}}

\newtheorem{theorem}{Theorem}[section]
\newtheorem{proposition}[theorem]{Proposition}
\newtheorem{corollary}[theorem]{Corollary}
\newtheorem{lemma}[theorem]{Lemma}

\newtheorem{definition}[theorem]{Definition}

\newtheorem{example}[theorem]{Example}

\newtheorem{remark}[theorem]{Remark}

\newproof{proof}{Proof}

\newcommand{\spaceFull}{\lower4pt\hbox{$\includegraphics[height=15pt]{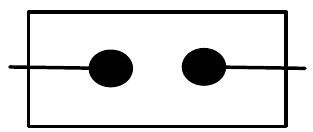}$}}
\newcommand{\spaceZero}{\lower4pt\hbox{$\includegraphics[height=15pt]{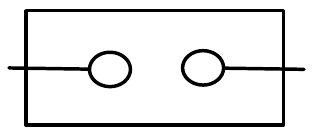}$}}
\newcommand{\spaceYaxis}{\lower4pt\hbox{$\includegraphics[height=15pt]{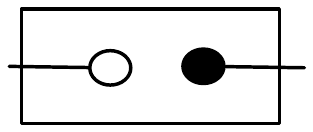}$}}
\newcommand{\spaceXaxis}{\lower4pt\hbox{$\includegraphics[height=15pt]{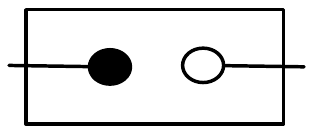}$}}
\newcommand{\scalarktwo}{\lower5.5pt\hbox{$\includegraphics[height=18pt]{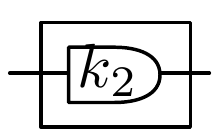}$}}
\newcommand{\spacekonektwo}{\lower5pt\hbox{$\includegraphics[height=18pt]{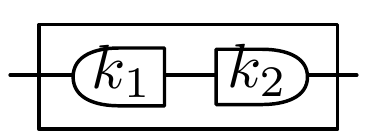}$}}

\def \SupSpan {\mathsf{Sp}}
\def \SupCospan {\mathsf{Cp}}
\def \IBRw {\IBR^{\scriptscriptstyle \SupSpan}} 
\def \IBRb {\IBR^{\scriptscriptstyle \SupCospan}} 

\newcommand\matrixTwoROneC[2]{{\tiny\begin{pmatrix}  #1 \\ #2 \end{pmatrix}}}

\def \PJ {\mathbb{J}} 
\def \PS {\mathbb{S}} 
\def \Top {\T^{\op}} 

%% file: source/intro.tex
We introduce the theory of Interacting Hopf Algebras, characterising linear relations. Its equations are obtained via Lack's composition of PROPs~\cite{Lack2004a}.

Diagrammatic formalisms are widespread in various fields, including computer science, control theory, logic and quantum information~\cite{BaezRosetta}.
Several recent approaches~\cite{Abramsky2004,Coecke2008,Coecke2012,Selinger2009,Ghica13,Pavlovic13,Fiore2013,Bonchi2014b,Bonchi2015,Fong2015,Baez2014a,brunihierarchical} consider diagrams rigorously as the arrows of a \emph{symmetric monoidal theory} (SMTs). By SMT we mean a presentation of a PROP: a set of generators---the \emph{syntax} of diagrams---together with a set of \emph{equations} that, in conjunction with the usual laws of symmetric monoidal categories, give the notion of diagram equality. Of particular importance are SMT featuring both algebraic and coalgebraic structure, subject to compatibility conditions: notable examples are Frobenius algebras and bialgebras whose equations witness an \emph{interaction} between a commutative monoid and cocommutative comonoid.

Lack~\cite{Lack2004a} showed that several such situations can be understood as arising from \emph{PROP composition} where a \emph{distributive law}---a notion closely related to standard distributive laws between monads~\cite{Street_MonadsI}---witnesses the interaction. The beauty of this approach is that one can consider distributive laws to be responsible for the newly introduced equations, resulting in a pleasantly modular account of the composite algebraic theory.
For example, the equations of (strongly separable) Frobenius algebra~\cite{Carboni1987} can be obtained in this way.  Another example is the theory of bialgebras: here monoids and comonoids interact through a different distributive law, thus yielding different equations.

\medskip
Our chief original contribution is the study of the interaction of the PROP $\ABR$ of Hopf algebras, parametrised over a principal ideal domain $\PID$, and its opposite $\ABRop$. As in the case of the PROP of commutative monoids and its opposite, two different distributive laws can be defined, yielding $\IBRw$ and $\IBRb$ respectively. 
Our main theory of interest
$\IBR$ is the result of merging together these two equational theories.
These ingredients constitute the topmost face in the following commutative cube in the category of PROPs.
\begin{equation}\label{eq:cube}
\tag{\mancube}
\raise30pt\hbox{$
\xymatrix@=5pt{
& {\ABR + \ABRop} \ar[dd]_(.3){\cong}|{\hole}
\ar[dl] \ar[rr] & & {\IBRw} \ar[dl] \ar[dd]^{\cong} \\
{\IBRb} \ar[rr] \ar[dd]_{\cong}  & & {\IBR} \ar@{.>}[dd] \\
& {\VectR+ \VectRop} \ar[dl] \ar[rr]|(.57){\hole} & & {\Span {\VectR}} \ar[dl] \\
{\Cospan {\VectR}} \ar[rr] & & {\SVR}
}$}
\end{equation}
The bottom face of~\eqref{eq:cube} describes the linear algebraic nature of our SMTs. First, $\ABR$ is isomorphic to the PROP $\VectR$ of $\PID$-matrices. Second, since the equations of $\IBRw$ and $\IBRb$ arise from distributive laws, these SMTs isomorphic to PROPs of spans and cospans of $\PID$-matrices, respectively --- these latter PROPs exist because  $\VectR$ has pullbacks and pushouts whenever $\PID$ is a principal ideal domain. The isomorphism between $\IBR$ and $\SVR$ follows from the fact that the top and the bottom faces of~\eqref{eq:cube} are pushouts. $\SVR$ is the PROP of linear relations over the field $\frPID$ of  $\PID$-fractions: an arrow $n \to m$ is a $\frPID$-linear subspace of $\frPID^{n}\times \frPID^{m}$, composition is relational.

We contend that $\IBR$ is a \emph{canonical syntax for (finite dimensional) linear algebra}: linear transformations, spaces, kernels, etc. are all represented faithfully in the graphical language. This perspective will be pursued in the paper: several proofs mimic---at the diagrammatic level---familiar techniques such as Gaussian elimination. We believe that that the string-diagrammatic treatment of linear algebra is of cross-disciplinary benefit: indeed, some applications of the theory herein have already been developed; see below.


\paragraph{Applications and related work}

For different choices of $\PID$, the theory of interacting Hopf algebras has several applications in diverse disciplines. A particularly interesting instance is the polynomial ring $\PID=\mathbb{R}[x]$: $\mathbb{IH}_{\scriptscriptstyle \mathbb{R}[x]}$ is a string-diagrammatic account of \emph{signal-flow graphs}, which are foundational structures of control theory and signal processing that capture behaviour defined via recurrence relations/differential equations.  $\mathbb{IH}_{\scriptscriptstyle \mathbb{R}[x]}$ provides a formal syntax and semantics, a sound and complete equational theory and an analogue of Kleene's theorem~\cite{Kleene} stating that all rational behaviours can be denoted within $\mathbb{IH}_{\scriptscriptstyle \mathbb{R}[x]}$. The interested reader is referred to~\cite{Bonchi2014b,Bonchi2015,Fong2015}.

After the submission of~\cite{Bonchi2014b} and the appearance of an earlier version of this manuscript on arXiv (\url{http://arxiv.org/abs/1403.7048}), Baez and Erbele~\cite{Baez2014a} independently gave an equivalent presentation of $\mathbb{IH}_{\scriptscriptstyle \mathbb{R}[x]}$. The main difference is our use of distributive laws, which 
 enables us to obtain $\IBR \cong \SVR$ using universal properties as well as the span/cospan factorisations in $\IBR$.

An earlier conference version of this work appeared in~\cite{BialgAreFrob14} and only considered the theory $\mathbb{IH}_{\scriptscriptstyle \mathbb{Z}_2}$, 
 which also has significant applications.
First, it is closely related to the algebra of stateless connectors~\cite{Bruni2006}, modeling concurrent interactions of software components. Second, it is the phase-free fragment of the ZX-calculus~\cite{Coecke2008,Coecke2009a}, an SMT for interacting quantum observables which originated in the research programme of categorical quantum mechanics~\cite{Abramsky2004,Abramsky2008:CQM}. Completeness for ZX has been intensively studied in recent years~\cite{Perdrix_completenessZX15,ZXIncomplete14,Backens-ZXcompleteness2} and our work yields a free model $\SV_{\scriptscriptstyle \mathbb{Z}_2}$ for the phase-free fragment. 
Our modular analysis also gives new insights about the algebra of quantum theories: while the 
  Frobenius structures have traditionally been regarded as being fundamental,  our
 construction reveals that the basic blocks are Hopf algebras, and the Frobenius equations arise by their composition.

\paragraph{Synopsis} Section~\ref{sec:background} provides the background on SMTs and composing PROPs.
In Section~\ref{sec:theorymatr} we recall the theory of Hopf Algebras on a principal domain $\PID$ and show that it presents the PROP of $\PID$-matrices. Section~\ref{sec:distrLawPullback} introduces a mild generalisation of Lack's technique for composing PROPs, which is needed to accommodate the case of interacting Hopf algebras.

In Section~\ref{sec:ibrw} we introduce the theories of interacting Hopf algebras for span and cospans of $\PID$-matrices. First, $\IBRw$ and its compact closed structure are introduced (Subsection~\ref{sec:cc}). Then, Subsection~\ref{sec:completeness} is devoted to proving that $\IBRw$ presents the PROP of spans of matrices. Finally, in Subsection~\ref{sec:IBRbCospan} we also give the presentation $\IBRb$ for cospans of matrices.

Section~\ref{sec:cubetop} concerns the theory of interacting Hopf algebras for linear subspaces. To obtain the characterisation we show that the bottom face of $\eqref{eq:cube}$ is a pushout (Subsection~\ref{sec:cubebottom}) and that the rear faces commute (Subsection~\ref{sec:cubeback}). 

Section~\ref{sec:instances} is an example of our construction: the theory of interacting Hopf algebras for rational subspaces.

%% file: source/Background.tex
\paragraph{Notation} 
$\catC[a,b]$ is the set of arrows from $a$ to $b$ in a small category $\catC$, composition of
$f \: a \to b$, $g\: b \to c$ is written $f\poi g \: a \to c$.  We will sometimes write
$a\tr{f}b$ or $a\tr{f \in \catC}b$ for $f \: a \to b$ in $\catC$. When names are unnecessary we simply write $\tr{\in \catC}$ or $\tr{}$ if $\catC$ is clear from the context. For $\catC$ symmetric monoidal, $\tns$ is its monoidal product and $\sigma_{a,b} \: a \tns b \to b \tns a$ is the symmetry associated with $a,b \in \catC$. 
Given $\catC$ with pullbacks, its span bicategory has the objects of $\catC$ as $0$-cells, spans of arrows of $\catC$ as $1$-cells and span morphisms as $2$-cells. We denote with $\Span{\catC}$ the \emph{category} obtained by identifying the isomorphic $1$-cells and forgetting the $2$-cells. Dually, if $\catC$ has pushouts we can form its bicategory of cospans and denote with $\Cospan{\catC}$ the category obtained by identifying the isomorphic $1$-cells and forgetting the $2$-cells.

\subsection{PROPs}\label{ssec:propsbackground}

A one-sorted \emph{symmetric monoidal theory} (SMT) is determined by $(\Sigma, E)$ where $\Sigma$ is the \emph{signature}: a set of \emph{generators} $o \: n\to m$ with \emph{arity} $n$ and \emph{coarity} $m$ where $m,n\in\N$. The set of $\Sigma$-terms is obtained by combining generators in $\Sigma$, the unit $\id \: 1\to 1$ and the symmetry $\sigma_{1,1} \: 2\to 2$ with $;$ and $\tns$. This is a purely formal process: given $\Sigma$-terms $t \: k\to l$, $u \: l\to m$, $v \: m\to n$, we construct new $\Sigma$-terms $t \poi u \: k\to m$ and $t \tns v \: k+n \to l+n$.  The set $E$ of \emph{equations} contains pairs of $\Sigma$-terms of the form $(t,t':k\to l)$; the only requirement is that $t$ and $t'$ have the same arity and coarity as $\Sigma$-terms.

SMTs are presentations of PROPs~\cite{MacLane1965,Lack2004a} (\textbf{pro}duct and \textbf{p}ermutation categories). A PROP is a strict symmetric monoidal category with objects natural numbers, where $\tns$ on objects is addition. Morphisms between PROPs are identity-on-objects strict symmetric monoidal functors. PROPs and their morphisms form the category $\PROP$.
Any SMT $(\Sigma,E)$ freely generates a PROP by letting the arrows $n\to m$ be the set of $\Sigma$-terms $n\to m$ modulo the laws of symmetric monoidal categories and the (smallest congruence containing the) equations $t=t'$ for any $(t,t')\in E$. There is a natural graphical representation of these terms as arrows of monoidal categories (see~\cite{Selinger2009}): we will commonly refer to these \emph{string diagrams}
as \emph{circuits}.

For example, let $(\Sigma_M,E_M)$ be the SMT of commutative monoids. The signature $\Sigma_M$ contains two generators:  multiplication --- which we depict as a circuit $\Wmult \: 2 \to 1$ --- and unit, represented as $\Wunit \: 0 \to 1$.
 Graphically, the generation of $\Sigma_M$-terms amounts to ``tiling'' $\Wmult$ and $\Wunit$ together with the circuit $\symNet$ ($\sigma_{1,1}\: 2 \to 2$) and $\Idnet$ ($\id_1 \: 1 \to 1$). Equations $E_M$ assert associativity \eqref{eq:wmonassoc}, commutativity \eqref{eq:wmoncomm} and identity \eqref{eq:wmonunitlaw}.
 \begin{multicols}{3}\noindent
\begin{equation}
\label{eq:wmonunitlaw}
\tag{A1}
\lower11pt\hbox{$\includegraphics[height=1cm]{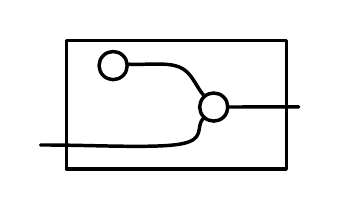}$}
\!\!\!
=\!
\lower5pt\hbox{$\includegraphics[height=.6cm]{graffles/idcircuit.pdf}$}
\end{equation}
\begin{equation}
\label{eq:wmoncomm}
\tag{A2}
\lower5pt\hbox{$\includegraphics[height=.6cm]{graffles/Wmult.pdf}$}
\!
=
\!\!\!\!
\lower11pt\hbox{$\includegraphics[height=1cm]{graffles/Wcomm.pdf}$}
\end{equation}
\begin{equation}
\label{eq:wmonassoc}
\tag{A3}
\lower12pt\hbox{$\includegraphics[height=1cm]{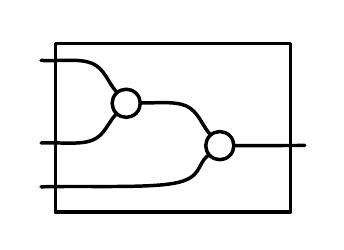}$}
\!\!\!
=
\!\!\!
\lower12pt\hbox{$\includegraphics[height=1cm]{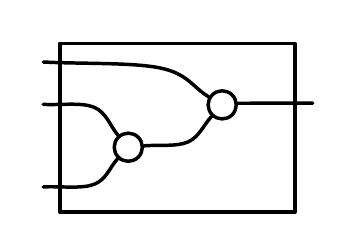}$}
\end{equation}
\end{multicols}
Let $\wmon$ denote the PROP freely generated by $(\Sigma_M,E_M)$. 
For later reference, we also introduce the PROP $\bcom$ of commutative comonoids, generated by the signature consisting of circuits $\Bcomult$, $\Bcounit$ and the following equations.
\begin{multicols}{3}\noindent
\begin{equation}
\label{eq:bcomonunitlaw}
\tag{A4}
\lower11pt\hbox{$\includegraphics[height=1cm]{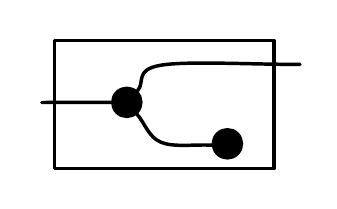}$}
\!\!\!
=
\!
\lower7pt\hbox{$\includegraphics[height=.7cm]{graffles/idcircuit.pdf}$}
\end{equation}
\begin{equation}
\label{eq:bcomoncomm}
\tag{A5}
\lower5pt\hbox{$\includegraphics[height=.6cm]{graffles/Bcomult.pdf}$}
\!
=
\!\!\!
\lower11pt\hbox{$\includegraphics[height=1cm]{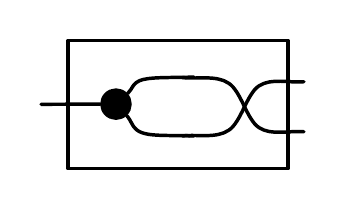}$}
\end{equation}
\begin{equation}
\label{eq:bcomonassoc}
\tag{A6}
\lower11pt\hbox{$\includegraphics[height=1cm]{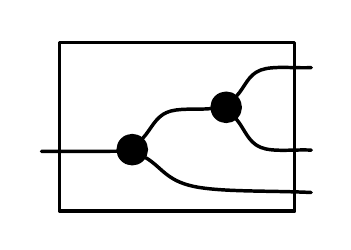}$}
\!\!\!
=
\!\!\!
\lower11pt\hbox{$\includegraphics[height=1cm]{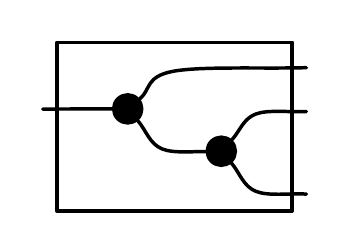}$}
\end{equation}
\end{multicols}

Modulo the white vs. black colouring---which will be justified later---the circuits of $\bcom$ can be seen as those of $\wmon$ ``reflected about the $y$-axis''. This observation yields $\bcom \cong (\wmon {})^{op}$. 

\begin{remark}[Models of a PROP] \label{rmk:models}
The assertion that $(\Sigma_M, E_M)$ \emph{is the SMT of commutative monoids}---and similarly for other SMTs appearing in our exposition---can be made precise using the notion of \emph{model} (sometimes also called algebra) of a PROP. 
Given a strict symmetric monoidal category $\catC$, a model of a PROP $\T$ in $\catC$ is a symmetric strict monoidal functor $\funF \: \T \to \catC$. Then $\mathsf{Model}(\T,\catC)$ is the category whose objects are the models of $\T$ in $\catC$. 

Turning to commutative monoids, there is a category $\mathsf{Monoid}(\catC)$ whose objects are the commutative monoids in $\catC$, i.e., objects $x \in \catC$ equipped with arrows $x \tns x \to x$ and $I \to x$, satisfying the usual equations.
Given any model $\funF \: \Mon \to \catC$, it follows that $\funF(1)$ is a commutative monoid in $\catC$: this yields a functor $\mathsf{Model}(\Mon, \catC) \to \mathsf{Monoid}(\catC)$. Saying that $(\Sigma_M, E_M)$ is the SMT of commutative monoids means that this functor is an equivalence natural in $\catC$. 
We shall not focus on models as they are not necessary for our applications: for us, the theory $\IBR$ of Interacting Hopf Algebras is more interesting as a diagrammatic language to express and reason about linear systems.  
\end{remark}

We will sometimes adopt the language of formal logic and refer to the free PROPs which arise from SMTs --- e.g.\ to the examples above --- as \emph{syntactic} PROPs in order to distinguish them from \emph{semantic} PROPs: an example of the latter is $\F$ where arrows $n\to m$ are functions 
\[\{0,\dots, n-1\}\to \{0, \dots, m-1\}.\] 
 There is an isomorphism $\wmon \cong \F$ that takes circuits (syntax) $c \in \wmon[n,m]$ to functions (semantics) of type $\{0,\dots, n-1\}\to \{0, \dots, m-1\}$. For instance, $\Wmult \tns \Wunit \: 2 \to 2$ maps to the function $f \: \{0,1\} \to \{0,1\}$ constant at $0$.

As observed by Lack~\cite{Lack2004a}, PROPs can also be seen as objects of a certain coslice category. To explain this, we need to introduce PROs: these are strict monoidal (i.e.\ not necessarily symmetric) categories with objects the natural numbers and addition as monoidal product. Morphisms of PROs are strict monoidal identity-on-objects functors. The PRO of permutations $\Perm$, where $\Perm[n,m]$ is empty if $n\neq m$ and otherwise consists of permutations on an $n$-element set, is of particular interest. PROPs can be understood as some of the objects of the coslice category $\mathbb{P}/\PRO$, where $\PRO$ is the category of PROs and their morphisms. In fact, PROPs define a full subcategory since morphisms of PROPs are those morphisms of PROs that preserve the permutation structure. Working in the coslice is intuitive: e.g. $\mathbb{P}$ is the initial PROP and to compute $\mathbb{T}_1+\mathbb{T}_2$ in $\PROP$ one identifies the permutations of $\mathbb{T}_1$ and $\mathbb{T}_2$. For SMTs, a useful observation is that if $\mb{T}_1$ is presented by $(\Sigma_1,E_1)$ and $\mb{T}_2$ by $(\Sigma_2,E_2)$, then $\mb{T}_1 + \mb{T}_2$ is presented by $(\Sigma_1 \uplus \Sigma_2, E_1\uplus E_2)$.

\subsection{Composing PROPs}\label{ssec:composingPROP}
The sum $\mb{T}_1 + \mb{T}_2$ is not a typical way of combining theories: more usual is to quotient $\mb{T}_1 + \mb{T}_2$ by equations that express some compatibility between structures in $\mb{T}_1$ and $\mb{T}_2$. This is a standard pattern in algebra: e.g.\ a ring is given by a monoid and an abelian group, subject to equations that ensure that the former distributes over the latter. Similarly, bialgebras and Frobenius algebras describe two different ways of combining a monoid and a comonoid. 

In~\cite{Lack2004a} Lack shows how these phenomena can be understood as arising from the operation of composing PROPs;  we now give a brief account. 
 As shown by Street~\cite{Street_MonadsI}, the theory of monads can be developed in an arbitrary bicategory. Similarly to how small categories are monads in the bicategory of spans in $\Set$---see e.g.~\cite{RosebrRWood_fact}---PROPs are monads on $\mathbb{P}$ in the bicategory $\mathsf{Prof}(\mathbf{Mon})$ of strict monoidal categories and profunctors~\cite{Lack2004a}.
%
PROPs $\T_1$ and $\T_2$ can be composed via a distributive law $\lambda \colon \T_2 \poi \T_1 \to \T_1  \poi \T_2$ between the associated monads, and  $\lambda$ makes $\T \poi  \T_2$ into a PROP whose arrows can be seen as formal \emph{pairs} $n\tr{f \in \T_1}z\tr{g \in \T_2}m$ of an arrow in $\T_1$, then one of $\T_2$.
A key observation is that (the graph of) $\lambda$ gives a set of equations $\tr{\in \T_2}\tr{\in \T_1} = \tr{\in \T_1}\tr{\in \T_2}$. In fact, if $\T_1$ and $\T_2$ are syntactic then $\T_1 \poi  \T_2$ is presented by the equations of $\T_1 +  \T_2$ together with those obtained from $\lambda$. 

For example, composing PROPs $\bcom$ and $\wmon$ of commutative comonoids and monoids yields the PROP of commutative bialgebras. First observe that circuits of $\bcom$ correspond to arrows of $\Fop$, because $\bcom \cong (\wmon {})^{\op} \cong \F^{\op}$. We can then express a distributive law $\lambda \: \wmon \poi \bcom \To \bcom \poi \wmon$ as having the type $\F \poi \Fop \To \Fop \poi \F$.
This amounts to saying that $\lambda$ maps \emph{cospans} $n \tr{f\in \F}\tl{g\in \F} m$ to \emph{spans} $n \tl{p\in \F}\tr{q\in \F} m$.
Defining this mapping via (chosen) pullback
satisfies the conditions of distributive laws~\cite{Lack2004a}. 
One can now read the relevant equations from pullback squares in $\F$. For instance:
\[
\xymatrix@R=5pt@C=15pt{
& 1 & && & & 1 \ar[dr]^{\Bcounit} &\\
\ar[ur]^{} 2 & & 0 \ar[ul]_{} & \ar@{|=>}[r] & & \ar[ur]^{\Wmult} 2 \ar[dr]_{\twoBcounit} & & 0 \\
& \ar[ul]^{} 0 \pullbackcorner \ar[ur]_{} & & && & 0 \ar[ur]_{\lower4pt\hbox{$\includegraphics[height=.5cm]{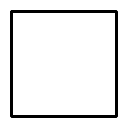}$}} & }
\quad
\lower20pt\hbox{$
\text{yields}
\quad
\Wmult ; \Bcounit = \twoBcounit \poi \lower4pt\hbox{$\includegraphics[height=.5cm]{graffles/idzerocircuit.pdf}$}
$}
\]
where 
the second diagram is obtained from the pullback by applying the isomorphisms  $\F \cong \Mon$ and
$\Fop \cong \Com$. In fact, the equations $\bcom \poi \wmon$ arise from (those of $\bcom + \wmon$ and) just four pullback squares (see \cite[\S 5.3]{Lack2004a}) that yield:
\begin{multicols}{2}
\noindent
\begin{equation}
\label{eq:unitsl}
\lower2pt\hbox{$
\tag{A7}
\lower5pt\hbox{$\includegraphics[height=.6cm]{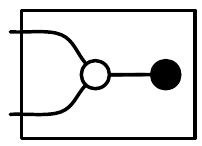}$}
=
\lower5pt\hbox{$\includegraphics[height=.6cm]{graffles/lunitsr.pdf}$}
$}
\end{equation}
\begin{equation}
\label{eq:unitsr}
\tag{A9}
\lower5pt\hbox{$\includegraphics[height=.6cm]{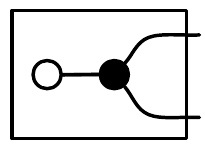}$}
=
\lower5pt\hbox{$\includegraphics[height=.6cm]{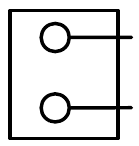}$}
\end{equation}
\begin{equation}
\label{eq:bialg}
\tag{A8}
\lower6pt\hbox{$\includegraphics[height=.6cm]{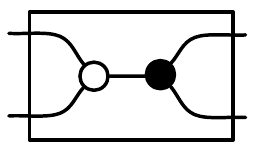}$}
=
\lower11pt\hbox{$\includegraphics[height=.9cm]{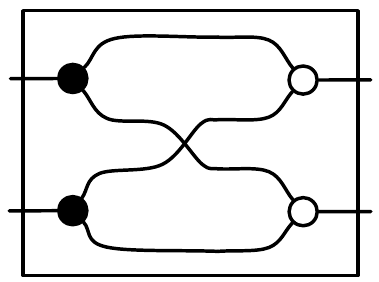}$}
\end{equation}
\begin{equation}
\label{eq:bwbone}
\tag{A10}
\lower4pt\hbox{$\includegraphics[height=.5cm]{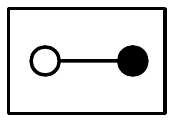}$}
=
\lower4pt\hbox{$\includegraphics[height=.5cm]{graffles/idzerocircuit.pdf}$}
\end{equation}
\end{multicols}
 Therefore $\bcom \poi \wmon$ is the free PROP of commutative bialgebras, obtained as the quotient of $\bcom + \wmon$ by \eqref{eq:unitsl}-\eqref{eq:bwbone}. Furthermore, each circuit $c \: n \to m$ can be factorised as $n\tr{\in \Com}\tr{\in \Mon}m$ and the SMT of commutative bialgebras is a presentation of the PROP $\Span{\F} \cong \Fop \poi \F$ of spans.

There is a dual presentation of $\Cospan{\F}$. The distributive law, of type $\bcom \poi \wmon\To \wmon\poi \bcom$, is defined by pushout in $\F$. Its equations are presented by the PROP of \emph{strongly separable Frobenius algebras}~\cite{Carboni1987}. We refer to~\cite{Lack2004a} for the details. Interestingly, the Frobenius equations also appear in our development (Section~\ref{sec:ibrw}), though for different reasons.

%% file: source/0_ABR.tex
In this section we recall the folklore presentation of the PROP of matrices over a principal ideal domain $\PID$. The resulting theory $\ABR$ of \emph{$\PID$-Hopf algebras} is constructed in a modular fashion, by composing PROPs.
First, let $\PROPR$ be the PROP generated by the signature consisting of \emph{scalars} $\scalar$ for each $k \in \PID$ and the following equations, where $k_1,k_2$ range over $\PID$.
\begin{multicols}{2}\noindent
\begin{equation}
\label{eq:unitscalar}\tag{A11}
\lower6pt\hbox{$\includegraphics[height=.6cm]{graffles/unitscalar.pdf}$}
=
\lower6pt\hbox{$\includegraphics[height=.6cm]{graffles/idcircuit.pdf}$}
\end{equation}
\begin{equation}
\label{eq:scalarmult}
\tag{A12}
\lower6pt\hbox{$\includegraphics[height=.7cm]{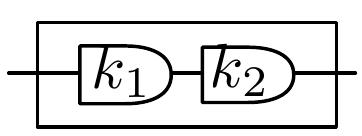}$}
=
\lower6pt\hbox{$\includegraphics[height=.7cm]{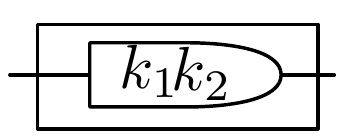}$}
\end{equation}
\end{multicols}
Our building blocks in this section are $\wmon$, $\bcom$ (introduced in Section~\ref{ssec:propsbackground}) and $\PROPR$, which we compose together using distributive laws of PROPs. 

\begin{lemma}~
\label{lemma:threelaws}
\begin{itemize}
  \item There is a distributive law $\sigma \: \wmon \poi \PROPR \To \PROPR \poi \wmon$ yielding a PROP $\PROPR \poi \wmon$ presented by the equations of $\PROPR + \wmon$ and, for all $k \in \PID$:
       \begin{multicols}{2}\noindent
\begin{equation}
\label{eq:scalarwmult}
\tag{A13}
\lower11pt\hbox{$\includegraphics[height=.9cm]{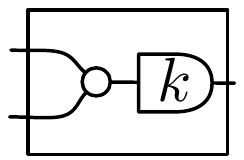}$}
=
\lower11pt\hbox{$\includegraphics[height=.9cm]{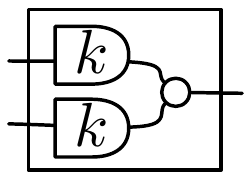}$}
\end{equation}
\begin{equation}
\label{eq:scalarwunit}
\tag{A14}
\lower7pt\hbox{$\includegraphics[height=.7cm]{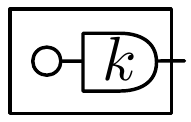}$}
=
\lower7pt\hbox{$\includegraphics[height=.7cm]{graffles/Wunit.pdf}$}
\end{equation}
\end{multicols}
  \item There is a distributive law $\tau \: \PROPR \poi \bcom \To \bcom\poi\PROPR$ yielding a PROP $\bcom\poi\PROPR$ presented by the equations of $\bcom + \PROPR$ and, for all $k \in \PID$:
\begin{multicols}{2}
\noindent
\begin{equation}
\label{eq:scalarbcomult}
\tag{A15}
\lower10pt\hbox{$\includegraphics[height=.9cm]{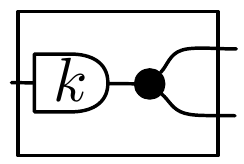}$}
=
\lower10pt\hbox{$\includegraphics[height=.9cm]{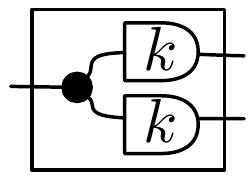}$}
\end{equation}
\begin{equation}
\label{eq:scalarbcounit}
\tag{A16}
\lower7pt\hbox{$\includegraphics[height=.7cm]{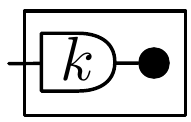}$}
=
\lower7pt\hbox{$\includegraphics[height=.7cm]{graffles/Bcounit.pdf}$}
\end{equation}
\end{multicols}
\end{itemize}
\end{lemma}
\begin{proof} For the first statement, let $\T$ be the PROP obtained through quotienting $\PROPR + \wmon$ by \eqref{eq:scalarwmult} and \eqref{eq:scalarwunit}. Then $\PROPR$ and $\wmon$ are subcategories of $\T$ and equations \eqref{eq:scalarwmult} and \eqref{eq:scalarwunit} yield a representation of each circuit of $\T$ as one of $\PROPR$ followed by one of $\wmon$, which is unique up-to-permutation. This factorisation, by \cite[Th. 4.6]{Lack2004a}, induces the required distributive law of PROPs. The proof of the second statement is similar.
\qed \end{proof}

We now combine the distributive laws of Lemma~\ref{lemma:threelaws} and $\lambda \: \wmon \poi \bcom \To \bcom \poi \wmon$ introduced in Section~\ref{ssec:composingPROP} to build the composite PROP $\bcom \poi \PROPR \poi \wmon$.
\begin{proposition} There is a distributive law $\theta \: \wmon \poi (\bcom \poi \PROPR) \To (\bcom \poi \PROPR) \poi \wmon$ yielding $\bcom \poi \PROPR \poi \wmon$ presented by the equations of $(\PROPR \poi \wmon) + (\bcom\poi\PROPR) + (\bcom \poi \wmon)$.
\end{proposition}
\begin{proof} In \cite{Cheng_IteratedLaws} Cheng shows that the natural transformation $\theta\df \lambda_{\PROPR} \poi \bcom \sigma$  (or, equivalently, the natural transformation $\varphi \df \PROPR \lambda \poi \tau_{\wmon} \: (\PROPR \poi \wmon) \poi \bcom \To \bcom \poi (\PROPR \poi \wmon)$) is a distributive law yielding the monad $\bcom \poi \PROPR \poi \wmon$ provided that the three distributive laws $\lambda$, $\sigma$ and $\tau$ satisfy the Yang-Baxter compatibility condition. This is given by commutativity of the following diagram, which can be easily verified by case analysis on the circuits of $\wmon \poi \PROPR \poi \bcom$.
\[\xymatrix@R=2pt{
& \wmon \poi\bcom \poi\PROPR \ar[r]^{\lambda_\PROPR} & \bcom \poi\wmon \poi\PROPR \ar[dr]^{\bcom \sigma} &\\
\wmon \poi\PROPR \poi\bcom \ar[ur]^{\wmon \tau} \ar[dr]_{\sigma_{\bcom}} & & & \bcom \poi\PROPR \poi\wmon\\
& \PROPR \poi\wmon \poi\bcom \ar[r]^{\PROPR \lambda} & \PROPR \poi \bcom \poi\wmon \ar[ur]_{\tau_{\wmon}} &\\
}\]

As shown in \cite{Cheng_IteratedLaws}, the multiplication for the monad $\bcom \poi \PROPR \poi \wmon$ --- and thus composition in the PROP $\bcom \poi \PROPR \poi \wmon$ --- is equivalently defined by $\theta$ or $\varphi$. This means that the equations holding in $\bcom \poi \PROPR \poi \wmon$ are all those given by the distributive laws composing $\theta$ and $\varphi$, that is, $\lambda$, $\sigma$ and $\tau$. By the presentation of these laws given in Section~\ref{ssec:composingPROP} and Lemma~\ref{lemma:threelaws}, it follows that $\bcom \poi \PROPR \poi \wmon$ can be presented by the equations of $(\PROPR \poi \wmon) + (\bcom\poi\PROPR) + (\bcom \poi \wmon)$.
\qed
\end{proof}

The PROP $\PROPR$ only accounts for the multiplicative part of $\PID$. In order to describe also its additive component, and thus faithfully capture $\PID$-matrices, we need to quotient $\bcom \poi \PROPR \poi \wmon$ by two more equations.

\begin{definition} The PROP $\ABR$ is defined as the quotient of $\bcom \poi \PROPR \poi \wmon$ by the following equations, for all $k_1,k_2 \in \PID$:
\begin{multicols}{2}
\noindent
\begin{equation}
\label{eq:zeroscalar}
\tag{A17}
\lower7pt\hbox{$\includegraphics[height=.7cm]{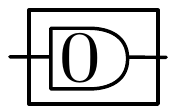}$}
=
\!\!
\lower11pt\hbox{$\includegraphics[height=1cm]{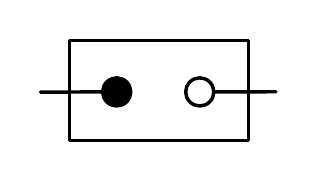}$}
\end{equation}
\begin{equation}
\label{eq:scalarsum}
\tag{A18}
\lower12pt\hbox{$\includegraphics[height=1cm]{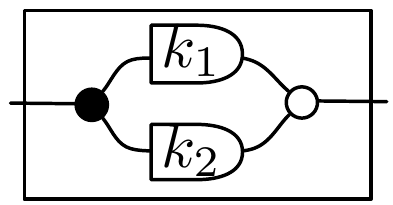}$}
=
\lower8pt\hbox{$\includegraphics[height=.8cm]{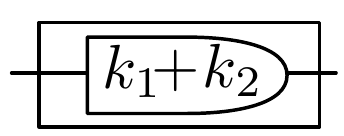}$}
\end{equation}
\end{multicols}
\end{definition}

\begin{remark} The name ``$\PID$-Hopf algebra'' is justified by the principal ideal domain $\Z$. Indeed, as we shall see in Section~\ref{sec:instances}, the presenatation of $\HA{\Z}$  consists of the usual equations of (commutative) Hopf algebras (see e.g.~\cite{Fiore2013,HopfRef1,HopfRef2}). Indeed, $\HA{\Z}$ can be presented by equations~\eqref{eq:wmonassoc}-\eqref{eq:bwbone} and those \eqref{eq:unitscalar}-\eqref{eq:scalarsum} where $k$ ranges over $\{-1,0,1\}$: $\scalarminusone$ is the \emph{antipode},  which we write $\antipode \df\! \scalarminusone$. The well-known \emph{Hopf law} holds in $\ABR$:
  \begin{equation} \label{eq:hopf}
  \tag{Hopf}
  \lower11pt\hbox{$\includegraphics[height=1.1cm]{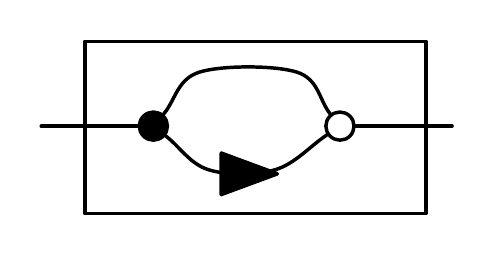}$}
  = \lower10pt\hbox{$\includegraphics[height=1cm]{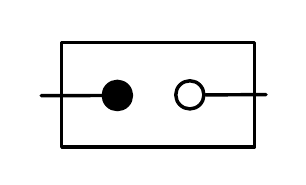}$}
  = \lower11pt\hbox{$\includegraphics[height=1.1cm]{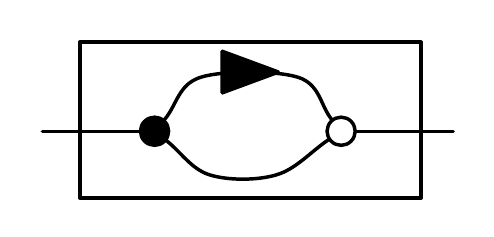}$}.
   \end{equation}
   \end{remark}

Any circuit in $\bcom \poi \PROPR \poi \wmon$, and therefore also any circuit in $\ABR$, can be factorised as $\tr{\in \bcom}\tr{\in \PROPR}\tr{\in \wmon}$. Moreover, by \eqref{eq:zeroscalar}-\eqref{eq:scalarsum}, we can assume that any port on the left $(i)$ has exactly one connection with any port on the right, and by \eqref{eq:scalarmult} and $(ii)$ that any such connection passes through exactly one scalar $\scalar$. In diagrams we will typically omit to draw $1$ scalars, by virtue of~\eqref{eq:unitscalar}, and omit the $0$ scalar by~\eqref{eq:zeroscalar}, leaving the ports in question disconnected. A circuit $b \poi s \poi w$ satisfying $(i)$ and $(ii)$ is said to be in \emph{matrix form} -- in such circuits we say that there is a \emph{$k$-path from $i$ to $j$} if $k$ is the scalar on the path from the $i$th port on the left to the $j$th port on the right, assuming a top-down enumeration. Circuits in matrix form have an obvious representation as $\PID$-matrices, as illustrated below.

\begin{example}\rm\label{ex:matrixform} Consider the circuit $t \in \ABR[3,4]$ (on the right) and its representation as a $4 \times 3$ matrix (on the left).

\begin{minipage}[c]{0.5\textwidth}
  \centering \includegraphics[width=70pt]{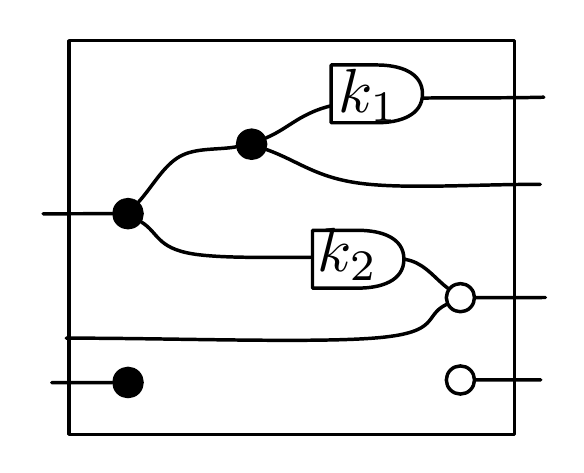}
 \end{minipage}
 \begin{minipage}[c]{0.3\textwidth}
 \centering
  $M = {\scriptsize \left(%
\begin{array}{ccc}
  k_1 & 0 & 0 \\
  1 & 0 & 0 \\
  k_2 & 1 & 0 \\
  0 & 0 & 0
\end{array}\right)}$
 \end{minipage}

\noindent
Note $M_{ij}=k$ exactly when there is a $k$-path from $j$ to $i$.
\end{example}

\noindent We will often write $\circuitAdots$ for the circuit, in matrix form, corresponding to a matrix $A$. We now make the matrix semantics of circuits in $\ABR$ formal: write $\VectR$ for the PROP whose arrows $n \to m$ are $m \times n$-matrices over $\PID$,  where $;$ is matrix multiplication and $A \tns B$ is the matrix $\tiny{\left(%
                \begin{array}{cc}
                 \!\!\! A \!\! & \!\! 0 \!\!\!\\
                 \!\!\! 0 \! \!& \!\! B\!\!\!
                \end{array}\right)}$. The symmetries are permutation matrices. Given matrices $A \: n \to z$, $B \: m \to z$, $C \: r \to n$ and $D \: r \to m$, we write $(A\, | \, B) \: n + m \to z$ and $(\frac{C}{D}) \: z \to n + m$ for the matrices given by universal property of the biproduct $n+m$: the notation reflects the way these matrices are constructed.

\begin{definition}\label{def:sem}
The morphism $\sem{\ABR} \: \ABR \to \VectR$ is defined inductively:
\begin{align*}
                    \Wunit \mapsto\  \initVect &&
                    \Bcounit\mapsto\ \finVect
                     &&
                   \Wmult \mapsto   {\scriptsize\left(%
                \begin{array}{cc}
                \!\!\!  1 \! &\!\! 1 \!\!\!
                \end{array}\right)}
                &&
                 \Bcomult \mapsto\ \tiny{\left(%
                \begin{array}{c}
                 \!\! 1 \!\!\\
                 \!\! 1\!\!
                \end{array}\right)}
                &&
                \scalar \mapsto\ {\scriptsize  \left(%
                \begin{array}{c}
                 \!\!\! k\!\!\!
                \end{array}\right)}
            \end{align*}
            \vspace{-.6cm}
            \begin{align*}
                   s\tns t  \mapsto  \sem{\ABR}(s) \tns\sem{\AB}(t)
                   &&
                  s \poi t \mapsto\ \sem{\ABR}(s) \poi \sem{\AB}(t)
                  &&
                 \end{align*}
                 where $\initVect \: 0 \to 1$ and $\finVect \: 1 \to 0$ are the unique arrows given by univeral properties of $0$ in $\VectR$. 
                 It follows that $\sem{\ABR}$ is well-defined, as it respects the equations of $\ABR$. 
\end{definition}

The following folklore result is of central importance for the original technical developments in this paper.
\begin{proposition}\label{prop:ab=vect} $\sem{\ABR} \: \ABR \to \VectR$ is an isomorphism of PROPs.
\end{proposition}
\begin{proof}
Since $\sem{\ABR}$ is identity-on-objects, it suffices to prove that $\sem{\ABR}$ is full and faithful. Fullness is immediate: given a matrix $M$, it is clear how to generalise the procedure described in Example~\ref{ex:matrixform} in order to obtain a circuit in matrix form that maps via $\sem{\ABR}$ to $M$.
For faithfulness, recall that any circuit of $\ABR$ can be first factorised as $\tr{\in \bcom}\tr{\in \PROPR}\tr{\in \wmon}$ and then put in matrix form. Therefore, it suffices to check that, for $c,d \: n \to m$ in matrix form, $\sem{\ABR}(c) = \sem{\ABR}(d)$ implies $c=d$. This follows by induction on $n$, $m$. 
\qed
\end{proof}
%
%
%

We are  interested in the interaction of $\ABR$ with its opposite $\ABRop$, which we now briefly describe. Circuits of $\ABRop$ are represented as those of $\ABR$ reflected about the $y$-axis, that means, $\ABRop$ is freely obtained by generators $\{\Bunit , \Bmult , \Wcounit , \Wcomult , \coscalar \mid k \in \PID\}$ and equations \eqref{eq:wmonunitlaw}-\eqref{eq:scalarsum} ``in the mirror'', which we indicate with \eqref{eq:wmonunitlaw}$^{\op}$-\eqref{eq:scalarsum}$^{\op}$. The duality between $\ABR$ and $\ABRop$ is witnessed by the obvious contravariant morphism $(\cdot)^{\star} \: \ABR \to \ABRop$.

 The PROP $\ABRop$ is isomorphic to $\VectRop$ via $\sem{\ABR}^{\op} \: \ABRop \to \VectRop$. This means that, since $\sem{\ABR}$ maps $\Bcomult$ to ${\tiny \matrixOneOne} \in \VectR[1,2]$, then $\sem{\ABR}^{\op}$ maps $\Bmult$ to ${\tiny \matrixOneOne} \in \VectRop[2,1]$. Therefore, one should intuitively follow the same procedure of Example~\ref{ex:matrixform} to compute the matrix of a circuit in $\ABRop$, but reading the circuit from right to left --- meaning that columns are ports on the right boundary and rows are ports on the left boundary.
We shall draw \lower4pt\hbox{$\includegraphics[height=.6cm]{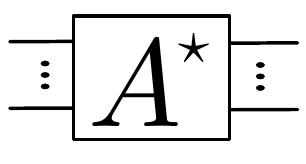}$} for the image under $(\cdot)^{\star}$ of the circuit representation of the matrix $A$. 

%% file: source/distrlawspbpo.tex
The theory of interacting Hopf algebras
is obtained by combining $\ABR$ and $\ABRop$ using the techniques introduced in Section~\ref{ssec:composingPROP}. For this, the original formulation by Lack~\cite{Lack2004a} is too restrictive. Recall that Lack identifies PROPs $\T$ and $\PS$ with monads on $\Perm$ in $\mathsf{Prof}(\mathbf{Mon})$. The 1-cell $\T  \poi \PS$ consists of pairs $\tr{f\in \T}\tr{g\in \PS}$, where $(f,g)$ and $(f',g')$ are identified if they are ``equal up-to permutation'', i.e. if $\exists$ an arrow $\pi$ in $\Perm$ such that $f \poi \pi = f'$ and $\pi \poi g' = g$. This, roughly speaking, 
amounts to identifying the permutations of $\T$ and $\PS$.
 
 This is the case for the distributive law $\F \poi \Fop \To \Fop \poi \F$ defined by pullback in $\F$ (Section~\ref{ssec:composingPROP}). However, this is a particularly fortunate situation: (co)limits in $\F$  are unique up-to permutation, because in $\F$ permutations happen to coincide with the isomorphisms. 
 In general, isomorphisms in an arbitrary PROP $\T$ include but do not coincide with permutations. This prevents us from using pullbacks to obtain a distributive law $\T \poi \Top \To \Top \poi \T$ in the sense of Lack: given $\tr{f \in \T}\tr{g \in \Top}$ and $\tr{f' \in \T}\tr{g' \in \Top}$ equal up-to permutation, pulling back $\tr{f}\tl{g}$ and $\tr{f'}\tl{g'}$ yields isomorphic pairs of arrows, but the witnessing iso is not necessarily a permutation. 

We therefore propose a mild generalisation of Lack's approach that allows us to consider distributive laws by pullback (resp.\ pushout) for an arbitrary PROP $\T$ with pullbacks (resp.\ pushouts). 
For distributive laws involving $\T$ and $\T^{\op}$, we need to identify more structure shared by the two PROPs, namely the sub-PROP $\PJ$ (called the \emph{core} of $\T$) whose arrows are the isomorphisms in $\T$. Formally, this amounts to view PROPs $\T$ and $\T^{\op}$ not as monads on $\Perm$ but rather on $\PJ$. Then, composites $\T \poi \T^{\op}$ and $\T^{\op} \poi \T$ in $\mathsf{Prof}(\mathbf{Mon})$ will identify composable pairs of arrows when they are equal up-to an arrow of $\PJ$, i.e., up-to iso in $\T$. 

\begin{proposition}\label{prop:distrLawPbPo} Let $\T$ be a PROP and $\PJ$ the core of $\T$.
\begin{enumerate}[itemsep=.5ex]
\item $\T$ and $\T^{\op}$ are monads on $\PJ$ in $\mathsf{Prof}(\mathbf{Mon})$.
    \item If $\T$ has pullbacks, there is a distributive law of type $\T \poi \T^{\op} \To \T^{\op} \poi \T$, defined by pullback, yielding the PROP $\T^{\op}\poi \T \cong \Span{\T}$.
    \item If $\T$ has pushouts, there is a distributive law of type $\T^{\op} \poi \T \To \T \poi \T^{\op}$, defined by pushout, yielding the PROP $\T\poi \T^{\op} \cong \Cospan{\T}$.
\end{enumerate}
\end{proposition} 
\begin{proof}It is routine to check commutativity of the relevant diagrams. The reader may consult~\cite[Sec. 2.4.5]{ZanasiThesis} for the details. \qed\end{proof}
\begin{remark} Rosebrugh and Wood~\cite{RosebrRWood_fact} consider distributive laws of categories and investigate distributive laws by pullback and pushout. To do this, they propose to relax the definition of distributive law so that the associated conditions are required to hold up-to an arrow of a fixed groupoid $\PJ$ (in our case, $\PJ$ is the core of $\T$). This yields a bicategory as the result, which can be turned into a category by quotienting hom-sets by equivalence up-to $\PJ$.

This construction does not work for PROPs: differently from categories,
 distributive laws of PROPs need to be well-defined as mappings between equivalence classes of pairs of arrows equal up-to permutation, which as explained above is not guaranteed for the case of pullback and pushout. Our mild generalisation of Lack's approach handles this challenge while staying
within the confines of the standard notion of distributive law.
\end{remark}

%% file: source/1_IBR.tex
In this section we commence the exploration of several theories that arise from
composing $\ABR$ with $\ABRop$, which is the main focus and contribution of this work.
Collectively, we refer to them as \emph{interacting Hopf algebras}.

We first introduce $\IBRw$ --- the superscript $\SupSpan$ represents the fact that $\IBRw$ will be shown to be the theory of \emph{spans} of $\PID$-matrices. In \S~\ref{sec:IBRbCospan} we introduce $\IBRb$, which will be shown to be the theory of \emph{cospans} of $\PID$-matrices.



\begin{definition}\label{def:IBRw} The PROP $\IBRw$ is the quotient of $\ABR + \ABRop$ by the following  equations, where $l$ is any non-zero element and $k$ any element of $\PID$.
\begin{multicols}{2}\noindent
 \begin{equation}
\label{eq:lcm}
\tag{W1}
\lower7pt\hbox{$\includegraphics[height=.7cm]{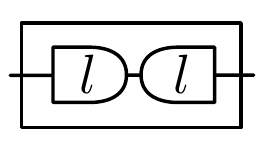}$}
=
\lower6pt\hbox{$\includegraphics[height=.6cm]{graffles/idcircuit.pdf}$}
\end{equation}
\begin{equation}
\label{eq:wbone}
\tag{W2}
\lower4pt\hbox{$\includegraphics[height=.5cm]{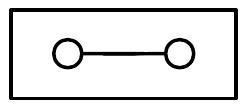}$}
=  \lower4pt\hbox{$\includegraphics[height=.5cm]{graffles/idzerocircuit.pdf}$}
\end{equation}
\end{multicols}
 \begin{multicols}{2}\noindent
\begin{equation}
\label{eq:WFrob}
\tag{W3}
\lower12pt\hbox{$\includegraphics[height=1.2cm]{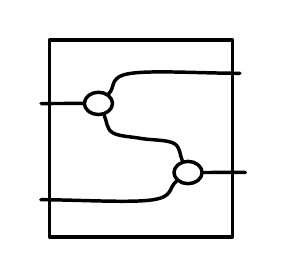}$}
\!\!
=
\!\!
\lower10pt\hbox{$\includegraphics[height=1cm]{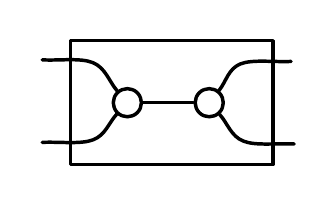}$}
\!\!
=
\!\!
\lower12pt\hbox{$\includegraphics[height=1.2cm]{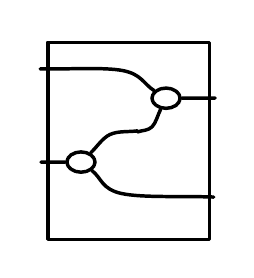}$}
\end{equation}
\begin{equation}
\label{eq:BFrob}
\tag{W4}
\lower12pt\hbox{$\includegraphics[height=1.2cm]{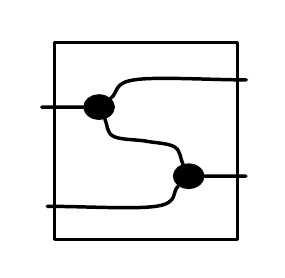}$}
\!\!
=
\!\!
\lower10pt\hbox{$\includegraphics[height=1cm]{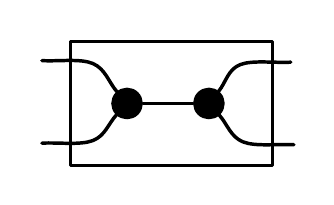}$}
\!\!
=
\!\!
\lower12pt\hbox{$\includegraphics[height=1.2cm]{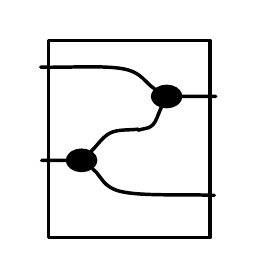}$}
\end{equation}
\end{multicols}
 \begin{multicols}{2}\noindent
\begin{equation}
\label{eq:lcc}
\tag{W5}
\lower12pt\hbox{$\includegraphics[height=1cm]{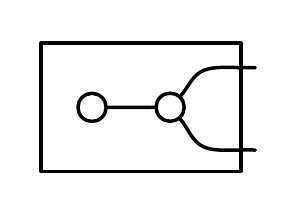}$}
\!\!
=
\!\!
\lower12pt\hbox{$\includegraphics[height=1cm]{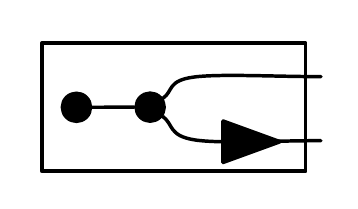}$}
\end{equation}
\begin{equation}
\label{eq:rcc}
\tag{W6}
\lower12pt\hbox{$\includegraphics[height=1cm]{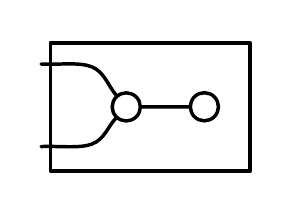}$}
\!\!
=
\!\!
\lower12pt\hbox{$\includegraphics[height=1cm]{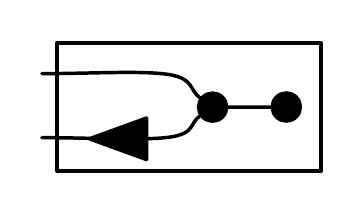}$}
\end{equation}
\end{multicols}
\begin{multicols}{2}\noindent
\begin{equation}
\label{eq:BccscalarAxiomOne}
\tag{W7}
\lower9pt\hbox{$\includegraphics[height=.8cm]{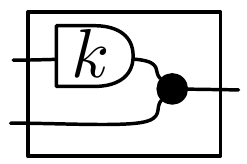}$}
=
\lower9pt\hbox{$\includegraphics[height=.8cm]{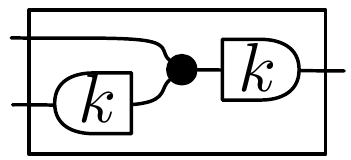}$}
\end{equation}
\begin{equation}
\label{eq:BccscalarAxiomTwo}
\tag{W8}
\lower9pt\hbox{$\includegraphics[height=.8cm]{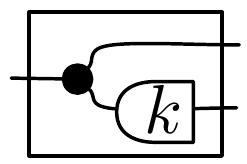}$}
=
\lower9pt\hbox{$\includegraphics[height=.8cm]{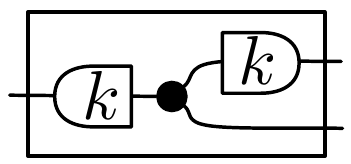}$}
\end{equation}
\end{multicols}
\end{definition}

We fix notation $\sigma_1 \: \ABR \to \IBRw$ and $\sigma_2 \: \ABRop \to \IBRw$ for the PROP morphisms interpreting circuits of $\ABR$ and $\ABRop$, respectively, as circuits of $\IBRw$. Syntactically speaking, the generators of $\ABR$ together with those of $\ABRop$ are also the generators of $\IBRw$ and therefore we will often abuse notation by confusing $c$ in $\ABR$ with $\sigma_1(c)$ in $\IBRw$, and the same for $\ABRop$.

\medskip

The following are some of the derived laws of $\IBRw$, where $k$ is any element and $l$ any non-zero element of $\PID$ (\emph{cf.}~\ref{AppDerLaws}). In~\eqref{eq:wunitcancelbcomult} below and in the sequel, we shall use the shorthand notation $\lower7pt\hbox{$\includegraphics[height=18pt]{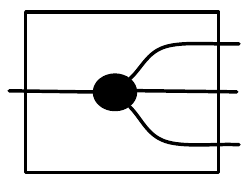}$}$ for the comultiplication $\lower7pt\hbox{$\includegraphics[height=18pt]{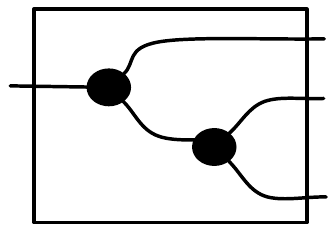}$}$ from $1$ to $3$, and more generally $\lower7pt\hbox{$\includegraphics[height=18pt]{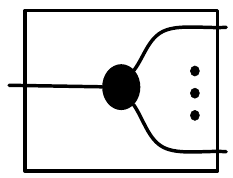}$}$ for the one from $1$ to an arbitrary $n$. This convention is harmless by~\eqref{eq:bcomonassoc}. We will adopt an analogous notation for multiplications $\Wmult$ of arity bigger than $2$.

\begin{multicols}{2}\noindent
\begin{equation}
\label{eq:lccb}
\tag{D1}
\lower7pt\hbox{$\includegraphics[height=.7cm]{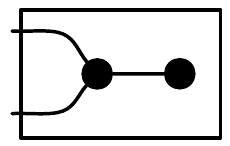}$}
\!
=
\!
\lower7pt\hbox{$\includegraphics[height=.7cm]{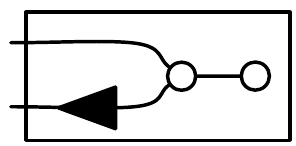}$}
\end{equation}
\begin{equation}
\label{eq:rccb}
\tag{D2}
\lower7pt\hbox{$\includegraphics[height=.7cm]{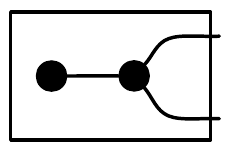}$}
\!
=
\!
\lower7pt\hbox{$\includegraphics[height=.7cm]{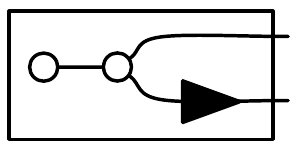}$}
\end{equation}
\end{multicols}
\begin{multicols}{2}\noindent
\begin{equation}
\label{eq:uniqueantipode}
\lower3pt\hbox{$
\tag{D3}
\lower5.5pt\hbox{$\includegraphics[height=.55cm]{graffles/antipode.pdf}$}
=
\lower5.5pt\hbox{$\includegraphics[height=.55cm]{graffles/antipodeop.pdf}$}
$}
\end{equation}
\begin{equation}\label{eq:QFrob}
\tag{D4}
\lower12pt\hbox{$\includegraphics[height=.9cm]{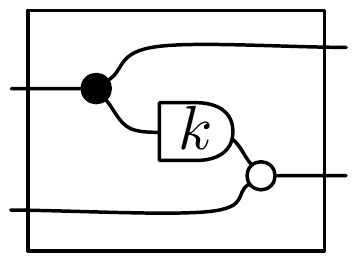}$} =
\lower12pt\hbox{$\includegraphics[height=.9cm]{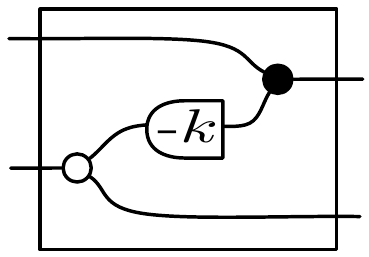}$}
\end{equation}
\end{multicols}
\begin{multicols}{2}\noindent
\begin{equation}
\label{eq:coscalarwunit}
\tag{D5}
\lower8pt\hbox{$\includegraphics[height=.7cm]{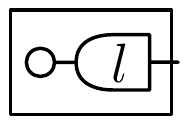}$}
=
\lower6pt\hbox{$\includegraphics[height=.6cm]{graffles/Wunit.pdf}$}
\end{equation}
\begin{equation}
\label{eq:scalarwcounit}
\tag{D6}
\lower7pt\hbox{$\includegraphics[height=.7cm]{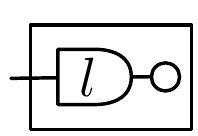}$}
 =
\lower6pt\hbox{$\includegraphics[height=.6cm]{graffles/Wcounit.pdf}$}
\end{equation}
\end{multicols}
\begin{multicols}{2}\noindent
\begin{align}
\label{eq:coscalarbcomult}
\tag{D7}
\lower10pt\hbox{$\includegraphics[height=.9cm]{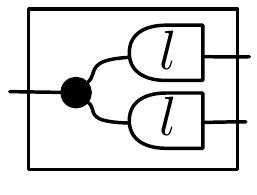}$}
=
\lower8pt\hbox{$\includegraphics[height=.8cm]{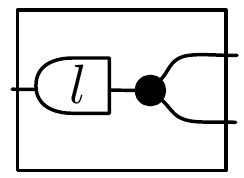}$}
\end{align}
\begin{equation}
\label{eq:scalarbmult}
\tag{D8}
\lower8pt\hbox{$\includegraphics[height=.8cm]{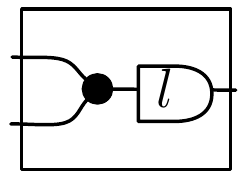}$}
=
\lower7.5pt\hbox{$\includegraphics[height=.9cm]{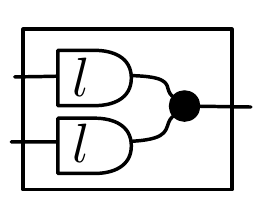}$}
\end{equation}
\end{multicols}
\noindent
\begin{equation}\label{eq:papillon}
\tag{D9}
\lower12pt\hbox{$\includegraphics[width=2.6cm]{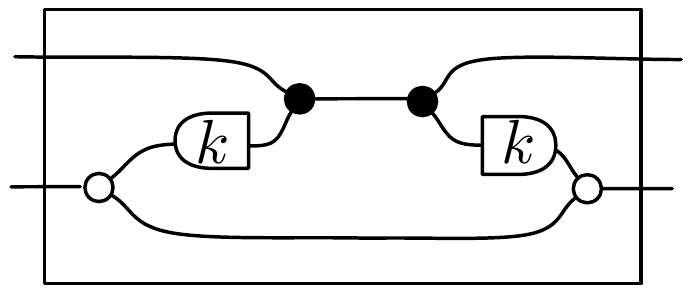}$}
=
\lower12pt\hbox{$\includegraphics[height=1.1cm]{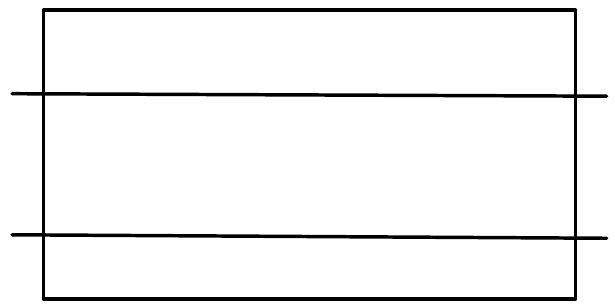}$}
 =
\lower12pt\hbox{$\includegraphics[width=2.4cm]{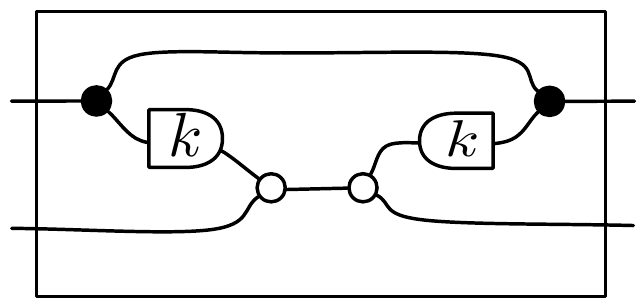}$}
\end{equation}
\vspace{-.9cm}\linebreak
\begin{multicols}{2}\noindent
\begin{align}
\label{eq:wunitcancelbcomult}
\tag{D10}
\lower9pt\hbox{$\includegraphics[height=.9cm]{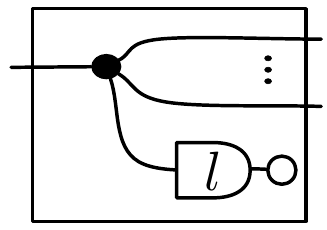}$}
=
\lower7pt\hbox{$\includegraphics[height=.7cm]{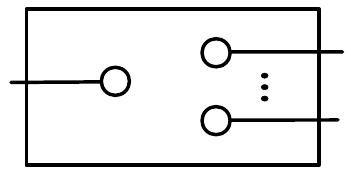}$}
\end{align}
\begin{align}
\label{eq:Bsep}
\tag{D11}
\lower6pt\hbox{$\includegraphics[height=.7cm]{graffles/BSep.pdf}$}
=
\lower6pt\hbox{$\includegraphics[height=.7cm]{graffles/idcircuit.pdf}$}
\end{align}
\end{multicols}
Equation~\eqref{eq:uniqueantipode} states that the antipodes of $\ABR$ and $\ABRop$ coincide in $\IBRw$, which allows us to use the same notation $\antipodesquare$ for the two of them. Also observe that, because of \eqref{eq:BFrob} and \eqref{eq:Bsep}, the black structure in $\IBRw$ forms a separable Frobenius algebra~\cite{Carboni1987}. The white structure, by \eqref{eq:WFrob}, also forms a Frobenius algebra that however is not separable, that is, the equation $\WSep\!\! = \idcircuit$ is not present. The situation is dual (separability for the white but not for the black structure) for $\IBRb$ investigated in Section~\ref{sec:IBRbCospan}.

%% file: source/2_CompactClosed.tex

The PROP $\IBRw$ enjoys a self-dual compact closed structure~\cite{kelly1980compactclosed}:
we associate $n$ with circuits $\eta_n \: 0 \to n + n$ and $\epsilon_n \: n+n \to 0$ defined by induction as follows:
\begin{multicols}{2}
\noindent
 \begin{eqnarray*}
  \hspace{-2cm}
   \lower18pt\hbox{$
   \alpha_0 \: 2 \to 2 \df  
 \lower5pt\hbox{$\includegraphics[height=.6cm]{graffles/symmetryalt.pdf}$}
   $}
 \end{eqnarray*}
  \begin{eqnarray*}
   \hspace{-2cm}
   \alpha_{n+1} \: 2(n+1) \to 2(n+1) \df \lower22pt\hbox{$\includegraphics[height=1.8cm]{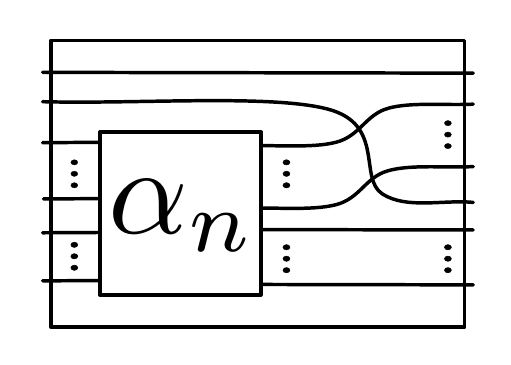}$}
 \end{eqnarray*}
 \end{multicols}
 \begin{multicols}{2}
\noindent
  \begin{eqnarray*}
  \hspace{-1.6cm}
       \lower12pt\hbox{$
   \eta_0 \: 0 \to 0 \df \lower4pt\hbox{$\includegraphics[height=.5cm]{graffles/idzerocircuit.pdf}$} 
   $}
   \end{eqnarray*}
  \begin{eqnarray*}
  \hspace{-2cm}
   \eta_{n+1} \: 0 \to 2(n+1)   \df  \lower22pt\hbox{$\includegraphics[height=1.6cm]{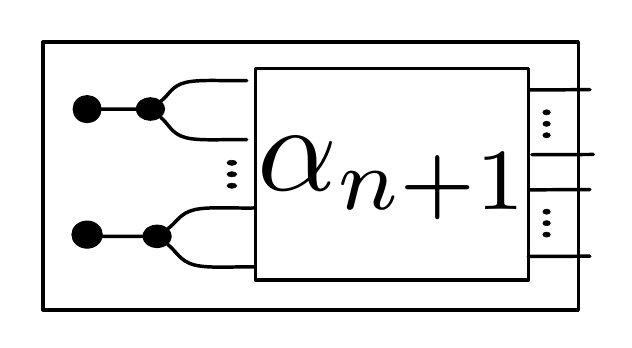}$}
   \end{eqnarray*}
   \end{multicols}
 \begin{multicols}{2}
\noindent
  \begin{eqnarray*}
  \hspace{-2cm}
     \lower12pt\hbox{$
   \beta_0 \: 2 \to 2  \df  
    \lower5pt\hbox{$\includegraphics[height=.6cm]{graffles/symmetryalt.pdf}$}
   $}
    \end{eqnarray*}
  \begin{eqnarray*}
  \hspace{-2cm}
   \beta_{n+1} \: 2(n+1) \to 2(n+1)  \df \lower22pt\hbox{$\includegraphics[height=1.6cm]{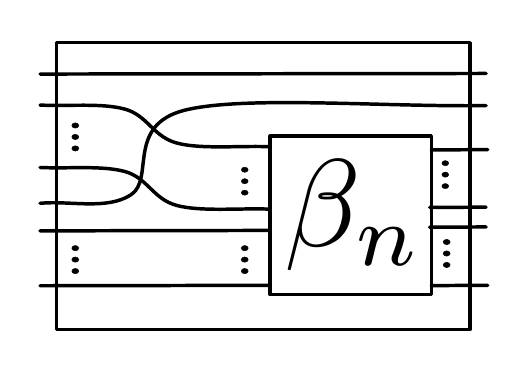}$} \end{eqnarray*}
   \end{multicols}
 \begin{multicols}{2}
\noindent
  \begin{eqnarray*}
  \hspace{-1.5cm}
       \lower18pt\hbox{$
   \epsilon_0 \: 0 \to 0   \df \lower4pt\hbox{$\includegraphics[height=.5cm]{graffles/idzerocircuit.pdf}$}
   $}
    \end{eqnarray*}
  \begin{eqnarray*}
  \hspace{-2cm}
   \epsilon_{n+1} \: 2(n+1) \to 0   \df  \lower22pt\hbox{$\includegraphics[height=1.8cm]{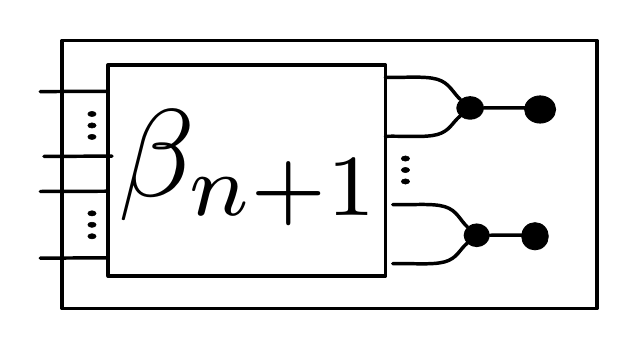}$}
   \end{eqnarray*}
 \end{multicols}

\noindent The first three instances of $\eta_n$ are:
  \begin{multicols}{3}
\noindent
  \begin{eqnarray*}
       \lower20pt\hbox{$
   \eta_1 = \lower12pt\hbox{$\includegraphics[height=1.1cm]{graffles/lccl.pdf}$}
   $}
    \end{eqnarray*}
  \begin{eqnarray*}
   \eta_2 =  \lower22pt\hbox{$\includegraphics[height=1.8cm]{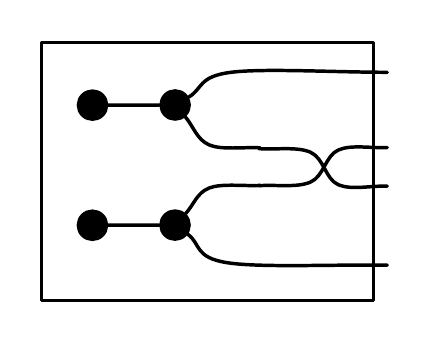}$}
   \end{eqnarray*}
     \begin{eqnarray*}
   \eta_3 =  \lower22pt\hbox{$\includegraphics[height=1.8cm]{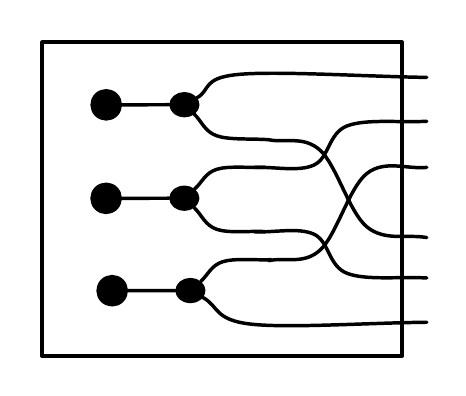}$}
   \end{eqnarray*}
 \end{multicols}

 We will often write \lccn for $\eta_n$, \rccn for $\epsilon_n$ and \idncircuit for $\id_n$. Similarly, \wcounitn (resp.\ \bcounitn) denotes the $n$-fold monoidal product of \Wcounit\ (resp.\ \Bcounit).

\begin{proposition}\label{prop:snakecc} $\IBRw$ is self-dual compact closed with structure given by $\eta_n$ and $\epsilon_n$ for each $n \in \IBRw$. \end{proposition}
\begin{proof} It suffices to verify the following equality, for each $n \in \IBRw$.
  \begin{equation}
  \label{eq:gensnake}
  \tag{CC1}
  \lower18pt\hbox{$\includegraphics[height=1.4cm]{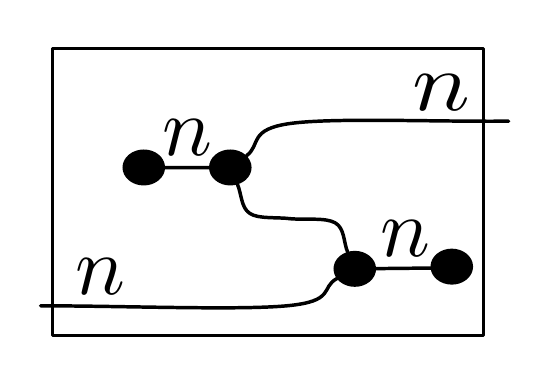}$}
  = \lower12pt\hbox{$\includegraphics[height=1cm]{graffles/idncircuit.pdf}$}
  = \lower18pt\hbox{$\includegraphics[height=1.4cm]{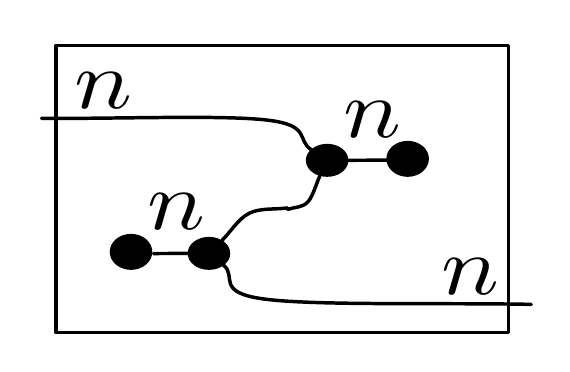}$}
   \end{equation}
   The details of this derivation in $\IBRw$ can be found in~\ref{AppCC}.
\qed\end{proof}

As observed in \cite[Remark~2.1]{Selinger07DaggerCC}, we can define a contravariant PROP morphism $\coc{(\cdot)}$ as follows:
\begin{align*}
 \lower12pt\hbox{$\includegraphics[height=1.2cm]{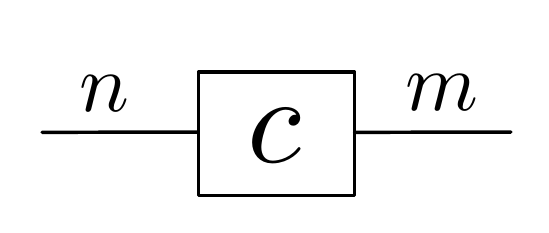}$} \mapsto \lower12pt\hbox{$\includegraphics[height=1.2cm]{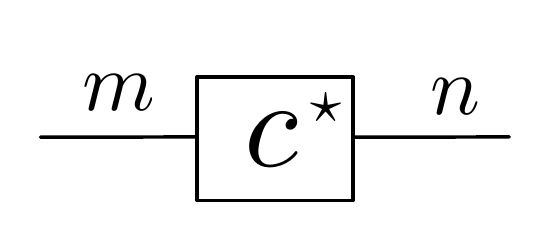}$} \df \lower18pt\hbox{$\includegraphics[height=1.7cm]{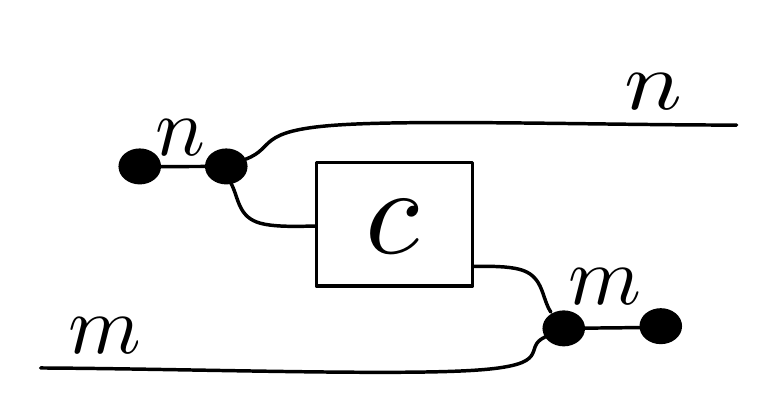}$}
 \end{align*}

 \begin{corollary} For any circuit $c \: n \to m$ of $\IBRw$,
 \begin{multicols}{2}\noindent
  \begin{equation}
  \hspace{-.1cm}
  \label{eq:ccsliding}
  \tag{CC2}
  \lower10pt\hbox{$\includegraphics[height=1.1cm]{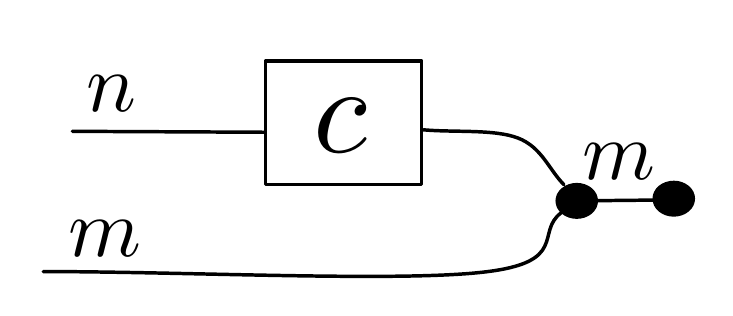}$}
   \! =  \! \lower13pt\hbox{$\includegraphics[height=1.1cm]{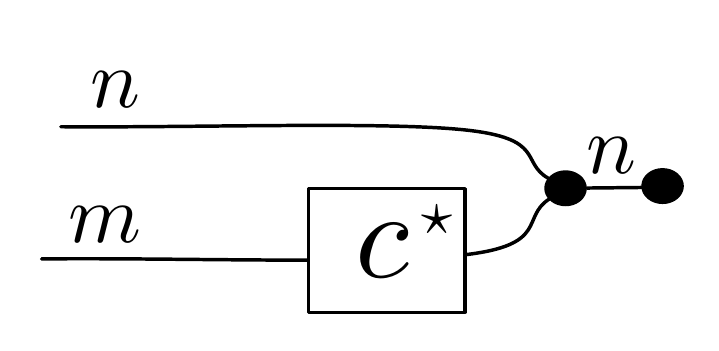}$}
   \end{equation}
     \begin{equation}
  \label{eq:ccsliding2}
  \tag{CC3}
  \lower8pt\hbox{$\includegraphics[height=1.1cm]{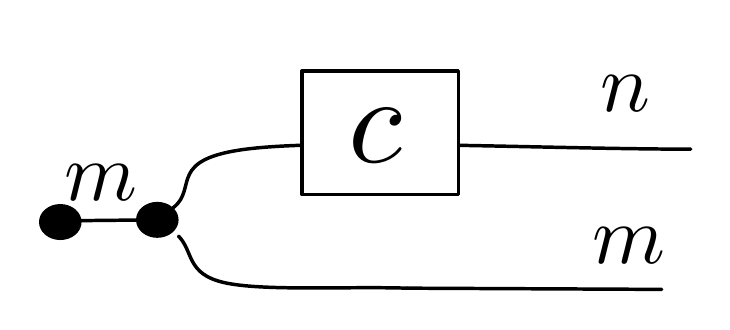}$}
  \! = \! \lower10pt\hbox{$\includegraphics[height=1.1cm]{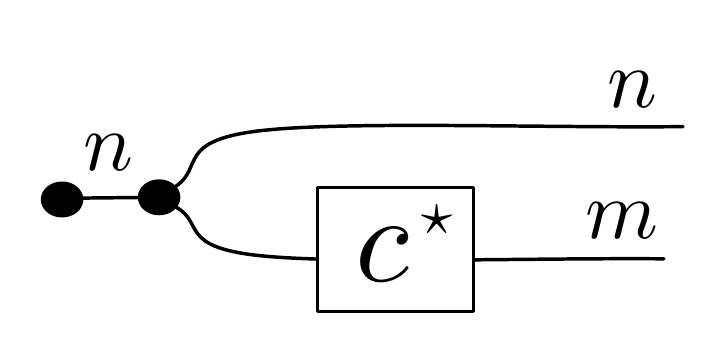}$}
   \end{equation}
   \end{multicols}
 \end{corollary}
 \begin{proof} The following is the derivation of \eqref{eq:ccsliding} in $\IBRw$. The one of \eqref{eq:ccsliding2} is analogous.
  \begin{equation*}
  \lower12pt\hbox{$\includegraphics[height=1.3cm]{graffles/ccslidingr.pdf}$}
  \eql{Def. $c^{\star}$} \lower22pt\hbox{$\includegraphics[height=1.7cm]{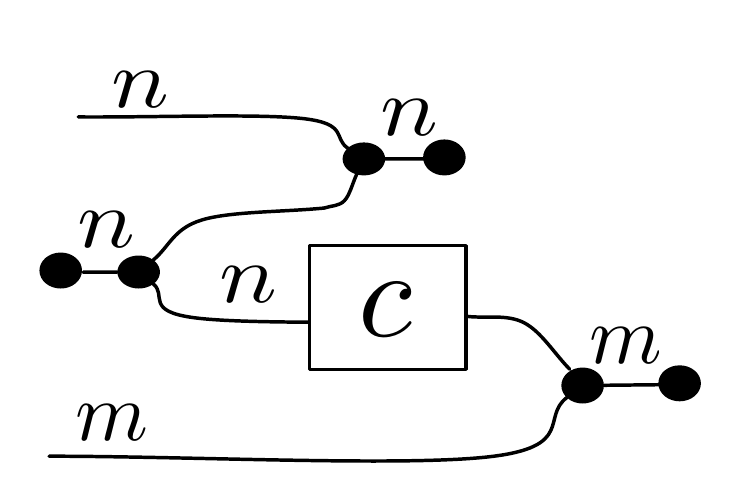}$}
  \eql{\eqref{eq:gensnake}} \lower12pt\hbox{$\includegraphics[height=1.3cm]{graffles/ccslidingl.pdf}$}
   \end{equation*}
 \qed\end{proof}

The following proposition ensures that the notation $\coc{(\cdot)}$ used above actually does not conflict with the one used for the contravariant identity $\ABR \to \ABRop$ defined in Section~\ref{sec:theorymatr}, in the sense that $\coc{\sigma_1(c)} = \sigma_1(\coc{c})$. First, let $\refl{(\cdot)} \: \IBRw \to \IBRw$ be the contravariant PROP morphism given inductively:
  \begin{multicols}{4}
\noindent
      \begin{eqnarray*}
     \Bcounit \mapsto \Bunit
    \end{eqnarray*}
   \begin{eqnarray*}
     \Bunit \mapsto \Bcounit
    \end{eqnarray*}
  \begin{eqnarray*}
     \Wunit \mapsto \Wcounit
    \end{eqnarray*}
  \begin{eqnarray*}
   \Wcounit \mapsto \Wunit
   \end{eqnarray*}
    \end{multicols}
    \smallskip
     \begin{multicols}{4}
     \noindent
     \begin{eqnarray*}
     \Wmult \mapsto \Wcomult
    \end{eqnarray*}
  \begin{eqnarray*}
   \Wcomult \mapsto \Wmult
   \end{eqnarray*}
        \begin{eqnarray*}
     \Bmult \mapsto \Bcomult
    \end{eqnarray*}
  \begin{eqnarray*}
   \Bcomult \mapsto \Bmult
   \end{eqnarray*}
      \end{multicols}
      \smallskip
     \begin{multicols}{2}
     \noindent
           \begin{eqnarray*}
     \scalar \mapsto \coscalar
    \end{eqnarray*}
  \begin{eqnarray*}
   \coscalar \mapsto \scalar
   \end{eqnarray*}
    \end{multicols}
    \smallskip
     \begin{multicols}{2}
     \noindent
     \begin{equation*}    \lower10pt\hbox{$\includegraphics[height=.9cm]{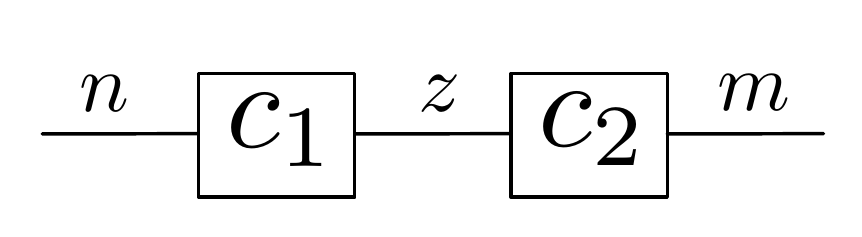}$}
  \!\mapsto\! \lower10pt\hbox{$\includegraphics[height=.9cm]{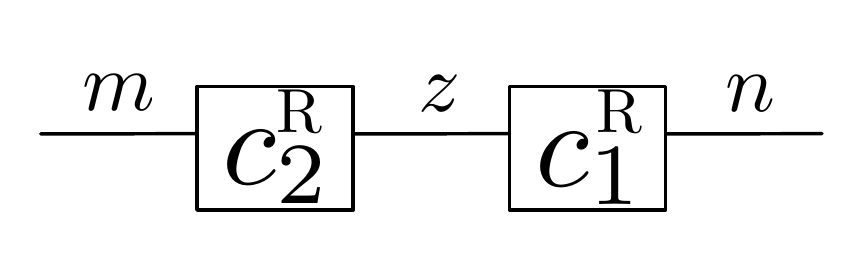}$}
   \end{equation*}
        \begin{equation*}
         \lower20pt\hbox{$\includegraphics[height=1.5cm]{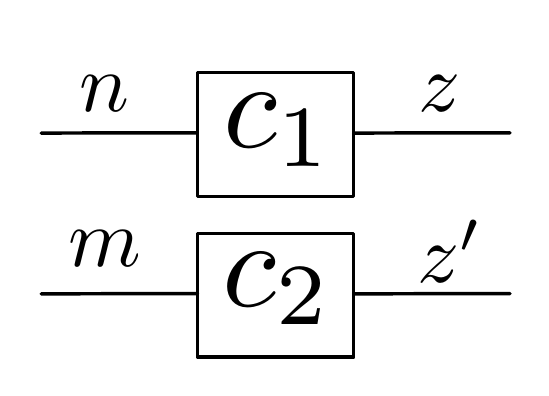}$}
   \!\mapsto\!\lower20pt\hbox{$\includegraphics[height=1.5cm]{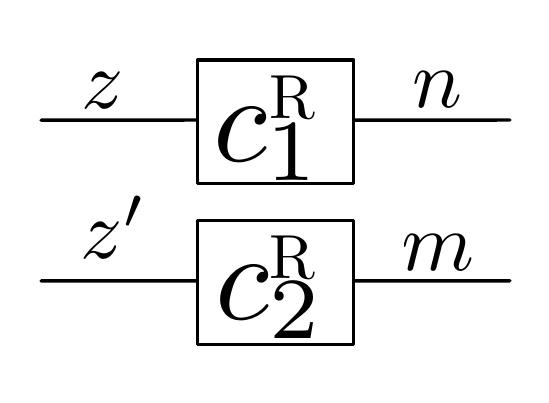}$}
   \end{equation*}
 \end{multicols}

\begin{proposition}\label{prop:star=refl} $\coc{c} = \refl{c}$ for all circuits $c \: n \to m$ of $\IBRw$.
\end{proposition}
\begin{proof} The proof is by induction on $c$. See~\ref{AppCC} for the details. \qed\end{proof}

%


%% file: source/25_KernelLemma.tex
Since $\PID$ is a principal ideal domain, every submodule of a free module is itself free (see e.g.~\cite[Ch. 23]{HandbookLinearAlgebra}): this means that pullbacks in the category of finite-dimensional free $\PID$-modules---which is equivalent to $\VectR$---can be computed as in the abelian category of $\PID$-modules.
Given that $\VectR$ has pullbacks, we can consider the PROP $\Span{\VectR}$. We now develop the tools necessary to show that $\IBRw$ is a presentation of $\Span{\VectR}$.

\begin{theorem}\label{th:Span=IBw} $\IBRw \cong \Span{\VectR}$. \end{theorem}

Our proof relies on the properties of composed PROPs. 
Seeing that $\ABR\cong\VectR$ (Proposition~\ref{prop:ab=vect}) and the fact that $\VectR$ has pullbacks
we form the PROP $\Span{\ABR} = \ABRop \poi \ABR$ via a distributive law $$\lambda_{pb} \: \ABR \poi \ABRop \to \ABRop \poi \ABR$$ which maps a cospan $\tr{\in \ABR}\tl{\in \ABR}$ to its pullback span $\tl{\in \ABR}\tr{\in \ABR}$ ---  $\lambda_{pb}$ is a distributive law by Proposition~\ref{prop:distrLawPbPo}
(\emph{cf.}\ the SMT of bialgebras of Section~\ref{ssec:composingPROP}).
Also, by Proposition~\ref{prop:ab=vect} we have that $\Span{\VectR} \cong \Span{\ABR}$.
Therefore, in order to prove Theorem~\ref{th:Span=IBw} it suffices to show that all equations of $\IBRw$ are derivable in $\Span{\ABR}$ (\emph{soundness}) and vice-versa (\emph{completeness}).

To show soundness, observe that the equations of $\IBRw$ are of two kinds: those of $\ABR + \ABRop$, which are also valid in $\Span{\ABR}$ by construction, and equations~\eqref{eq:lcm}-\eqref{eq:BccscalarAxiomTwo}. For the latter, each has the shape $p \poi \coc{q} = \coc{f} \poi g$, where $p,q,f,g$ are circuits of $\ABR$, and $(f,g)$ is the pullback of $(p,q)$ in $\ABR$.
\begin{example} Equation~\eqref{eq:lcc} corresponds to the pullback square in $\ABR$ on the left below. The pullback square in $\VectR$ is illustrated on the right.
\begin{eqnarray*}
\xymatrix@R=12pt@C=15pt{
& \ar[dl]_{\Bcounit} 1 \pushoutcorner \ar[dr]^{\BcomultSingleAntipode} & \\
0 \ar[dr]_{\Wunit} && 2 \ar[dl]^{\Wmult} \\
& 1 &
}
& \qquad \xymatrix{\\ \ar@{|->}[r]^{\sem{\ABR}} &}\qquad &
\xymatrix@R=12pt@C=15pt{
& \ar[dl]_{\finVect} 1 \pushoutcorner \ar[dr]^{\tiny{\left(%
               \begin{array}{c}
                 \!\! 1 \!\!\\
                 \!\! -1\!\!
                \end{array}\right)}} & \\
0 \ar[dr]_{\initVect} && 2 \ar[dl]^{\tiny{\left(%
                 \begin{array}{cc}
                \!\!  1 \!\! &\! 1 \!\!
                \end{array}\right)}} \\
& 1 &
}
\end{eqnarray*}
\end{example}

It remains to show completeness:
 we need to verify that any pullback in $\ABR$ (or, equivalently, in $\VectR$) yields an equation which is derivable in $\IBRw$. The proof of Theorem~\ref{th:Span=IBw} thus reduces to the proof of the following.

\begin{proposition} \label{prop:IBwComplete} Given a pullback square in $\VectR$ (below left), the corresponding circuit equation (below right) is derivable in $\IBRw$.
\[\xymatrix@R=10pt@C=10pt{
&\ar[dl]_{C} r \pushoutcorner \ar[dr]^{D} & \\
n \ar[dr]_{A} & & m \ar[dl]^{B}\\
& z  & }
\qquad \qquad
\lower20pt\hbox{
\lower8pt\hbox{$\includegraphics[height=.8cm]{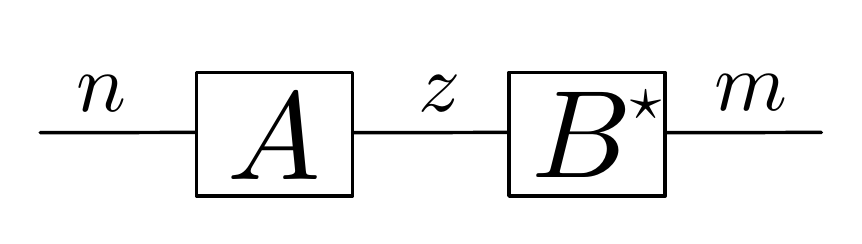}$} =
\lower8pt\hbox{$\includegraphics[height=.8cm]{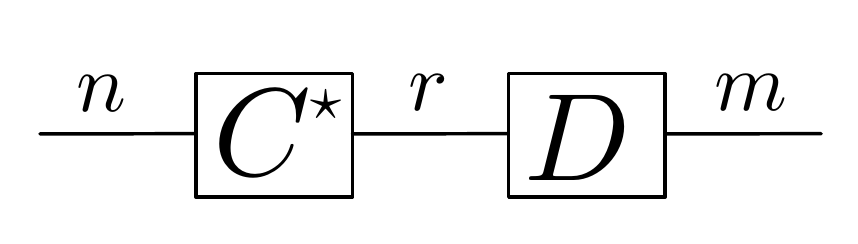}$}
}
\]
\end{proposition}

%% file: source/4_InvMatrices.tex
In order to prove Proposition~\ref{prop:IBwComplete} it is useful to first develop a string diagrammatic treatment of invertible matrices.



\begin{lemma}\label{lemma:invertiblestar} For $U \in \VectR[n,n]$ invertible, the following holds in $\IBRw$:
\begin{equation}\label{eq:invertiblestar}
\lower11pt\hbox{$\includegraphics[height=1cm]{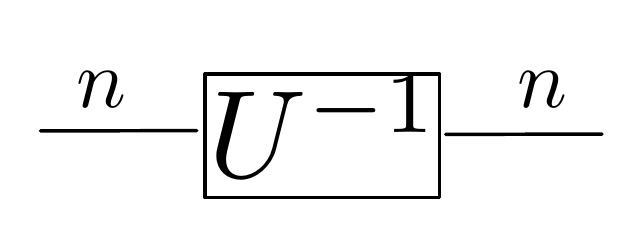}$} =
\lower11pt\hbox{$\includegraphics[height=1cm]{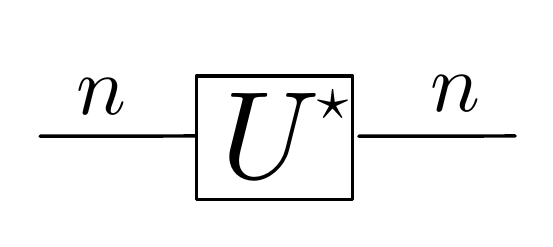}$}
\end{equation}
\end{lemma}
\begin{proof} Recall that an invertible $n \times n$ $\PID$-matrix is one obtainable from the identity $n \times n$ matrix by application of elementary row operations. Thus we can prove our statement by induction on the number of applied operations.

The base case is the one in which no row operation is applied and thus $U = \id_n$. Then we have the following equality in $\IBRw$, yielding~\eqref{eq:invertiblestar}.
 \[
\lower11pt\hbox{$\includegraphics[height=1cm]{graffles/circuitUminusone.pdf}$} =
\lower4pt\hbox{$\includegraphics[height=.8cm]{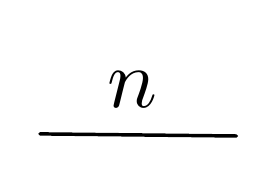}$} =
\lower11pt\hbox{$\includegraphics[height=1cm]{graffles/circuitUstar.pdf}$}
\]

Inductively, suppose that $U$ is obtained by swapping two rows of an invertible matrix $V$. We can assume without loss of generality that the two rows are one immediately above the other, with $j$ the number of rows above them and $m$ the number of rows below, where $n = j+2+m$. In circuit terms, this means that
\[
\lower11pt\hbox{$\includegraphics[height=1cm]{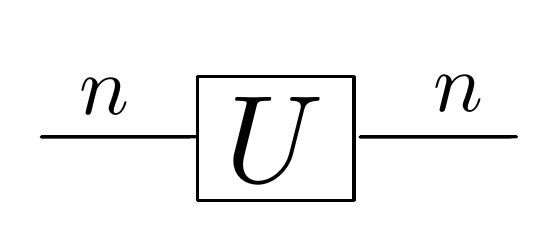}$} =
\lower11pt\hbox{$\includegraphics[height=1cm]{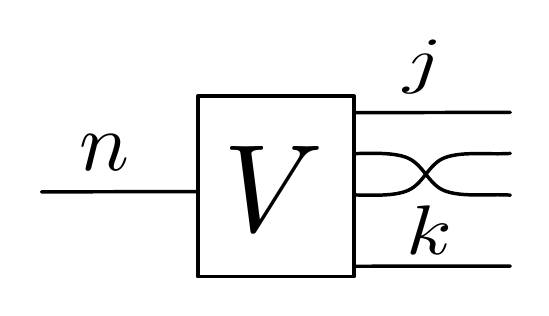}$}
\]
In order to show~\eqref{eq:invertiblestar}, it suffices to prove that the circuit representing $U^{\star}$ is the inverse of $U$, that is, $U\poi U^{\star} = \id_n = U^{\star}\poi U$. This is given by the following derivations.
\begin{eqnarray*}
\lower14pt\hbox{$\includegraphics[height=1.3cm]{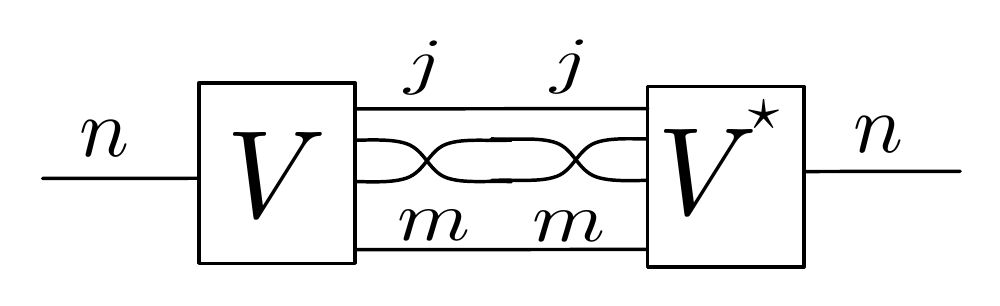}$}
\eql{Axiom SMCs}
\lower11pt\hbox{$\includegraphics[height=1cm]{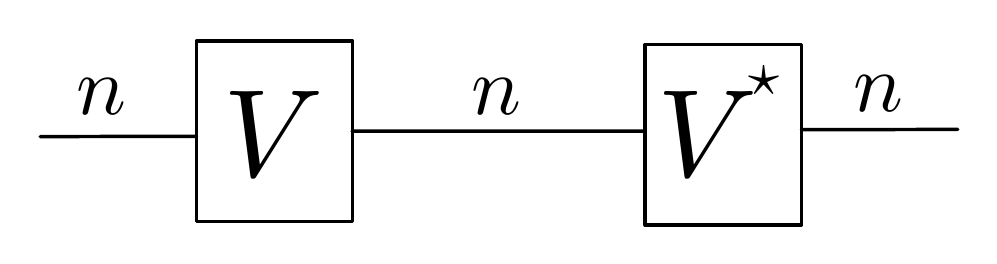}$}
\eql{Ind. hyp.}
\lower4pt\hbox{$\includegraphics[height=.8cm]{graffles/circuitidn.pdf}$} \\
\lower14pt\hbox{$\includegraphics[height=1.3cm]{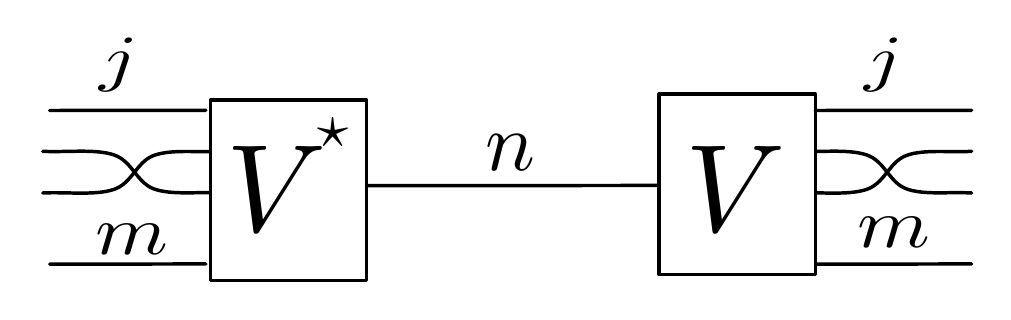}$}
\eql{Ind. hyp.}
\lower12pt\hbox{$\includegraphics[height=1.2cm]{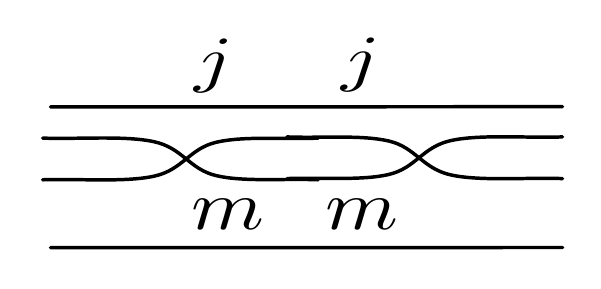}$}
\eql{Axiom SMCs}
\lower4pt\hbox{$\includegraphics[height=.8cm]{graffles/circuitidn.pdf}$}
 \end{eqnarray*}

 The next inductive case that we consider is the one of row sum. As above, we may assume that such operation is applied to adjacent rows of an invertible matrix $V$. The circuit representing $U$ has the following shape, where $j+2+m = n$:
\[
\lower11pt\hbox{$\includegraphics[height=1cm]{graffles/circuitU.pdf}$} =
\lower18pt\hbox{$\includegraphics[height=1.5cm]{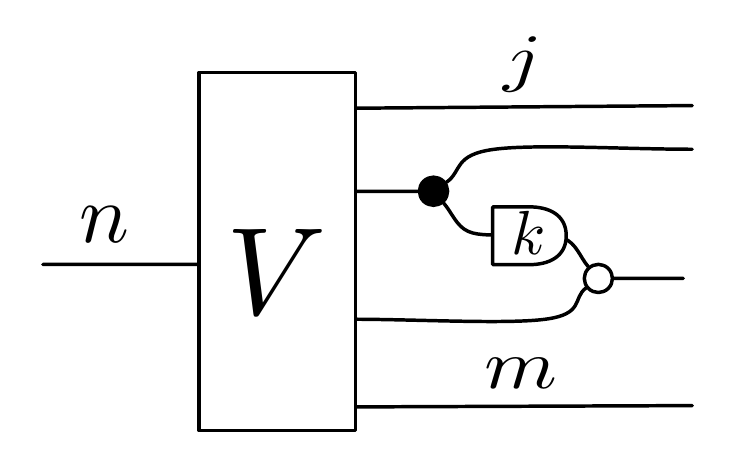}$}
\]
The following two derivations prove that $U^{\star}$ is the inverse of $U$:
\begin{gather*}
\lower18pt\hbox{$\includegraphics[height=1.9cm]{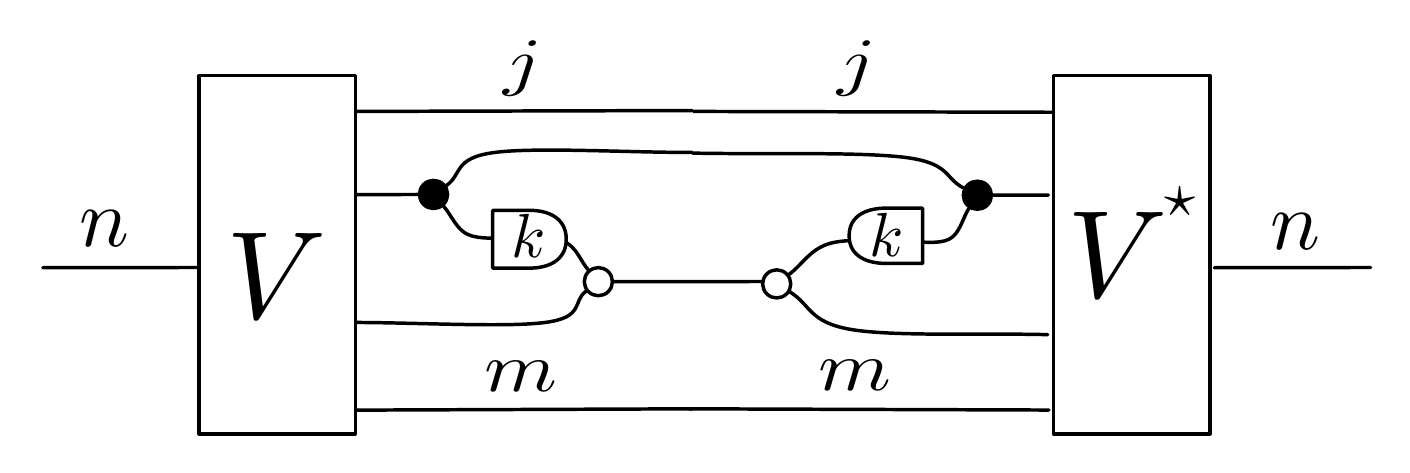}$}
\eql{\eqref{eq:papillon}}
\lower11pt\hbox{$\includegraphics[height=1.2cm]{graffles/circuitVVstar.pdf}$}
\eql{Ind. hyp.}
\lower4pt\hbox{$\includegraphics[height=.8cm]{graffles/circuitidn.pdf}$} \\
\lower18pt\hbox{$\includegraphics[height=2.1cm]{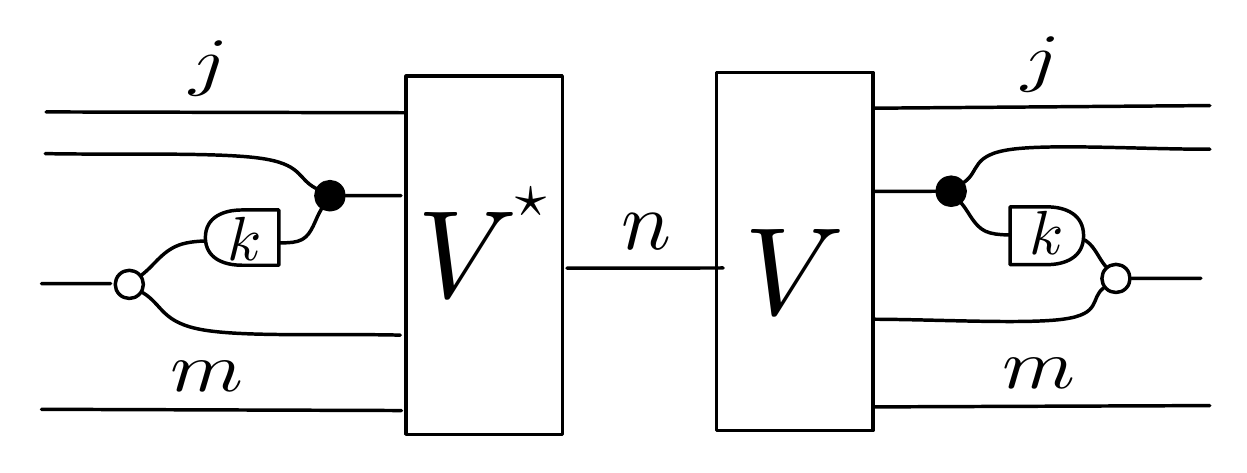}$}
\eql{Ind. hyp.}
\lower16pt\hbox{$\includegraphics[height=2cm]{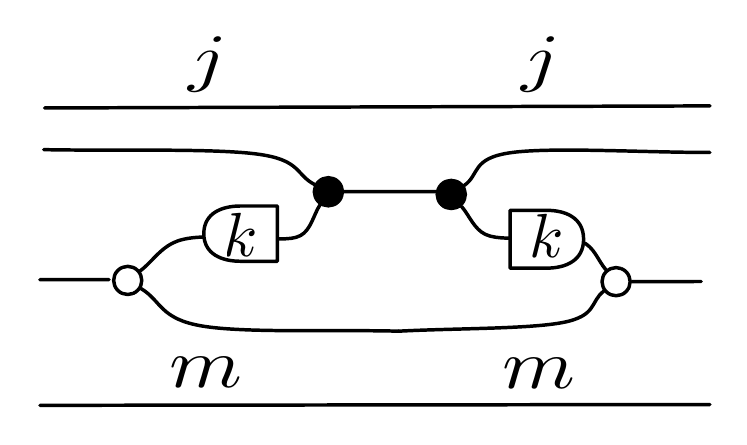}$}
\eql{\eqref{eq:papillon}}
\lower4pt\hbox{$\includegraphics[height=.8cm]{graffles/circuitidn.pdf}$}
 \end{gather*}
 Finally, we have the inductive case in which $U$ is obtained by $V$ via multiplication of a row by a invertible element $i \in \PID$. We denote with $i^{-1} \in \PID$ the multiplicative inverse of $i$. The circuit representing $U$ has the following shape, where $z+1+m=n$:
 \[
\lower11pt\hbox{$\includegraphics[height=1cm]{graffles/circuitU.pdf}$} =
\lower14pt\hbox{$\includegraphics[height=1.3cm]{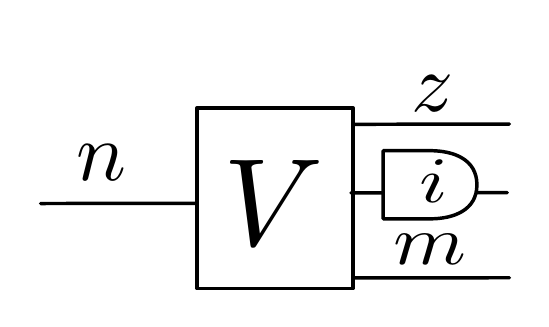}$}
\]
 and we can derive the desired equalities in $\IBRw$ as follows.
 \begin{eqnarray*}
\lower13pt\hbox{$\includegraphics[height=1.6cm]{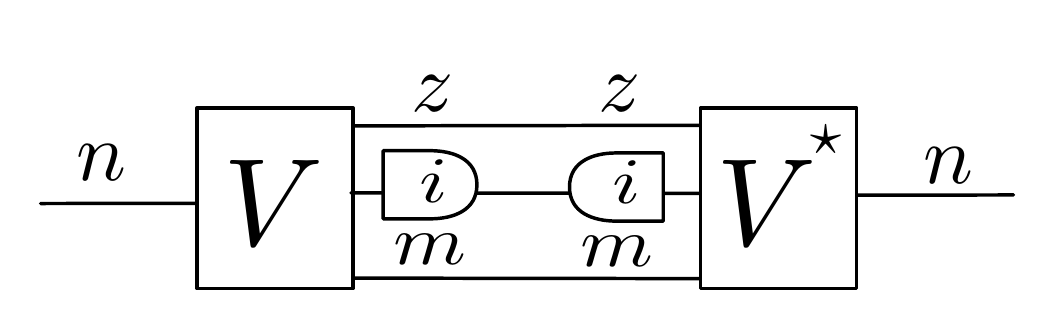}$}
\eql{\eqref{eq:lcm}}
\lower11pt\hbox{$\includegraphics[height=1cm]{graffles/circuitVVstar.pdf}$}
\eql{Ind. hyp.}
\lower4pt\hbox{$\includegraphics[height=.8cm]{graffles/circuitidn.pdf}$}
\end{eqnarray*}
\begin{eqnarray*}
\lower13pt\hbox{$\includegraphics[height=1.6cm]{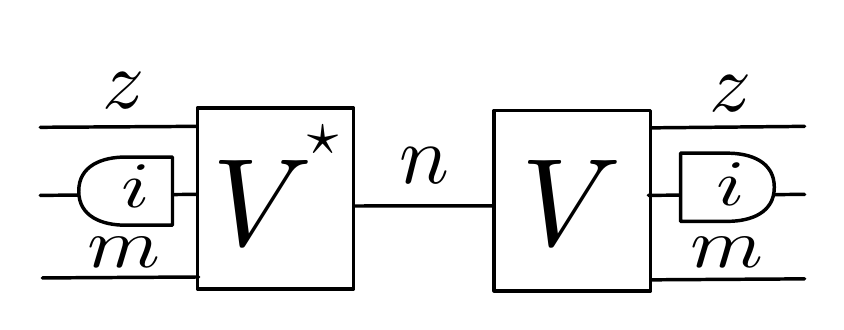}$}
\!\!\eql{IH}\!\!
\lower12pt\hbox{$\includegraphics[height=1.5cm]{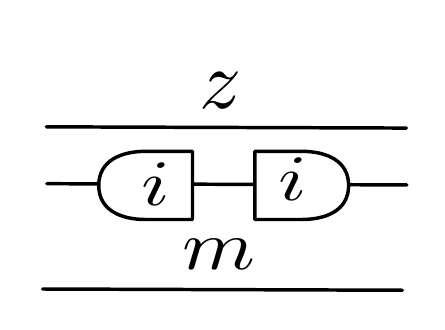}$}
& \!\!\eql{\eqref{eq:lcm}}\!\! &
\lower12pt\hbox{$\includegraphics[height=1.7cm]{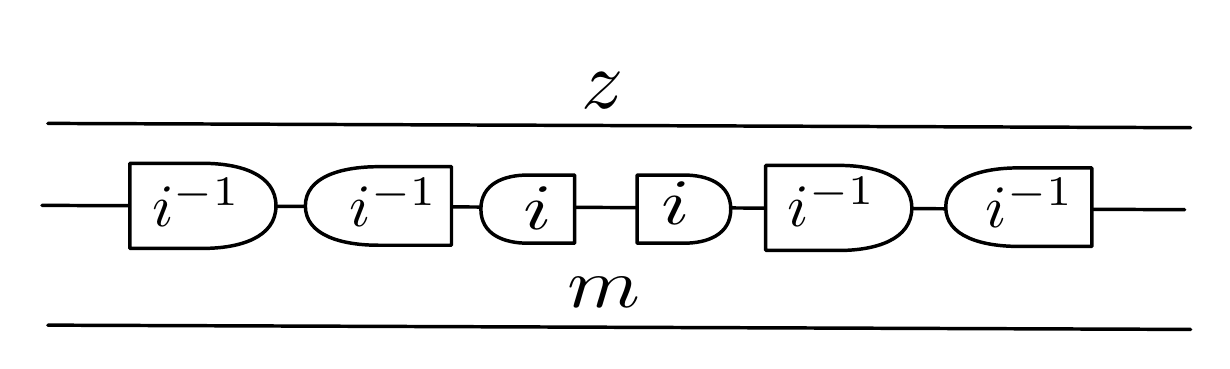}$}\\
& \!\!\eql{\eqref{eq:scalarmult}}\!\! &
\lower12pt\hbox{$\includegraphics[height=1.5cm]{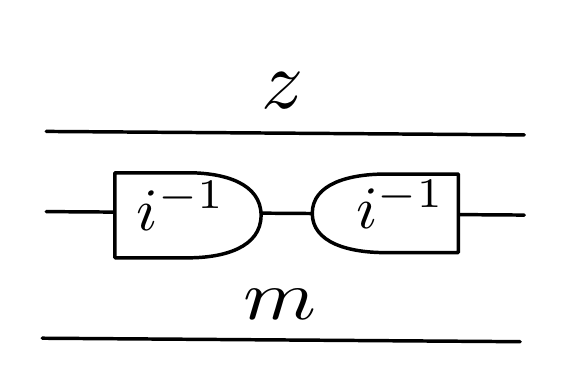}$}
\!\!\eql{\eqref{eq:lcm}}\!\!
\lower4pt\hbox{$\includegraphics[height=.8cm]{graffles/circuitidn.pdf}$}
 \end{eqnarray*}
\qed
\end{proof}

The next lemma guarantees that spans which are identified in $\Span{\VectR}$ are not distinguished by the equational theory of $\IBRw$. This means considering invertible $\PID$-matrices, as they are precisely the isomorphisms in $\VectR$; indeed, recall that arrows of $\Span{\VectR}$ are isomorphism classes of spans in $\VectR$: we identify $n \tl{A} z \tr{B} m$ and $n \tl{C} z \tr{D} m$ iff there is an invertible matrix $U \in \VectR[z,z]$ such that the following diagram commutes:
\begin{eqnarray}\label{diag:isospan}
\vcenter{
\xymatrix@C=15pt@R=15pt{ && \ar[drr]^{B} z  \ar[dll]_{A}&& \\
 n && \ar[ll]^{C} z \ar[u]^>>>>U \ar[rr]_{D}&& m }
 }
\end{eqnarray}
\begin{lemma}\label{lemma:mirror} Let $A,B,C,D,U$ be as in \eqref{diag:isospan}. Then the following equation holds in $\IBRw$:
\[
\lower11pt\hbox{$\includegraphics[height=1cm]{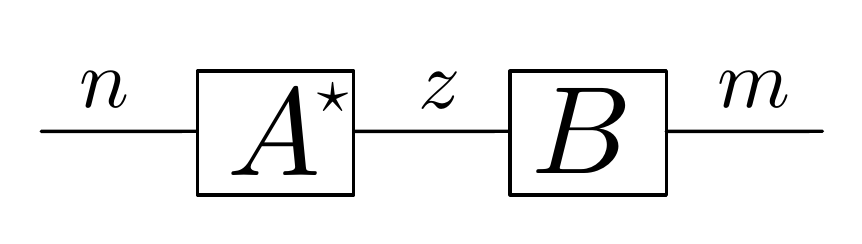}$} = \lower11pt\hbox{$\includegraphics[height=1cm]{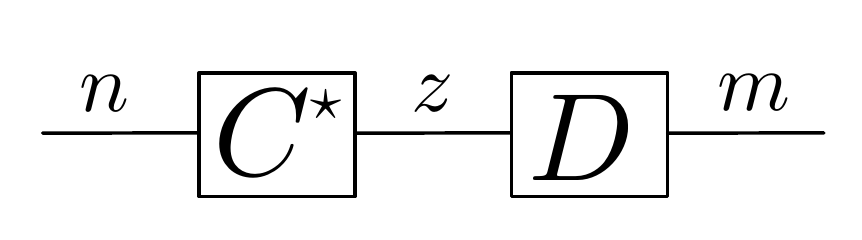}$}
\]
\end{lemma}
\begin{proof} Since $\ABR \cong \VectR$, commutativity of \eqref{diag:isospan} yields the following equalities of circuits in $\ABR$:
        \begin{equation}\label{eq:mirror1}
         \lower9pt\hbox{$\includegraphics[height=1cm]{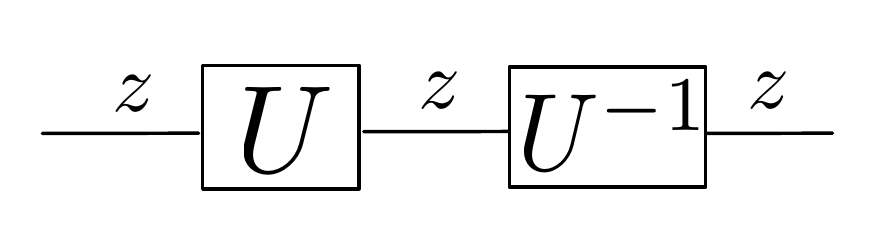}$}
   =\lower4pt\hbox{$\includegraphics[height=1cm]{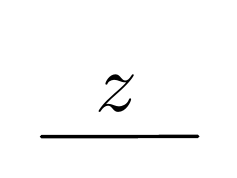}$}
   =          \lower9pt\hbox{$\includegraphics[height=1cm]{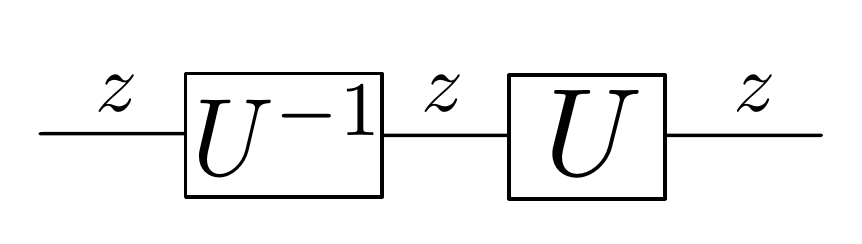}$}
   \end{equation}
        \begin{multicols}{2}
     \noindent
        \begin{equation}\label{eq:mirror2}    \lower7pt\hbox{$\includegraphics[height=.8cm]{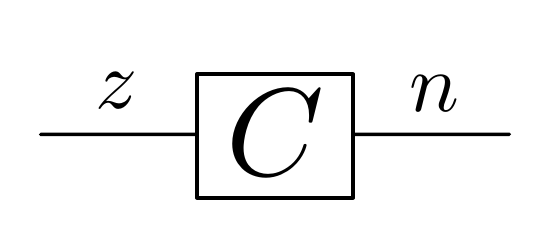}$}
  \!=\! \lower7pt\hbox{$\includegraphics[height=.8cm]{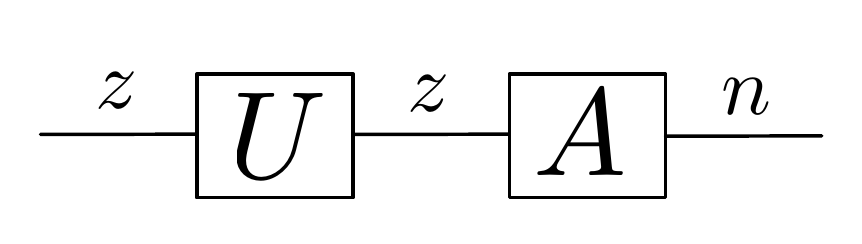}$}
   \end{equation}
        \begin{equation}   \label{eq:mirror3} \lower7pt\hbox{$\includegraphics[height=.8cm]{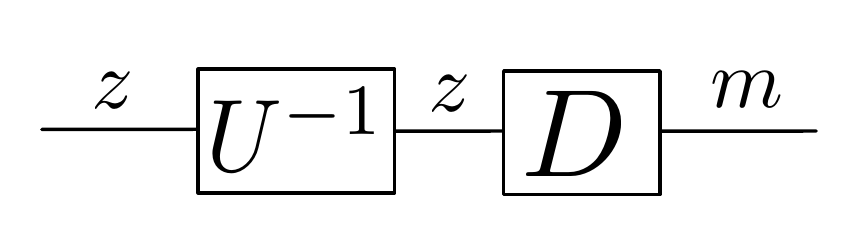}$}
  \!=\! \lower7pt\hbox{$\includegraphics[height=.8cm]{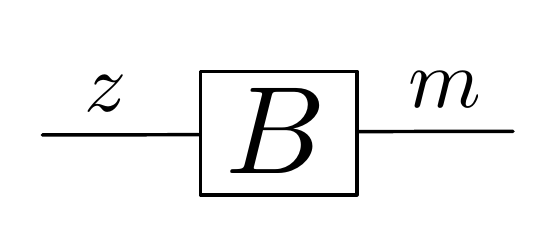}$}
   \end{equation}
 \end{multicols}
Since $\ABR$ is a sub-theory of $\IBRw$, these equations are also valid in $\IBRw$. The statement of the lemma is then given by the following derivation.
\begin{eqnarray*}
\lower11pt\hbox{$\includegraphics[height=1cm]{graffles/circuitCstarDz.pdf}$}
& \eql{\eqref{eq:mirror1}} & \lower11pt\hbox{$\includegraphics[height=1cm]{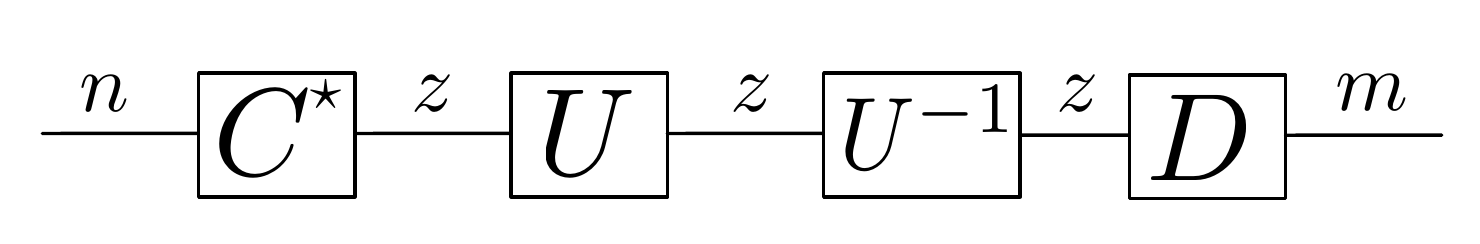}$} \\
& \eql{\eqref{eq:mirror3}} & \lower11pt\hbox{$\includegraphics[height=1cm]{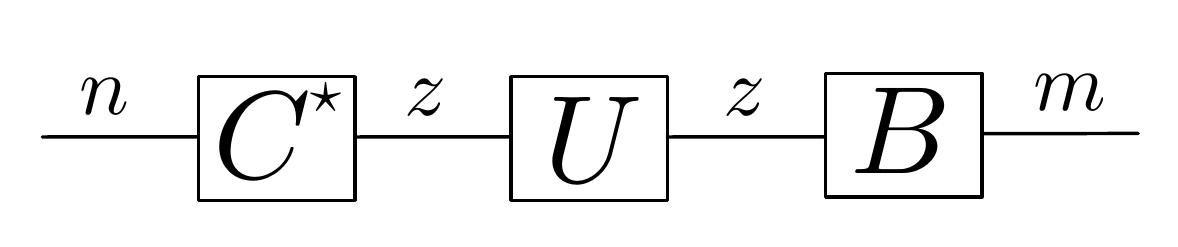}$} \\
& \eql{\eqref{eq:mirror2}} &
\lower13pt\hbox{$\includegraphics[height=1.2cm]{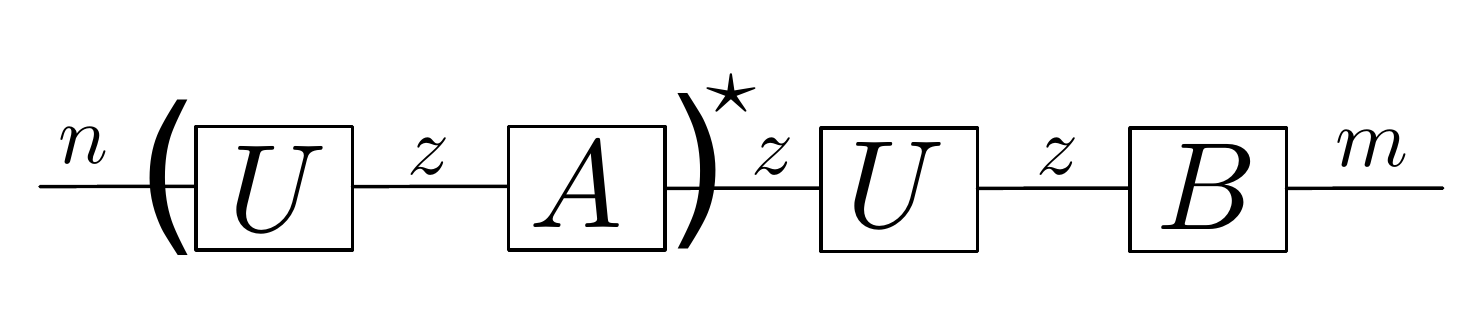}$} \\
& \eql{Def. $\coc{(\cdot)}$} &
\lower11pt\hbox{$\includegraphics[height=1cm]{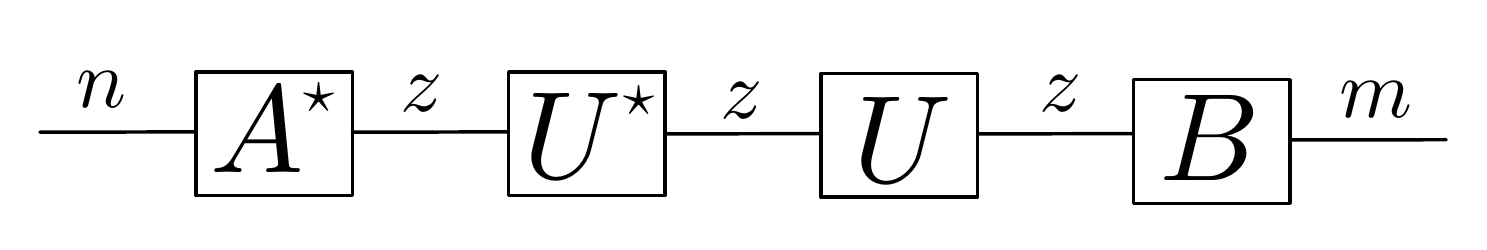}$} \\
& \eql{Lemma~\ref{lemma:invertiblestar}} &
\lower11pt\hbox{$\includegraphics[height=1cm]{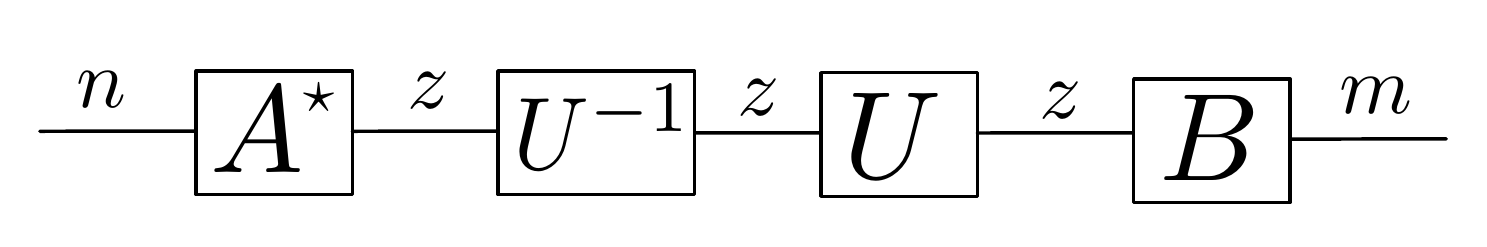}$} \\
& \eql{\eqref{eq:mirror1}} &
\lower11pt\hbox{$\includegraphics[height=1cm]{graffles/circuitAstarB.pdf}$}
\end{eqnarray*} \qed
\end{proof}

The next lemma is an important ingredient in the proof of Proposition~\ref{prop:IBwComplete}: it allows us to reduce, in the graphical theory, the computation of pullbacks to the computation of kernels. In the following, we use the notation $Ker(A)$ for the matrix representing the kernel of $A$ --- or, more precisely, 
the indicated arrow in the pullback square below:
\begin{equation*}
\vcenter{
\xymatrix@R=8pt@C=10pt{
&\ar[dl]_{\Ker{A}} r \pushoutcorner \ar[dr]^{\finVect} & \\
n \ar[dr]_{A} & & 0 \ar[dl]^{\initVect}\\
& z  & }
}
\end{equation*}

\begin{lemma}\label{lemma:pbKernel} Given a pullback square in $\VectR$ as on the left, the equation on the right holds in $\IBRw$:
\[\xymatrix@R=10pt@C=10pt{
&\ar[dl]_{C} r \pushoutcorner \ar[dr]^{D} & \\
n \ar[dr]_{A} & & m \ar[dl]^{B}\\
& z  & }
\qquad \qquad
\lower20pt\hbox{
\lower12pt\hbox{$\includegraphics[height=1cm]{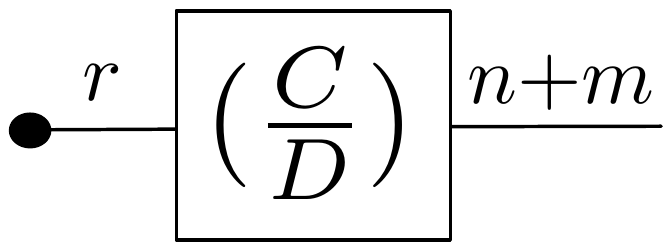}$} =
\lower8pt\hbox{$\includegraphics[height=.8cm]{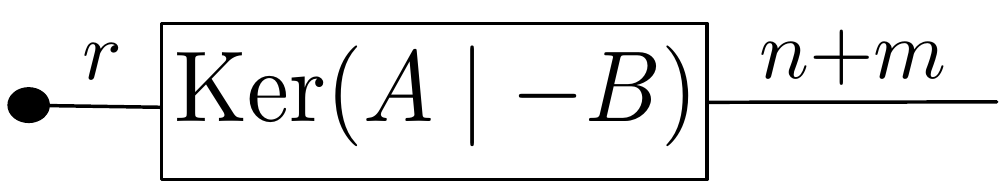}$}
}
\]
\end{lemma}
\begin{proof}
The pullback of $A\: n \to z$ and $B\: m \to z$ in the category of $\PID$-modules and linear maps can be obtained by computing the kernel of the matrix $(A|{-B})$. The pullback span ($C'$ and $D'$) then arises by post-composing $\Ker{A|-B}$ with the projections out of the biproduct $n \tns m$.
\[\xymatrix@C=40pt@R=32pt{
&\ar[dl]_{C'} \PID^r \ar[d]|{\Ker{A|{-B}}} \ar[dr]^{D'} & \\
\PID^n \ar[dr]_{A}  & \PID^n \tns \PID^m \ar[d]|{(A|-B)} & \PID^m \ar[dl]^{B} \\
& \PID^z  & }\]
Thus the spans $ \tl{C'} \tr{D'}$
 and $\tl{C} \tr{D}$ are isomorphic and, using the conclusion of Lemma~\ref{lemma:mirror}, we infer that
\begin{equation}\label{eq:mirrorAppliedToKernel} \tag{$\nabla$}
\lower8pt\hbox{$\includegraphics[height=.8cm]{graffles/circuitCstarDr.pdf}$} =
 \lower6pt\hbox{$\includegraphics[height=.6cm]{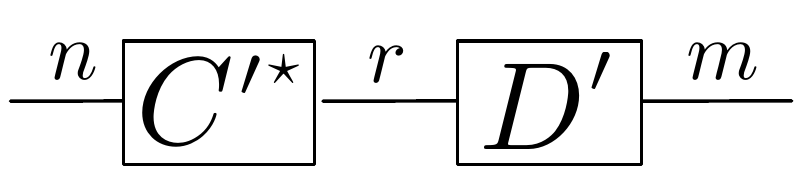}$}
\end{equation}
 from which follows that
\begin{equation}\label{eq:mirrorAppliedToKernel2} \tag{$\triangle$}
\lower10pt\hbox{$\includegraphics[height=1cm]{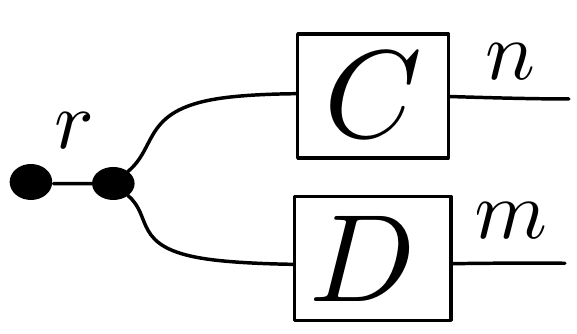}$} \!\!\eql{\eqref{eq:ccsliding2}}\!\!
\lower10pt\hbox{$\includegraphics[height=1cm]{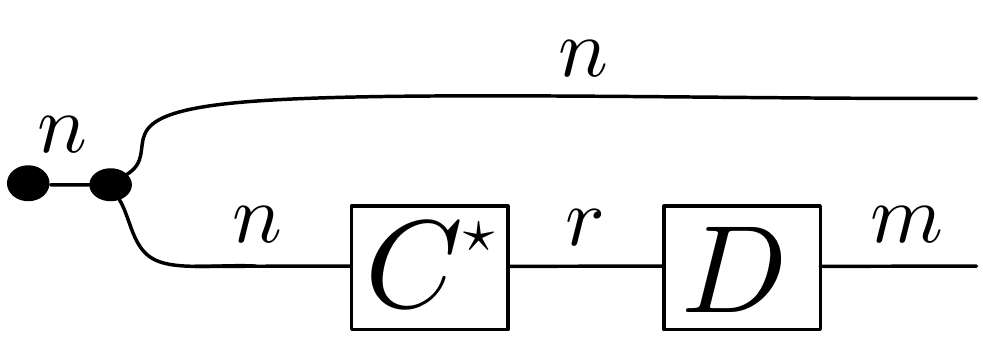}$} \!\!\eql{\eqref{eq:mirrorAppliedToKernel}}\!\!
\lower10pt\hbox{$\includegraphics[height=1cm]{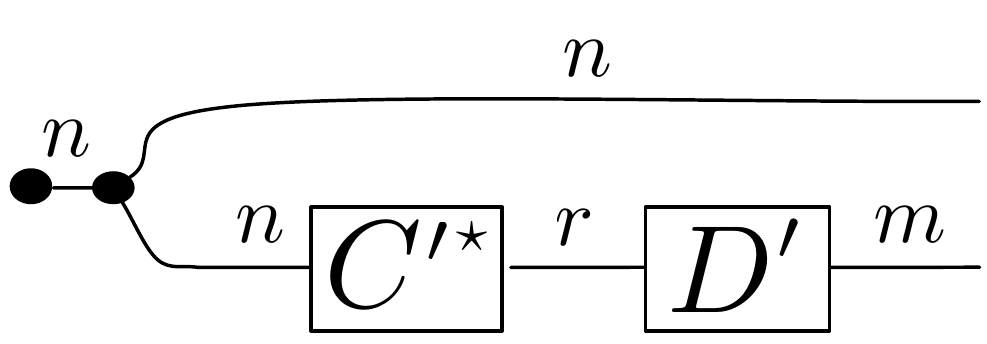}$} \!\!\eql{\eqref{eq:ccsliding2}}\!\!
\lower10pt\hbox{$\includegraphics[height=1cm]{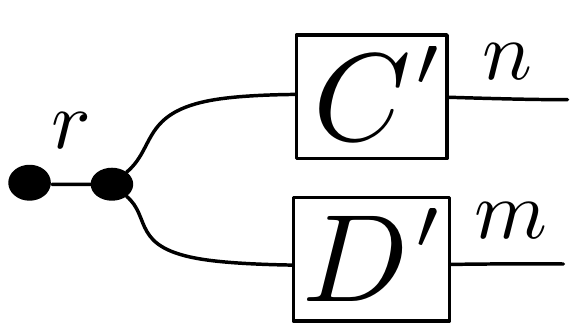}$}.
\end{equation}
 We can now conclude the proof of our statement:
 \[
\lower8pt\hbox{$\includegraphics[height=.8cm]{graffles/BcounitCoverD.pdf}$} \!\eql{Def. $\left(\frac{C}{D}\right)$}\!
\lower10pt\hbox{$\includegraphics[height=1cm]{graffles/ccCD.pdf}$} \!\eql{\eqref{eq:mirrorAppliedToKernel2}}\!
\lower10pt\hbox{$\includegraphics[height=1cm]{graffles/ccCprimeDprime.pdf}$} \!\eql{Def. $\left(\frac{C'}{D'}\right)$}\!
\lower8pt\hbox{$\includegraphics[height=.8cm]{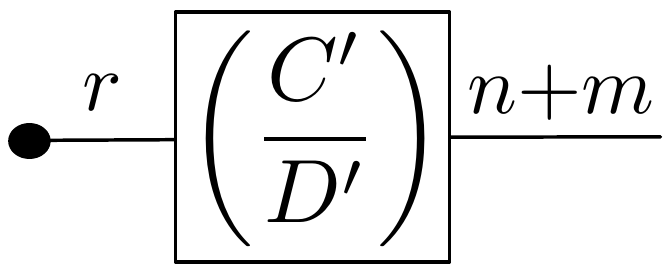}$} \!\eql{}\!
\lower6pt\hbox{$\includegraphics[height=.6cm]{graffles/BcounitkernelAminusB.pdf}$}.
\]
\qed\end{proof} 

%% file: source/5_ProofMatrixEqualKernel.tex
 We now describe how the kernel computation of a matrix can be formulated
 within the equational theory of $\IBRw$.  We first recall some linear algebra that will be used in our argument.
 \begin{definition}\label{Def:HNF} An $m \times n$ matrix $A$ is said to be in \emph{Hermite normal form} (HNF) if there is a natural number $r \leq n$ and a strictly increasing function $f \: [r+1,n] \to [1,m]$ associating to each column $i$ a row $f(i)$, such that:
 \begin{enumerate}
   \item the first $r$ columns of $A$ have all entries with value $0$;
   \item for all columns $i$ with $r+1 \leq i \leq n$, $A_{f(i),i} \neq 0$ and
   \item for all $j \gr f(i)$, $A_{j,i} = 0$.
 \end{enumerate}
\end{definition}

A matrix in HNF is given in Example~\ref{ex:HNF} below. In the following we list some useful properties of the HNF, the first of which is immediate.
\begin{lemma}\label{lemma:HNFtriangular} Suppose that $A$ is an $m \times n$ matrix in HNF and fix a column $i \leq n$. Then $A_{f(i),j} = 0$ for all columns $j \ls i$.
\end{lemma}
Every $\PID$-matrix $A$ is column-equivalent to some matrix $B$ in HNF (see e.g. \cite{McDonald:linearalgebrarings,Cohen:1995ComputAlg}).
The transformation of $A$ into $B$ can be encoded as an invertible matrix $U$, obtained by applying to the identity matrix the sequence of elementary column operations allowing to pass from $A$ to $B$. Then $B = AU$ and we can compute from $U$ the kernel of $A$ as follows.

\begin{proposition}\label{prop:kernelfromHNF} For an $m \times n$ matrix $A$, let $B = AU$ be its HNF and $r \leq m$ the number of initial $0$-columns of $B$. Then the first $r$ columns of $U$ form a basis for the kernel of $A$.
\end{proposition}
\begin{proof} A proof can be found for the PID of integers in \cite[Prop. 2.4.9]{Cohen:1995ComputAlg}, which we reformulate here for an arbitrary PID. We include the details because the next result will be essentially a graphical rendition of the argument.

 For $i \leq r$, let $\vlist{u}_i$ be the $i$-th column of $U$. By definition $A \vlist{u}_i = B_i$, which is a $0$-vector because $i \leq r$. Thus all first $r$ columns of $U$ are elements of the kernel of $A$. Conversely, let $\xx$ be a vector such that $A\xx = 0$. Then $A\xx = A U U^{-1} \xx = B U^{-1} \xx$ because $U$ is invertible. Let $y_1, \dots, y_n$ be the coordinates of $\yy \df U^{-1}\xx$. For each $i$ in $[r+1,n]$, we show that $y_i = 0$, by backward induction on $i$. This unfolds as a kind of ``chain reaction'':
\begin{itemize}
  \item[(I)] \label{pt:chainreaction1} if $i =n$, let $f(n)$ be given as in Definition~\ref{Def:HNF}. Since $B\yy = 0$, then the $f(n)$-th coordinate of $B\yy$ is
      \begin{equation}\label{eq:propHNF1} \tag{$\bigtriangleup$}
          B_{f(n),1}y_1 + \dots + B_{f(n),n}y_n = 0.
      \end{equation}
      By Lemma~\ref{lemma:HNFtriangular}, $B_{f(n),1}, \dots, B_{f(n),{n-1}}$ are all equal to $0$, meaning by \eqref{eq:propHNF1} that $B_{f(n),n}y_n = 0$. By property 2 of HNF, $B_{f(n),n} \neq 0$ and thus, since $\PID$ has no non-zero divisors, $y_n = 0$.
  \item[(II)] \label{pt:chainreaction2} For $i$ with $r \ls i \ls n$, the $f(i)$-th coordinate of $B\yy$ is $B_{f(i),1}y_1 + \dots + B_{f(i),n}y_n = 0$ and by induction hypothesis $y_j = 0$ for all $j$ such that $i \ls j \leq n$. By Lemma~\ref{lemma:HNFtriangular}, $B_{f(i),1}, \dots, B_{f(i),{i-1}}$ are all equal to $0$, which means, analogously to the base case, that $B_{f(i),i}y_i = 0$ and since $B_{f(i),i}$ then $y_i = 0$.
  \item[(III)] Thus we proved that the coordinates $y_{r+1},\dots,y_n$ of $\yy$ are equal to $0$. Instead the first $r$ coordinates of $\yy$ can be arbitrary, because the $j$-th row of $B\yy$, for $j \leq r$, is give by $B_{j,1}y_1 + \dots + B_{j,n}y_n = 0$ and we know that, by property 1 of HNF, the entries $B_{j,1},\dots, B_{j,n}$ have value $0$.
\end{itemize}
Therefore the kernel of $B$ is generated by the first $r$ canonical basis vectors $\vlist{c}_1, \dots \vlist{c}_r$ of $\PID^n$. Since $B = AU$, then $U\vlist{c}_1, \dots, U\vlist{c}_r$ form a basis for the kernel of $A$. But those are just the first $r$ columns of $U$: hence we have proven the statement of the theorem.
\qed\end{proof}

We now recast the core of Proposition~\ref{prop:kernelfromHNF} ``in purely graphical terms''. For an instance of the construction used in the proof, see Example~\ref{ex:HNF}.

\begin{lemma}\label{lemma:BHNFequalKernel} Let $B$ be an $m \times n$ $\PID$-matrix in HNF and $r$ the number of initial $0$-columns of $B$. Then the following holds in $\IBRw$:
$$\lower13pt\hbox{$\includegraphics[height=1.2cm]{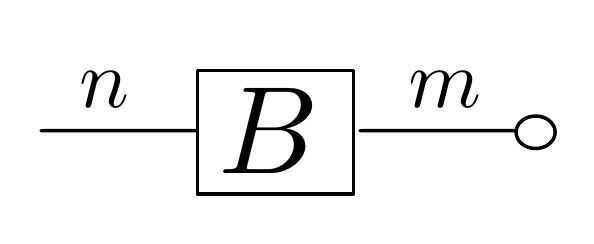}$} = \lower13pt\hbox{$\includegraphics[height=1.4cm]{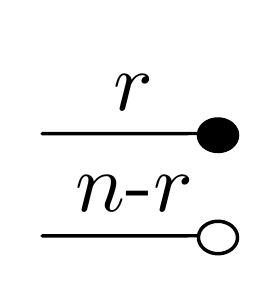}$}$$
\end{lemma}
\begin{proof}
The idea is to show that the kernel computation described in the proof of Proposition~\ref{prop:kernelfromHNF} can be carried out on circuits using the equational theory of $\IBRw$. Since $B$ is in HNF, the corresponding circuit (in matrix form) can be assumed of a particular shape, that we depict below right. 

\noindent \begin{minipage}[t]{.60\textwidth}
$P$ is some circuit only made of symmetries $\symNet$ and scalars $\scalar$ as basic components. By property 1 of HNF, the first $r$ columns of $B$ only have $0$ entries, meaning that the topmost $r$ ports on the left boundary are not connected to the right boundary. Also, by Lemma~\ref{lemma:HNFtriangular} we know that the $f(n)$-th row of $B$ (where $f \: [r+1,n] \to [1,m]$ is as in Definition~\ref{Def:HNF}) has only one non-$0$ value $k \in \PID$, at position $B_{f(n),n}$.  In circuit terms, this allows us to assume that the $f(n)$-th port on the right
\end{minipage}
\begin{minipage}[t]{.35\textwidth}
\vspace{-.8cm}$$\includegraphics[height=5.5cm]{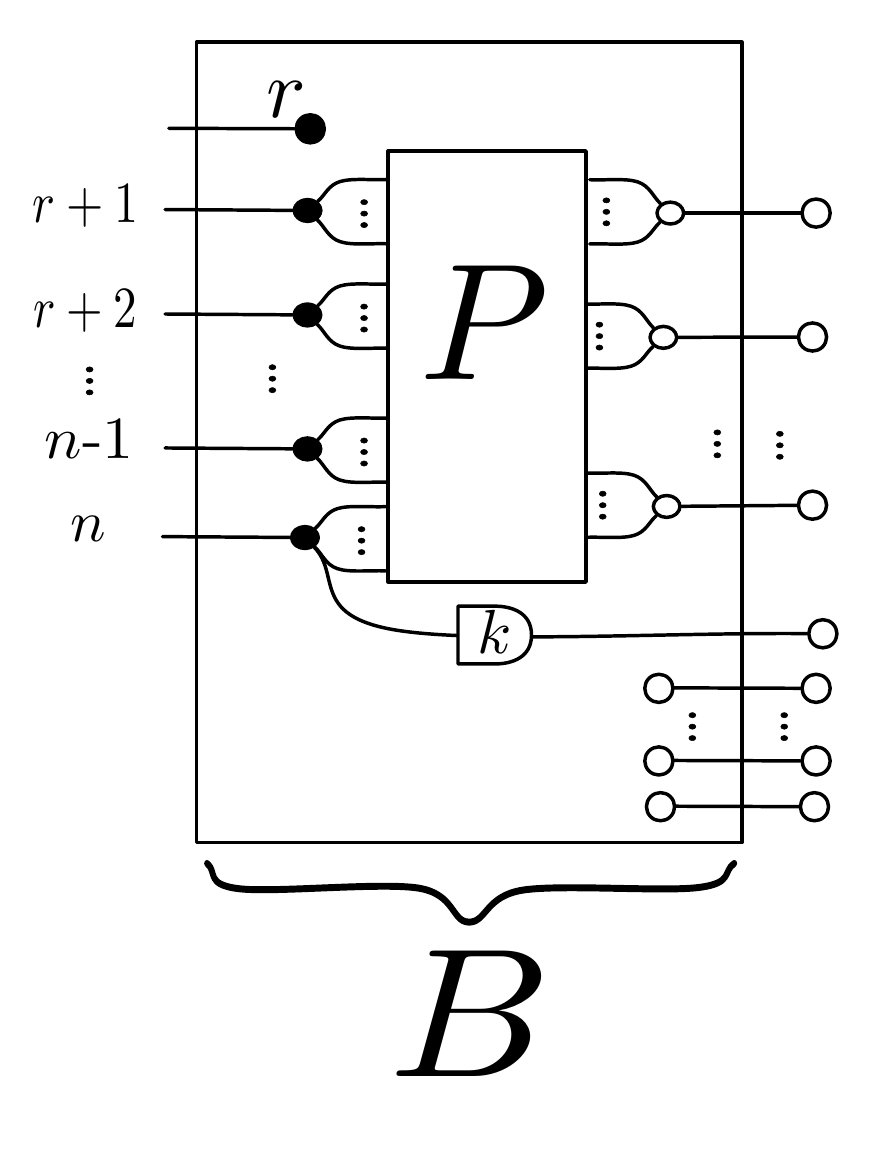}$$
\end{minipage}

\noindent boundary only connects to the $n$-th and last port on the left boundary.
As yet another consequence of the definition of HNF, we know that, for each $i$ with $m \geq i \gr f(n)$, row $i$ only has $0$ entries, allowing us to represent all the rows below $f(n)$ in the circuit above as ports on the right boundary not connected to any port on the left.
Once we plug counits on the right of the circuit representing $B$, we trigger the chain reaction described in the proof of Proposition~\ref{prop:kernelfromHNF}, which we now reproduce in circuit terms. By backward induction on $i$ with $n \geq i \gr r$, we construct circuits $B_n, \dots , B_{r+1}$ such that:
\[
\lower11pt\hbox{$\includegraphics[height=1cm]{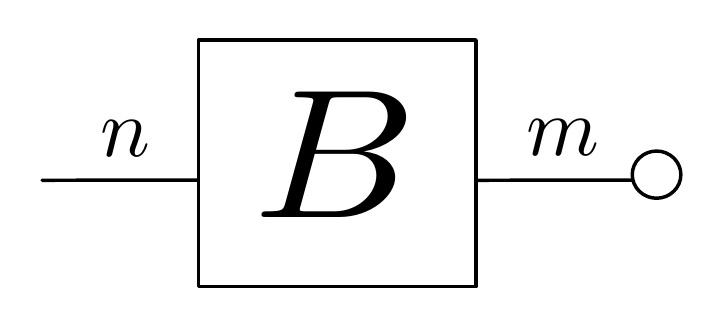}$} =
\lower11pt\hbox{$\includegraphics[height=1cm]{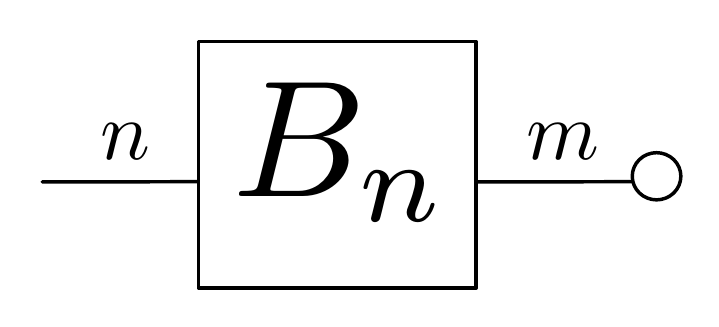}$}
= \dots =
\lower11pt\hbox{$\includegraphics[height=1cm]{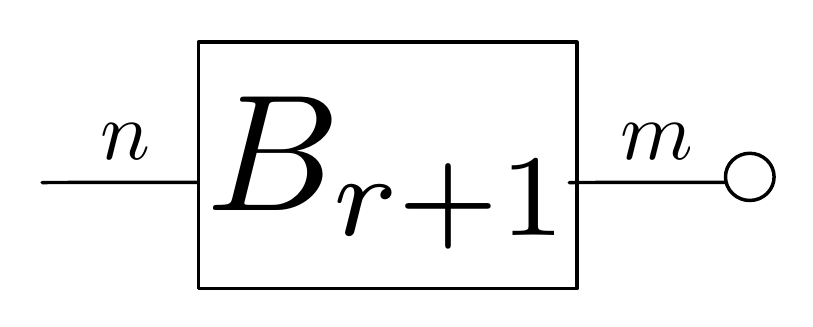}$} =
\lower11pt\hbox{$\includegraphics[height=1.2cm]{graffles/circuitCounitsrn-r.pdf}$}
\]
Clearly, this suffices to prove the main statement.
\begin{itemize}
\item[(I)] For the base case, suppose $i = n$. Since $k \neq 0$, we can use the derived law \eqref{eq:wunitcancelbcomult} of $\IBRw$ to ``disconnect'' the $n$-th port on the left from any port on the right. We define $B_n$ in terms of the resulting circuit.
\begin{eqnarray*}
\lower81pt\hbox{$\includegraphics[height=5cm]{graffles/circuitBdetail.pdf}$} \ \eql{\eqref{eq:wunitcancelbcomult}} \
\lower84pt\hbox{$\includegraphics[height=5cm]{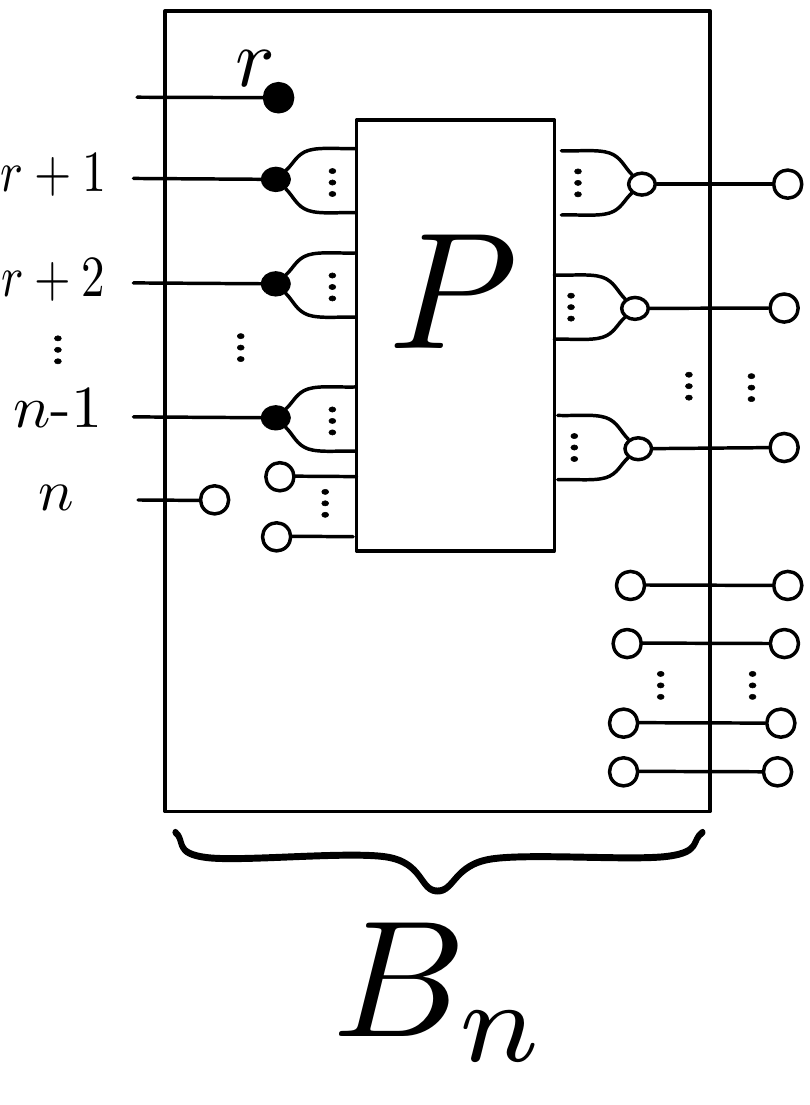}$} \dfop
\lower11pt\hbox{$\includegraphics[height=1.2cm]{graffles/circuitBncounits.pdf}$}
\end{eqnarray*}
We assign the name $P_n$ to the circuit $P$ depicted above and proceed with the inductive step of $i$ with $n \gr i \gr r$.

\item[(II)]
\noindent \begin{minipage}[t]{.60\textwidth}
The inductive construction
gives us a circuit $B_{i+1}$ as on the right. The $i$-th port on the left boundary corresponds to column $i$ in $B$ and thus it is assigned a row $f(i)$. This corresponds to the $f(i)$-th port on the right boundary of the circuit representing $B_{i+1}$. By Lemma~\ref{lemma:HNFtriangular}, such a port has no connections with ports $1,\dots,i-1$ on the left boundary. Moreover, by inductive hypothesis it also has no connections with ports $i+1,\dots,n$ on the left boundary. Therefore port $f(i)$ on the right
\end{minipage}
\begin{minipage}[t]{.35\textwidth}
\vspace{-.3cm}\[
\includegraphics[height=5cm]{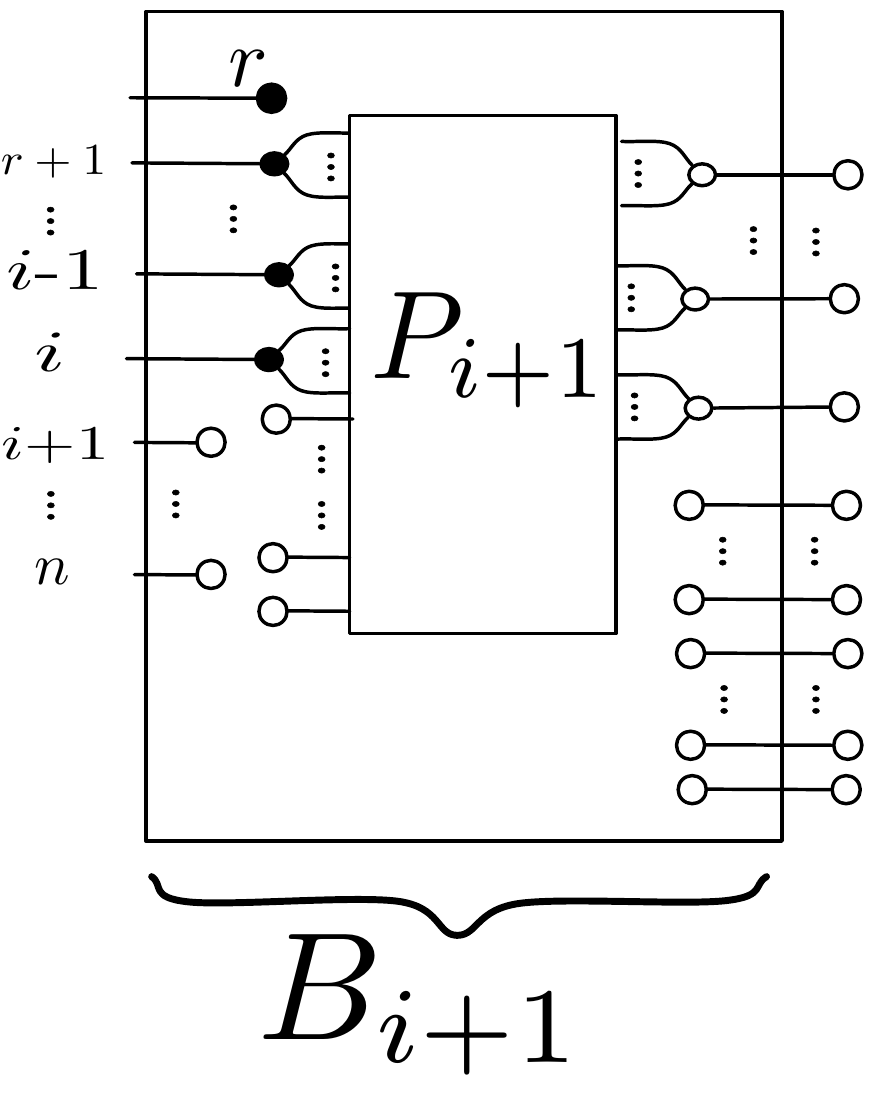}
\]
\end{minipage}
 \noindent connects only to port $i$ on the left. These connections are part of the circuit $P_{i+1}$ --- which by inductive construction only contains $\symNet$ and $\scalar$ as basic components. It should then be clear that we can ``move port $f(i)$ towards the left side of the circuit'', isolating its connections from the others in $P_{i+1}$, while preserving equality in $\IBRw$.
The resulting circuit is the depicted below, where $P_i$ results from the rearrangement of $P_{i+1}$ in order to allow the move of port $f(i)$ towards the left side of the circuit.
\begin{equation} \label{eq:circuitBn+1}
\lower70pt\hbox{$\includegraphics[height=5cm]{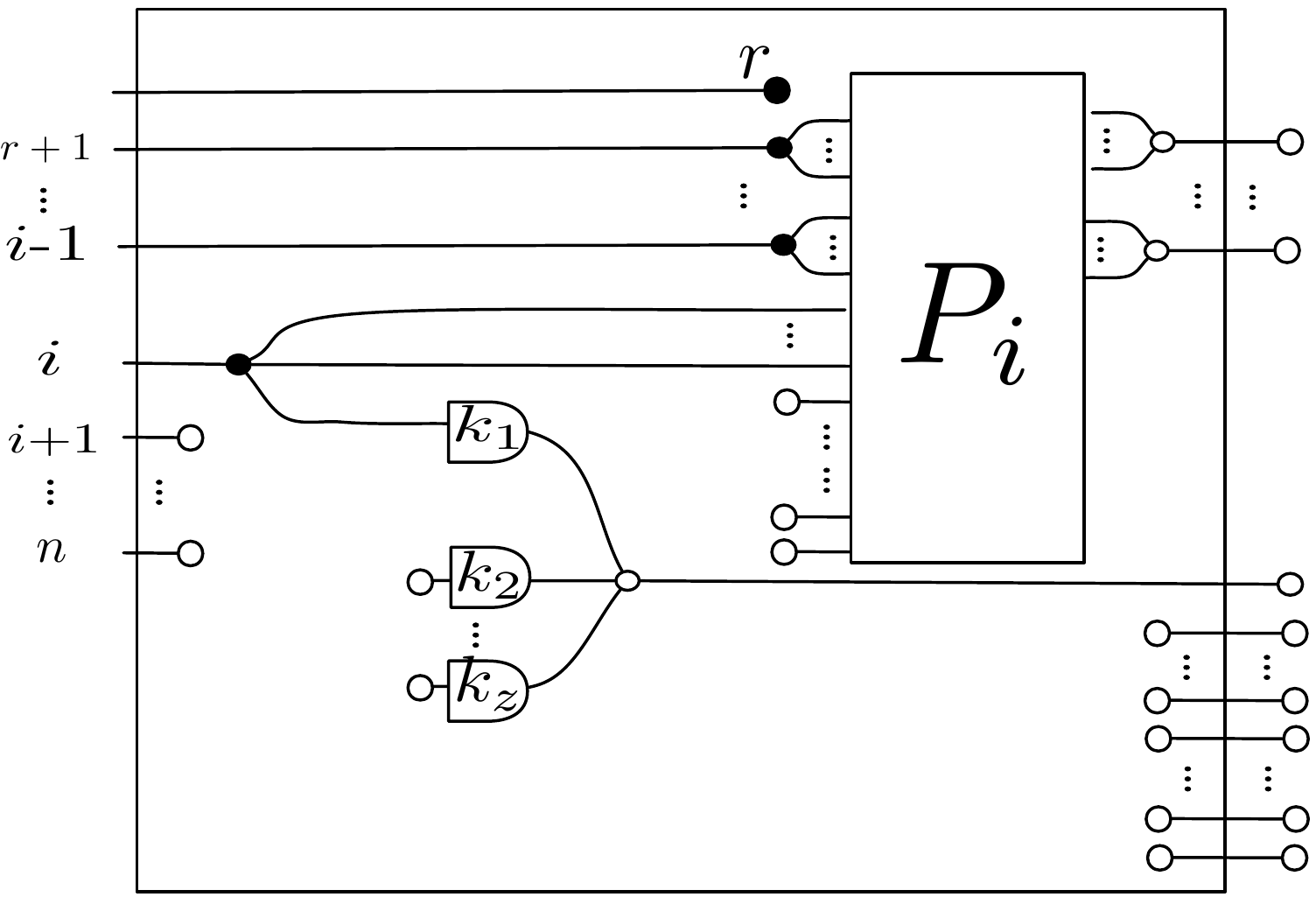}$}
 \end{equation}
 We now focus on the sub-diagram depicting the connection of port $i$ on the left with (former) port $f(i)$. In the derivation below, \eqref{eq:wunitcancelbcomult} can be applied because $k_1 = B_{f(i),i} \neq 0$.
\begin{equation*}
\lower35pt\hbox{$\includegraphics[height=2.5cm]{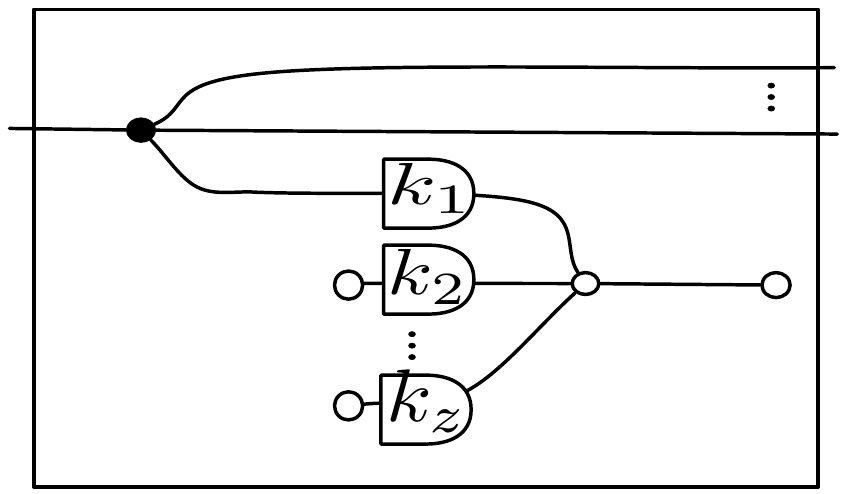}$}
 \eql{\eqref{eq:scalarwunit},\eqref{eq:wmonunitlaw}}
\lower21pt\hbox{$\includegraphics[height=1.6cm]{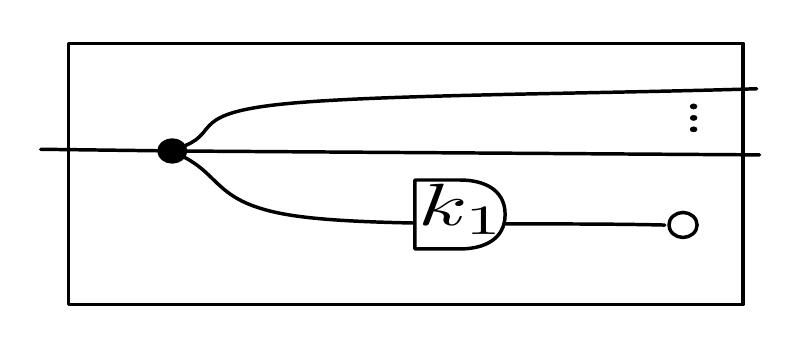}$}
 \eql{\eqref{eq:wunitcancelbcomult}}
\lower15pt\hbox{$\includegraphics[height=1.1cm]{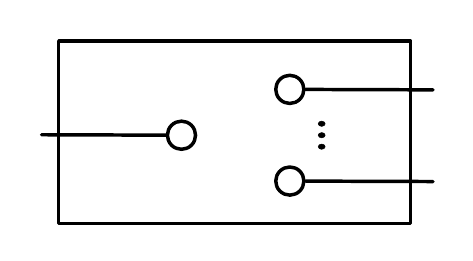}$}
 \end{equation*}
Thus \eqref{eq:circuitBn+1} is equal to the circuit below left, from which we define $B_{i}$.
\[
\lower62pt\hbox{$\includegraphics[height=4.4cm]{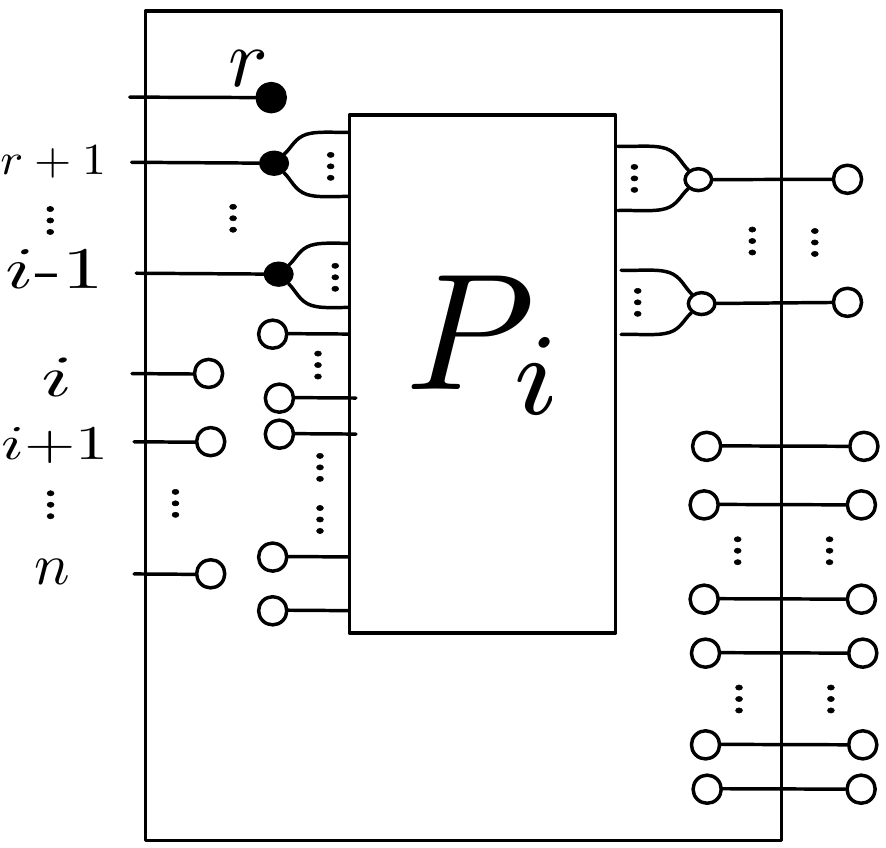}$}
\quad \dfop \quad
\lower15pt\hbox{$\includegraphics[height=1.5cm]{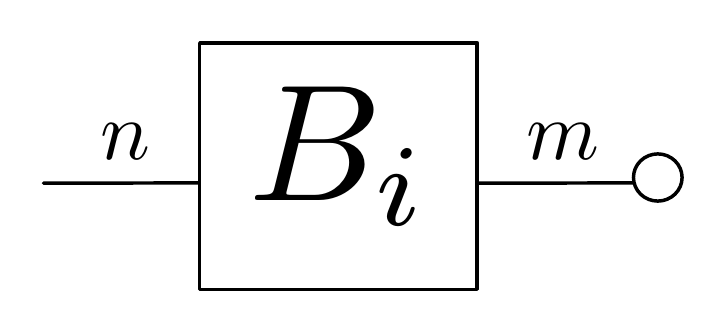}$}
\]
\item[(III)] Finally, at step $r+1$, our inductive construction produces a circuit as on the left below. We have disconnected all ports $i$ on the left and all ports $f(i)$ on the right: $P_{r+1}$ only contains the entries $\scalar$ on rows not in the image of $f$ (if any). We can then easily remove also this last piece of information.
\begin{align*}
\lower80pt\hbox{$\includegraphics[height=5cm]{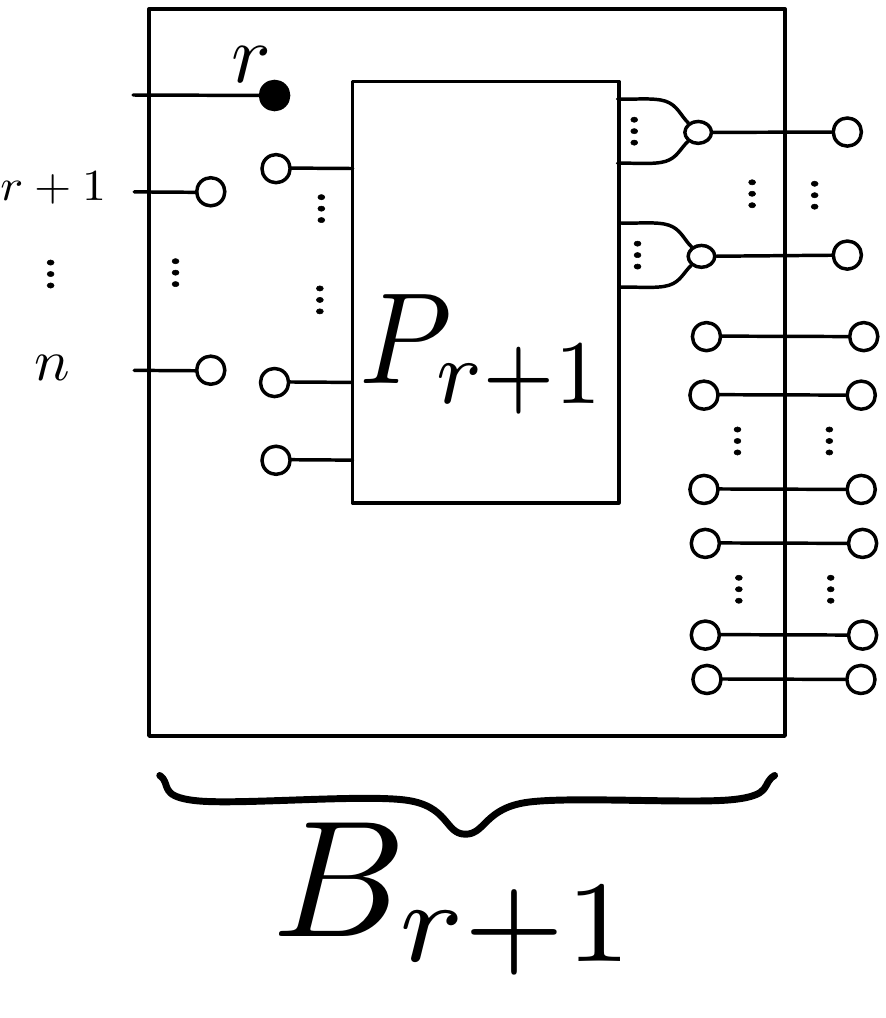}$}\ \ \  & \eql{\eqref{eq:scalarwunit},\eqref{eq:wmonunitlaw}}   & 
\lower42pt\hbox{$\includegraphics[height=3.6cm]{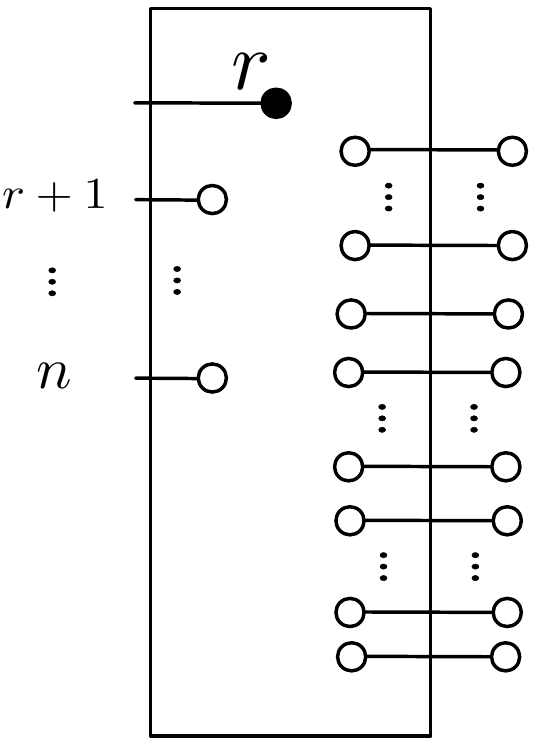}$}
\ \ \ & \eql{\eqref{eq:wbone}} &
\lower15pt\hbox{$\includegraphics[height=1.6cm]{graffles/circuitCounitsrn-r.pdf}$}
\end{align*}
For the first equality, observe that by inductive construction $P_{r+1}$ is only made of basic components of the kind $\symNet$ and $\scalar$ : the white units plugged on the left boundary of $P_{r+1}$ cancel $\symNet$ by naturality of symmetries in the symmetric monoidal category $\IBRw$ and cancel $\scalar$ by~\eqref{eq:scalarwunit}. The second equality holds by repeated application of~\eqref{eq:wbone}. \qed
\end{itemize}
\end{proof}


 \begin{example}\label{ex:HNF} We show the construction of Lemma~\ref{lemma:BHNFequalKernel} on a circuit in matrix form that represents the following $\Z$-matrix in HNF.
\begin{eqnarray*}
{\scriptsize \left(
  \begin{array}{cccc}
    0 & 0 & 2 & -1\\
    0 & 4 & 1 & -3\\
    0 & 0 & 1 & 0\\
    0 & 0 & 0 & 0\\
    0 & 0 & 0 & 3
  \end{array}
\right)}
\end{eqnarray*}
\begin{center}
\includegraphics[height=7cm]{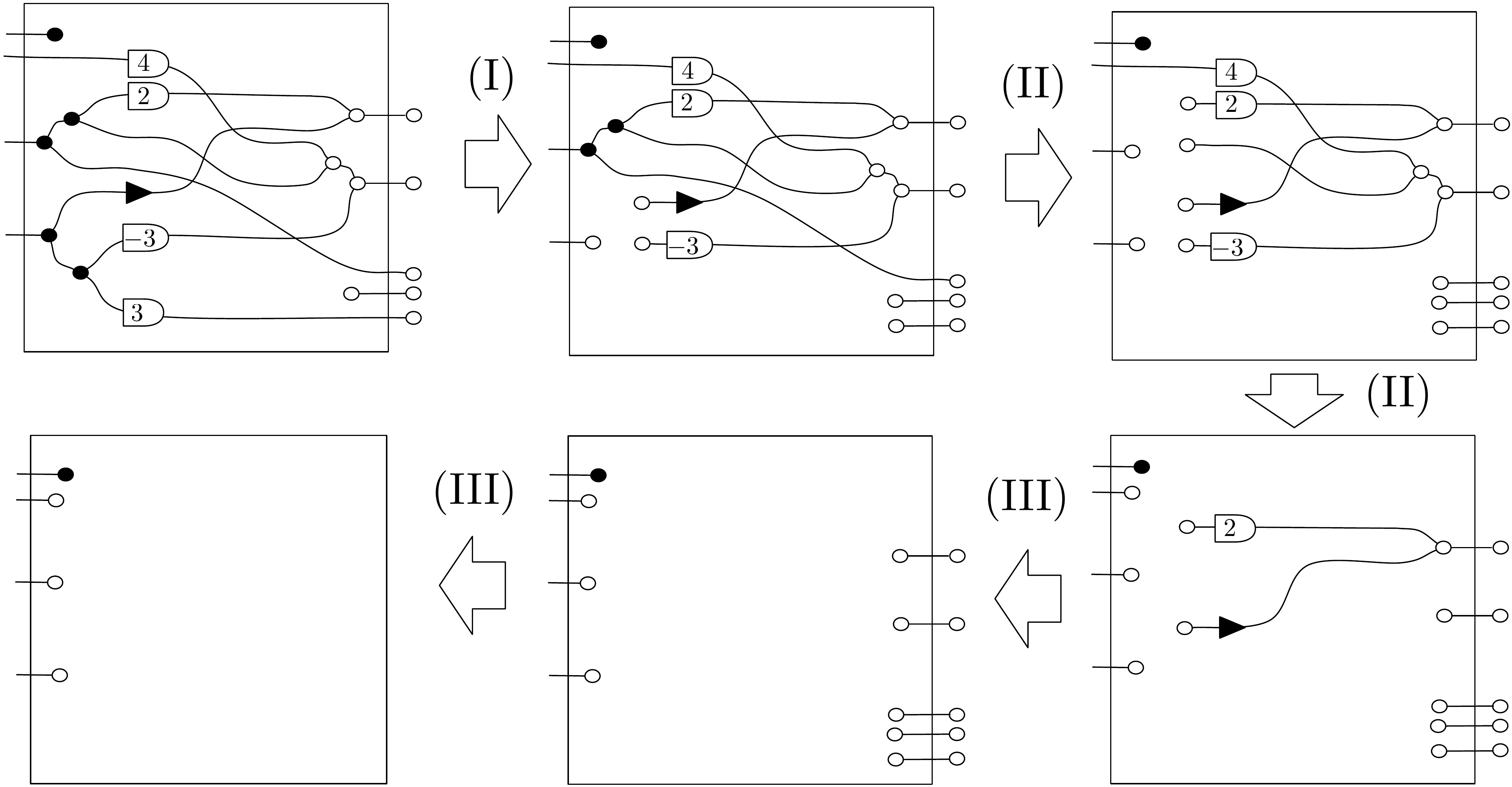}
\end{center}
\end{example}

Given $A \in \VectR[n,m]$ and $r \leq n$, let the \emph{$r$-restriction of $A$} be the matrix $\restr{A}{r} \in \VectR[r,m]$ consisting of the first $r$ columns of $A$. It is useful to make the following observation.

\begin{lemma}\label{lemma:invertiblerestriction} Let $U \in \VectR[n,m]$ be a matrix and fix $r \leq n$. Then the following holds in $\IBRw$:
$$\lower11pt\hbox{$\includegraphics[height=.9cm]{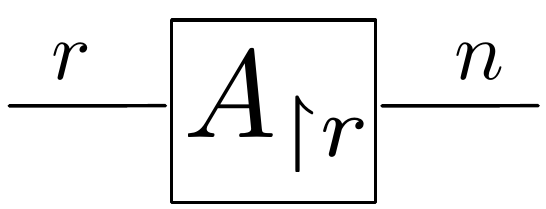}$} =
\lower9pt\hbox{$\includegraphics[height=.9cm]{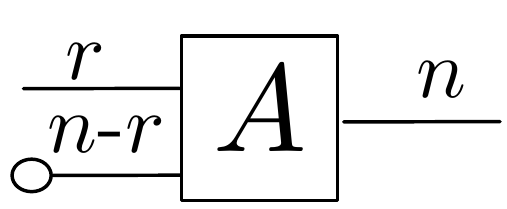}$}$$
\end{lemma}
\begin{proof} Observe that multiplying the matrix corresponding to $\lower3pt\hbox{\includegraphics[height=.7cm]{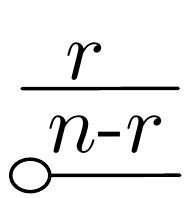}}$ by $A$ yields $\restr{A}{r}$. Then the statement holds by the isomorphism $\ABR \cong \VectR$. \qed\end{proof}

We now have all the ingredients to state the soundness of kernel computation for an arbitrary $\PID$-matrix of $\VectR$.

\begin{proposition}\label{prop:matrixequalkenrel}
Let $A \in \VectR[n,m]$ be a $\PID$-matrix. Then the equation below left, which corresponds to the pullback on the right, is valid in $\IBRw$.
\begin{eqnarray}\label{eq:kernelplback}
\lower13pt\hbox{$\includegraphics[height=1.2cm]{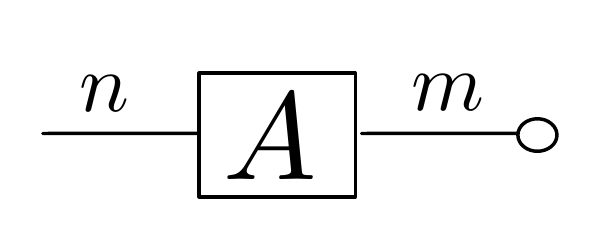}$} =
\lower13pt\hbox{$\includegraphics[height=1.2cm]{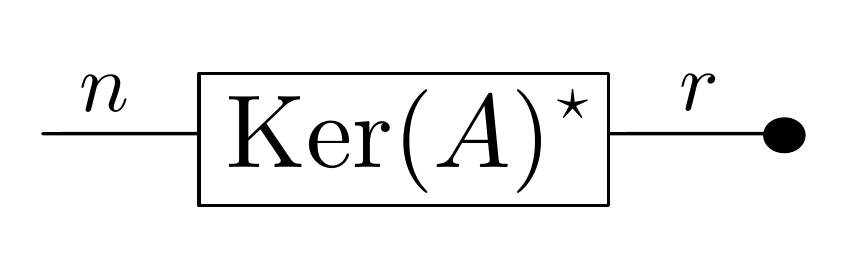}$} & \qquad &
\vcenter{
\xymatrix@R=8pt@C=10pt{
&\ar[dl]_{\Ker{A}} r \pushoutcorner \ar[dr]^{\finVect} & \\
n \ar[dr]_{A} & & 0 \ar[dl]^{\initVect}\\
& z  & }
}
\end{eqnarray}
\end{proposition}
\begin{proof}Let $B = AU$ be the HNF of $A$ for some invertible matrix $U \: n \to n$. Consider the following derivation in $\IBRw$.
\begin{eqnarray*}
\lower13pt\hbox{$\includegraphics[height=1.2cm]{graffles/circuitA.pdf}$} &\eql{}&
\lower9pt\hbox{$\includegraphics[height=.8cm]{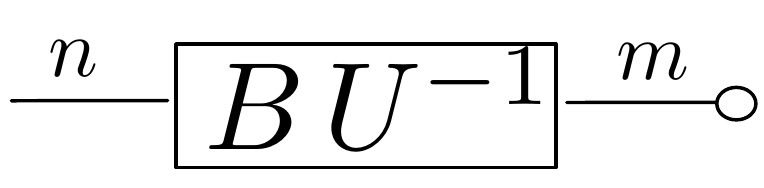}$}
\\
&\eql{}&
\lower11pt\hbox{$\includegraphics[height=1cm]{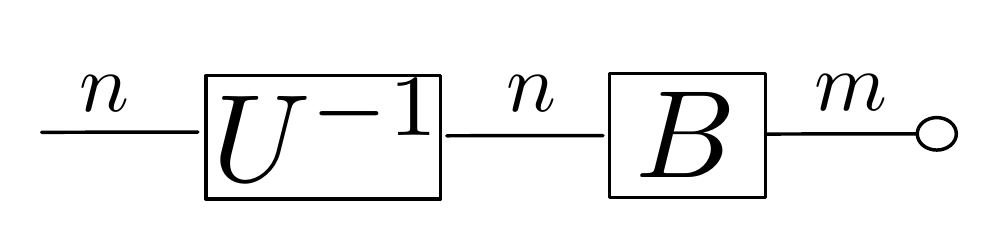}$}
\\
&\eql{Lemma \ref{lemma:invertiblestar}}&
\lower11pt\hbox{$\includegraphics[height=1cm]{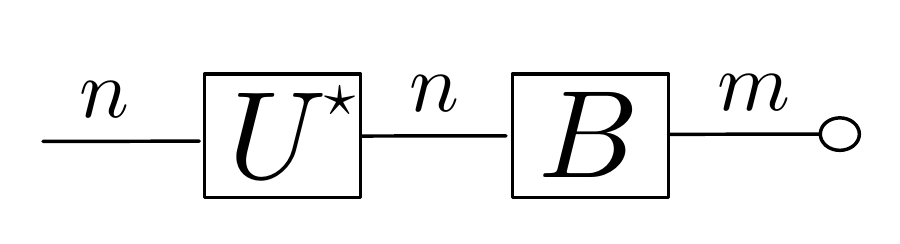}$} \\
&\eql{Lemma \ref{lemma:BHNFequalKernel}}&
\lower13pt\hbox{$\includegraphics[height=1.2cm]{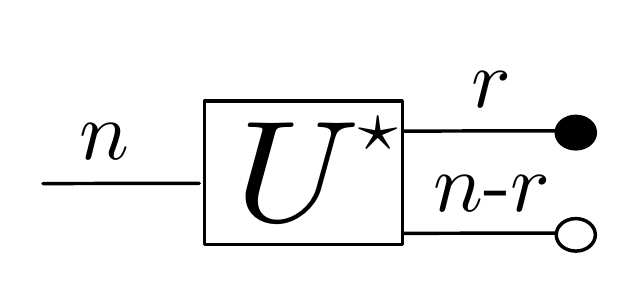}$} \\
&\eql{Prop. \ref{prop:star=refl}}&
\lower15pt\hbox{$\includegraphics[height=1.4cm]{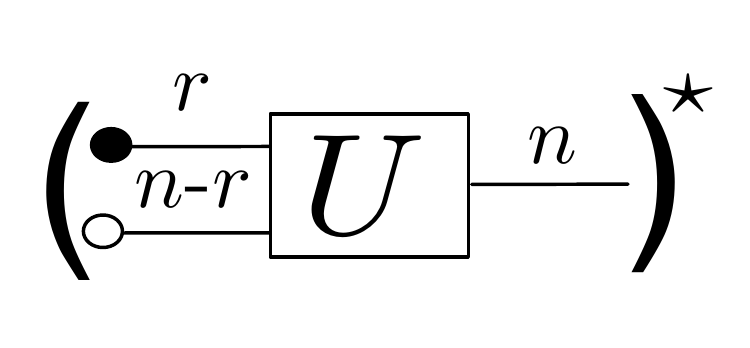}$} \\
&\eql{Lemma \ref{lemma:invertiblerestriction}}&
\lower15pt\hbox{$\includegraphics[height=1.4cm]{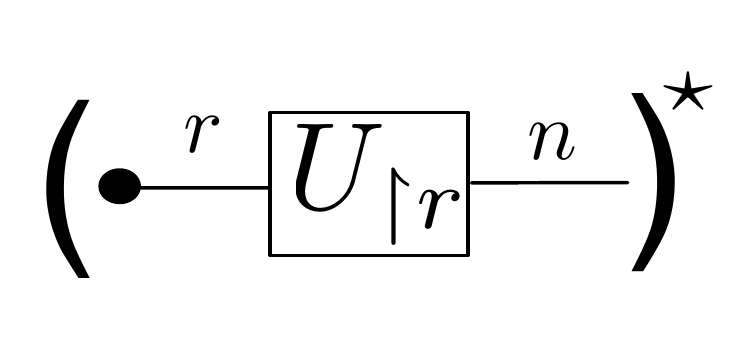}$} \\
&\eql{Prop. \ref{prop:star=refl}}&
\lower12pt\hbox{$\includegraphics[height=1.1cm]{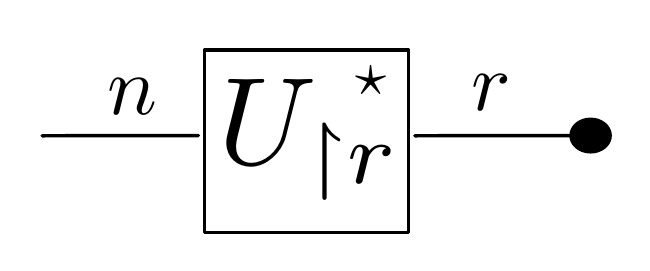}$}
 \end{eqnarray*}
By Proposition~\ref{prop:kernelfromHNF}, the columns of the matrix $\restr{U}{r} \: r \to n$ yield a basis for the kernel of $A$. Thus $\restr{U}{r} \: r \to n$ together with $\finVect \: r \to 0$ is also a pullback span in \eqref{eq:kernelplback} and since $\sem{\ABR}(\finVect \: r \to 0) = \circuitrbcounits$ we know by Lemma~\ref{lemma:mirror} that
\begin{equation*}
\lower12pt\hbox{$\includegraphics[height=1.1cm]{graffles/circuitUrestrrcounits.pdf}$}
=
\lower12pt\hbox{$\includegraphics[height=1.1cm]{graffles/circuitKerAstar.pdf}$} \end{equation*}
which concludes the proof of our statement.
\qed\end{proof} 

%% file: source/6_ProofCompleteness.tex
We now have all the ingredients to provide a proof of our completeness statement, from which the characterization result of Theorem~\ref{th:Span=IBw} follows.

\begin{proof}[Proposition~\ref{prop:IBwComplete}]

Let $A,B,C,D$ be as in the statement of Proposition~\ref{prop:IBwComplete} and consider the following derivation in $\IBRw$:
\begin{eqnarray}\label{eq:dercompl}
\nonumber \lower20pt\hbox{$\includegraphics[height=1.7cm]{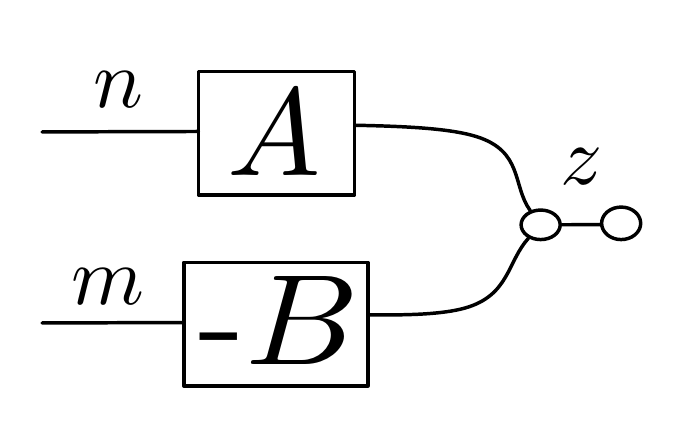}$}
&\eql{Def. $\sem{\ABR}(A|\minus B)$}&
\lower10pt\hbox{$\includegraphics[height=1cm]{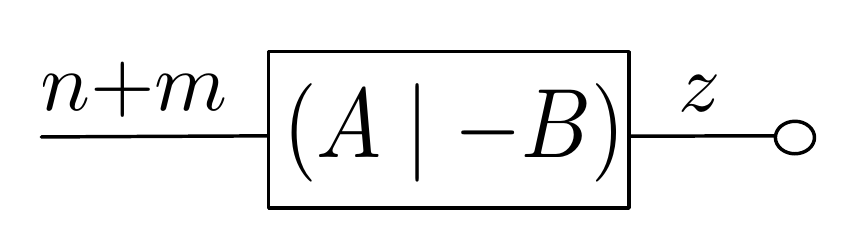}$} \\ \nonumber
&\eql{Prop.~\ref{prop:matrixequalkenrel}}&
\lower10pt\hbox{$\includegraphics[height=1cm]{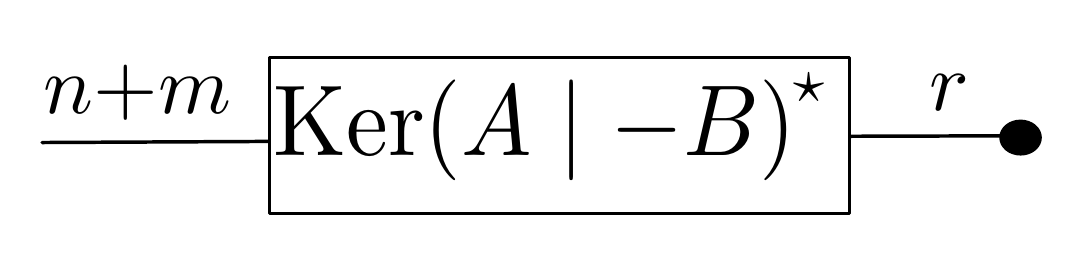}$} \\
&\eql{Lemma~\ref{lemma:pbKernel}}&
\lower15pt\hbox{$\includegraphics[height=1.2cm]{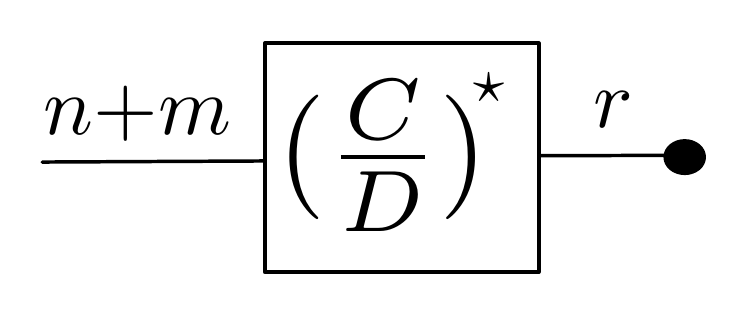}$} \\ \nonumber
&\eql{Def. $\sem{\ABR}(\frac{C}{D})$}&
\lower20pt\hbox{$\includegraphics[height=1.7cm]{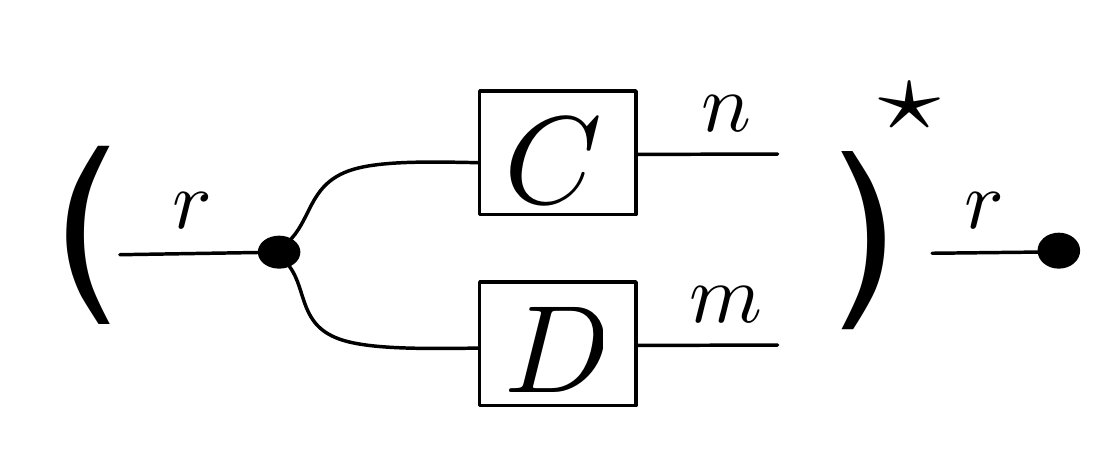}$} \\ \nonumber
&\eql{Prop. \ref{prop:star=refl}}&
\lower19pt\hbox{$\includegraphics[height=1.6cm]{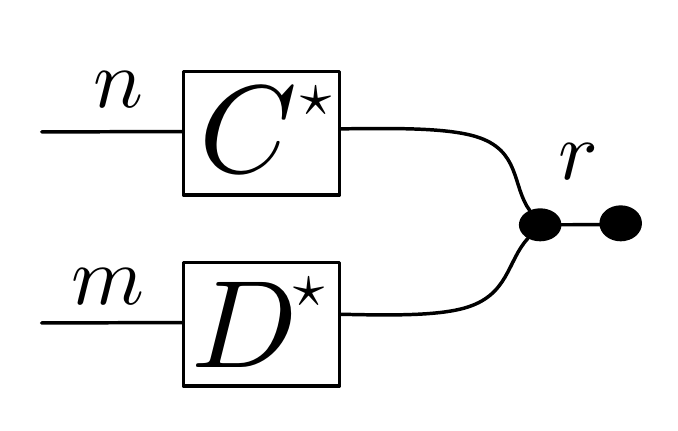}$}
 \end{eqnarray}
The proof is concluded by the following derivation, yielding the desired equation in $\IBRw$.
\begin{eqnarray*}
\lower10pt\hbox{$\includegraphics[height=1cm]{graffles/circuitABstar.pdf}$}
&\eql{Def. $\coc{(\cdot)}$}&
\lower18pt\hbox{$\includegraphics[height=2cm]{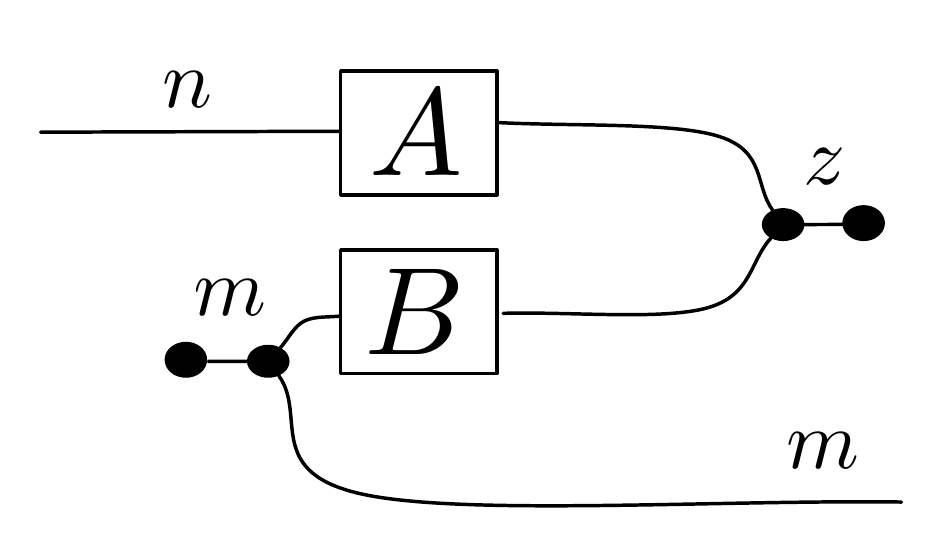}$} \\
&\eql{\eqref{eq:lccb}}&
\lower18pt\hbox{$\includegraphics[height=2cm]{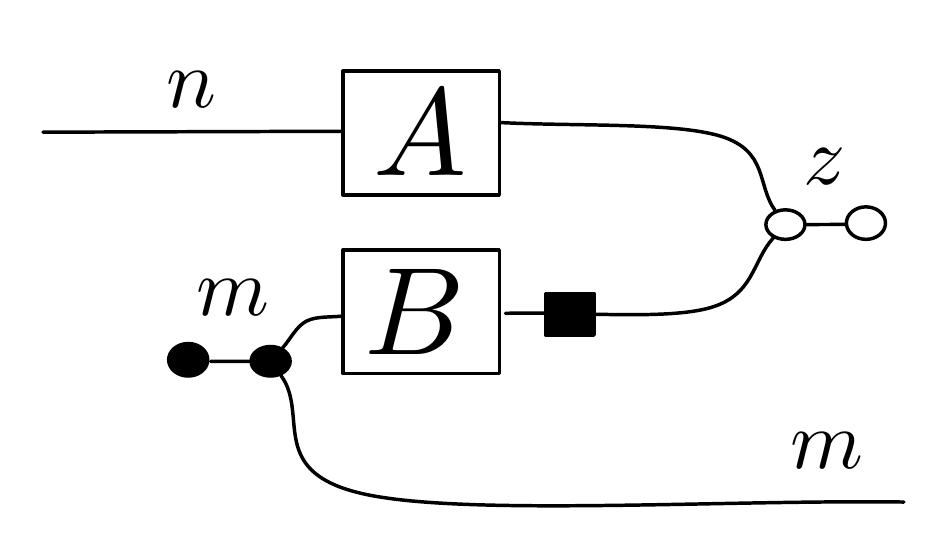}$} \\
&\eql{\eqref{eq:scalarwmult},\eqref{eq:scalarmult}}&
\lower18pt\hbox{$\includegraphics[height=2cm]{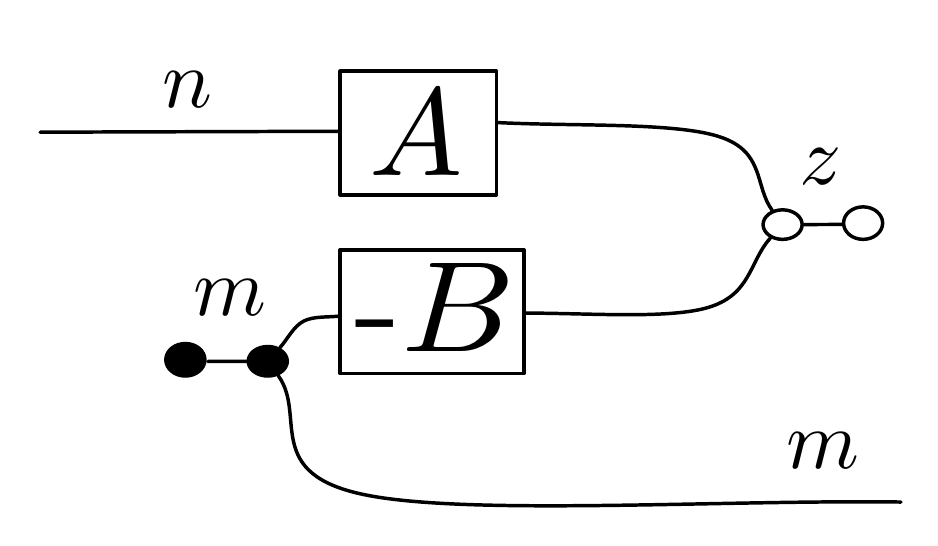}$} \\
&\eql{\eqref{eq:dercompl}}&
\lower18pt\hbox{$\includegraphics[height=2cm]{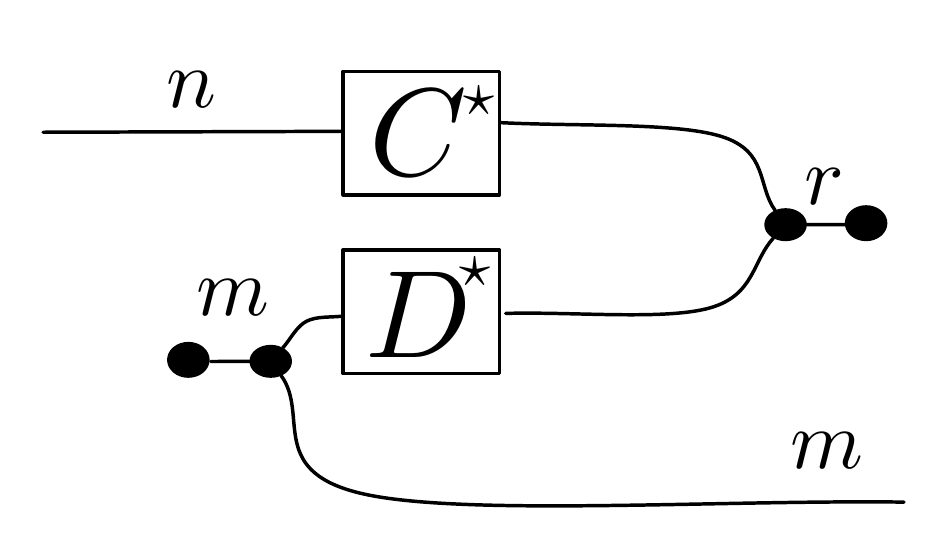}$} \\
&\eql{Def. $\coc{(\cdot)}$}&
\lower10pt\hbox{$\includegraphics[height=1cm]{graffles/circuitCstarDr.pdf}$}
 \end{eqnarray*}
We detail the various derivation steps. First, we can ``bend'' our circuit using the compact-closed structure $\coc{(\cdot)}$. Then we iteratively apply equation~\eqref{eq:lccb} to turn the rightmost part of the compact-closed structure from black into white. This produces $z$ copies of the antipode \antipodesquare. The third equality is given by iteratively applying axiom~\eqref{eq:scalarwmult} to push the antipodes in front of each scalar in circuit $B$, and then multiply all those scalars by the antipode value $-1$ using axiom~\eqref{eq:scalarmult}. As a result, we obtain the (circuit representing) the matrix $-B$. Then we can easily conclude using derivation~\eqref{eq:dercompl}.
\qed\end{proof}

This concludes the proof of Theorem~\ref{th:Span=IBw}. As an immediate consequence, we obtain the following factorisation property.

\begin{corollary}\label{cor:factorisationIBRw} Let $c \in \IBRw[n,m]$ be a circuit. Then $c = \sigma_2(c_1);\sigma_1(c_2)$ with $c_1 \in \ABRop[n,z]$ and $c_2 \in \ABR[z,m]$ for some natural number $z$.
\end{corollary}

%% file: source/9_cubeibbface.tex
In this section we provide a circuit characterization of $\Cospan{\VectR}$. Since we already have such a result for $\Span{\VectR}$, and $\VectR$ is self-dual by matrix transpose, then our strategy will be to understand the transpose in terms of circuits, as this will give ``for free'' also the syntactic PROP of $\Cospan{\VectR}$. We begin with the presentation of $\Cospan{\VectR}$.

\begin{definition}\label{def:IBRb} The PROP $\IBRb$ is the quotient of $\ABR + \ABRop$ by the following equations, for $k$ any element and $l$ any non-zero element of $\PID$.
\begin{multicols}{2}\noindent
 \begin{equation}
\label{eq:lcmop}
\tag{B1}
\lower7.5pt\hbox{$\includegraphics[height=.75cm]{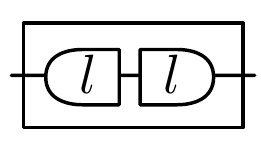}$}
=
\lower6pt\hbox{$\includegraphics[height=.6cm]{graffles/idcircuit.pdf}$}
\end{equation}
\begin{equation}
\label{eq:bbone}
\tag{B2}
\lower6pt\hbox{$\includegraphics[height=.6cm]{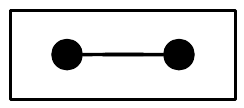}$}
=  \lower4pt\hbox{$\includegraphics[height=.5cm]{graffles/idzerocircuit.pdf}$}
\end{equation}
\end{multicols}
 \begin{multicols}{2}\noindent
\begin{equation}\tag{B3}
\lower15pt\hbox{$\includegraphics[height=1.3cm]{graffles/BFrobS.pdf}$}
\!\!
=
\!\!
\lower11pt\hbox{$\includegraphics[height=1cm]{graffles/BFrobX.pdf}$}
\!\!
=
\!\!
\lower15pt\hbox{$\includegraphics[height=1.3cm]{graffles/BFrobZ.pdf}$}
\end{equation}
\begin{equation}\tag{B4}
\lower15pt\hbox{$\includegraphics[height=1.3cm]{graffles/WFrobS.pdf}$}
\!\!
=
\!\!
\lower11pt\hbox{$\includegraphics[height=1cm]{graffles/WFrobX.pdf}$}
\!\!
=
\!\!
\lower15pt\hbox{$\includegraphics[height=1.3cm]{graffles/WFrobZ.pdf}$}
\end{equation}
\end{multicols}
 \begin{multicols}{2}\noindent
\begin{equation}\tag{B5}
\lower7pt\hbox{$\includegraphics[height=.7cm]{graffles/bccr.pdf}$}
=
\!\!
\lower11pt\hbox{$\includegraphics[height=1cm]{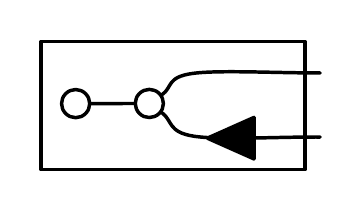}$}
\end{equation}
\begin{equation}\tag{B6}
\lower7pt\hbox{$\includegraphics[height=.7cm]{graffles/bccl.pdf}$}
=
\!\!
\lower11pt\hbox{$\includegraphics[height=1cm]{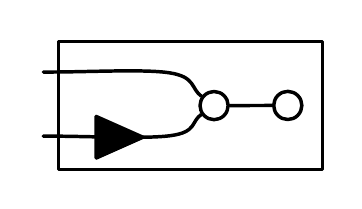}$}
\end{equation}
\end{multicols}
\begin{multicols}{2}\noindent
\begin{equation}\tag{B7}
\label{eq:WcccoscalarAxiomOne}
\lower9pt\hbox{$\includegraphics[height=.8cm]{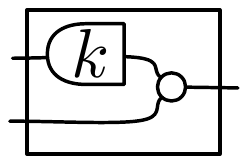}$}
=
\lower9pt\hbox{$\includegraphics[height=.8cm]{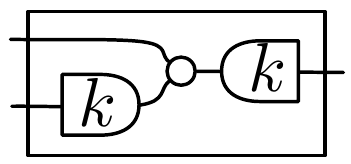}$}
\end{equation}
\begin{equation}\tag{B8}
\label{eq:WcccoscalarAxiomTwo}
\lower9pt\hbox{$\includegraphics[height=.8cm]{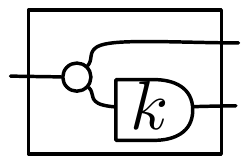}$}
=
\lower9pt\hbox{$\includegraphics[height=.8cm]{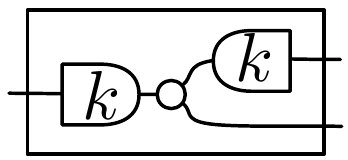}$}
\end{equation}
\end{multicols}
\end{definition}

Similarly to the case of $\IBRw$, we write $\tau_1 \: \ABR \to \IBRb$ and $\tau_2 \: \ABRop \to \IBRb$ for the PROP morphisms interpreting circuits of $\ABR$ and $\ABRop$, respectively, as circuits of $\IBRb$.

\medskip

The axioms of $\IBRb$ are the \emph{photographic negative} of the ones of $\IBRw$, that is, they are the same modulo swapping the black and white colors (and the orientation of scalar circuits). More formally, we inductively define a PROP morphism $\pn \: \IBRb \to \IBRw$ by the following mapping.
  \begin{multicols}{4}
\noindent
      \begin{eqnarray*}
     \Bcounit \mapsto \Wcounit
    \end{eqnarray*}
   \begin{eqnarray*}
  \hspace{-.7cm}   \Bunit \mapsto \Wunit
    \end{eqnarray*}
  \begin{eqnarray*}
   \hspace{-.7cm}  \Wunit \mapsto \Bunit
    \end{eqnarray*}
  \begin{eqnarray*}
   \Wcounit \mapsto \Bcounit
   \end{eqnarray*}
    \end{multicols}
     \begin{multicols}{4}
     \noindent
     \begin{eqnarray*}
     \Wmult \mapsto \Bmult
    \end{eqnarray*}
  \begin{eqnarray*}
  \hspace{-.7cm} \Wcomult \mapsto \Bcomult
   \end{eqnarray*}
        \begin{eqnarray*}
  \hspace{-.7cm}   \Bmult \mapsto \Wmult
    \end{eqnarray*}
  \begin{eqnarray*}
   \Bcomult \mapsto \Wcomult
   \end{eqnarray*}
      \end{multicols}
     \begin{multicols}{4}
     \noindent
           \begin{eqnarray*}
     \scalar \mapsto \coscalar
    \end{eqnarray*}
  \begin{eqnarray*}
 \hspace{-.7cm}  \coscalar \mapsto \scalar
   \end{eqnarray*}
     \begin{equation*}
     \lower2pt\hbox{\hspace{-1.4cm}$c\poi c' \mapsto \pn(c)\poi\pn(c')$}
   \end{equation*}
     \begin{equation*}
     \lower2pt\hbox{\hspace{-.7cm}$c\tns c' \mapsto \pn(c)\tns\pn(c')$}
   \end{equation*}
 \end{multicols}

\noindent The next lemma confirms that $\pn$ is well-defined.

\begin{lemma}\label{lemma:negativerespectequations} For all circuits $c,c'$ of $\IBRb$,  $c = c'$ in $\IBRb$ if and only if $\pn(c) = \pn(c')$ in $\IBRw$.
\end{lemma}
\begin{proof} By construction, the equations presenting $\IBRw$ are the image under $\pn$ of the equations presenting $\IBRb$. Thus the statement is also true for all the derived laws of the two theories. \qed \end{proof}

\begin{lemma} $\pn$ is an isomorphism of PROPs. \end{lemma}
\begin{proof} Fullness of $\pn$ is easily verified by induction on $c \in \IBRw$ and faithfulness follows by the ``only if'' direction of Lemma~\ref{lemma:negativerespectequations}.
\qed \end{proof}

We now specify the matrix counterpart of $\pn$. The operation of taking the transpose of a matrix yields a PROP isomorphism $(\cdot)^{T} \: \VectR \cong \VectRop$. 
This also induces a PROP morphism $\tra \: \Span{\VectR} \to \Cospan{\VectR}$ mapping $n \tl{A} z \tr{B} m$ into $n \tr{A^T} z \tl{B^T} m$. To see that this assignment is functorial, observe that pushouts in $\VectR$ --- giving composition in $\Cospan{\Vect}$ --- can be calculated by transposing pullbacks of transposed matrices. Because $(\cdot)^{T}$ is an isomorphism, also $\tra$ is an isomorphism.

We can now obtain an isomorphism between $\IBRb$ and $\Cospan{\Vect}$ as:
\begin{equation}\label{eq:IsoCospan}  \xymatrix{ \IBRb \ar[r]^{\pn}&  \IBRw \ar[r]^-{\cong} & \Span{\Vect} \ar[r]^{\tra} & \Cospan{\Vect}} \text{.}  \end{equation}

%% file: source/10_cubetopface.tex
In this section we give the presentation $\IBR$ of the PROP of linear relations, obtained by merging theories $\IBRw$ and $\IBRb$. Since we want to identify the generators of $\ABR+\ABRop$ on which both $\IBRw$ and $\IBRb$ are based, we formally define it as the following pushout in $\PROP$.
 \begin{equation}\label{eq:topface}
 \tag{Top}
\vcenter{
\xymatrix@=20pt{ & \ar[dl]_{[\sigma_1,\sigma_2]} \ABR+\ABRop \ar[dd]|{[\varphi_1,\varphi_2]} \ar[dr]^{[\tau_1,\tau_2]} & \\
 \IBRw \ar[dr]_{\Theta} & & \ar[dl]^{\Lambda} \IBRb \\
 & \IBR &
 }
 }
 \end{equation}

The PROP morphism $\Theta$ quotients $\IBRw$ by the equations of $\IBRb$ and $\Lambda$ quotients $\IBRb$ by the ones of $\IBRw$. Then $[\varphi_1,\varphi_2] \: \ABR + \ABRop \to \IBR$ is defined by commutativity of the diagram.  We can give a presentation of the resulting theory $\IBR$ as follows.

\begin{definition}\label{def:IBR}  The PROP $\IBR$ is the quotient of $\ABR + \ABRop$ by the following equations, for $l$ any non-zero element of $\PID$.
  \begin{multicols}{2}\noindent
 \begin{equation}\tag{I1} \label{eq:lcmIH}
\lower8pt\hbox{$\includegraphics[height=.75cm]{graffles/lcml_l.pdf}$}
=
\lower6pt\hbox{$\includegraphics[height=.6cm]{graffles/idcircuit.pdf}$}
\end{equation}
 \begin{equation}\tag{I2} \label{eq:lcmopIH}
\lower8pt\hbox{$\includegraphics[height=.75cm]{graffles/lcmopl_l.pdf}$}
=
\lower6pt\hbox{$\includegraphics[height=.6cm]{graffles/idcircuit.pdf}$}
\end{equation}
\end{multicols}
 \begin{multicols}{2}\noindent
\begin{equation}\tag{I3}\label{eq:WFrobIBR}
\lower15pt\hbox{$\includegraphics[height=1.3cm]{graffles/WFrobS.pdf}$}
\!\!
=
\!\!
\lower11pt\hbox{$\includegraphics[height=1cm]{graffles/WFrobX.pdf}$}
\!\!
=
\!\!
\lower15pt\hbox{$\includegraphics[height=1.3cm]{graffles/WFrobZ.pdf}$}
\end{equation}
\begin{equation}\tag{I4}\label{eq:BFrobIBR}
\lower15pt\hbox{$\includegraphics[height=1.3cm]{graffles/BFrobS.pdf}$}
\!\!
=
\!\!
\lower11pt\hbox{$\includegraphics[height=1cm]{graffles/BFrobX.pdf}$}
\!\!
=
\!\!
\lower15pt\hbox{$\includegraphics[height=1.3cm]{graffles/BFrobZ.pdf}$}
\end{equation}
\end{multicols}
 \begin{multicols}{2}\noindent
\begin{equation}\tag{I5}\label{eq:lccIBR}
\lower12pt\hbox{$\includegraphics[height=1cm]{graffles/lccr.pdf}$}
\!\!
=
\!\!
\lower12pt\hbox{$\includegraphics[height=1cm]{graffles/blackcceta.pdf}$}
\end{equation}
\begin{equation}\tag{I6}\label{eq:rccIBR}
\lower12pt\hbox{$\includegraphics[height=1cm]{graffles/rccl.pdf}$}
\!\!
=
\!\!
\lower12pt\hbox{$\includegraphics[height=1cm]{graffles/blackccepsilon.pdf}$}
\end{equation}
\end{multicols}
\begin{multicols}{2}\noindent
\begin{equation}\tag{I7}\label{eq:WSepIBR}
\lower7pt\hbox{$\includegraphics[height=.7cm]{graffles/WSep.pdf}$}
=
\lower6pt\hbox{$\includegraphics[height=.6cm]{graffles/idcircuit.pdf}$}
\end{equation}
\begin{equation}\tag{I8}\label{eq:BSepIBR}
\lower7pt\hbox{$\includegraphics[height=.7cm]{graffles/BSep.pdf}$}
=
\lower6pt\hbox{$\includegraphics[height=.6cm]{graffles/idcircuit.pdf}$}
\end{equation}
\end{multicols}
\end{definition}
One can readily verify that the axioms above suffice to present the theory resulting from the pushout \eqref{eq:topface}. In particular, the missing equations from the presentations of $\IBRw$ --- \eqref{eq:wbone}, \eqref{eq:BccscalarAxiomOne} and \eqref{eq:BccscalarAxiomTwo} --- and of $\IBRb$ --- \eqref{eq:bbone}, \eqref{eq:WcccoscalarAxiomOne} and \eqref{eq:WcccoscalarAxiomTwo} --- are all derivable from \eqref{eq:lcmIH}-\eqref{eq:BSepIBR} (see~\ref{AppDerLawsIH}).

By definition, $\IBR$ is both a quotient of $\IBRw \cong \Span{\ABR}$ and of $\IBRb\cong \Cospan{\ABR}$. Therefore, it inherits their factorisation property.
\begin{theorem}[Factorisation of $\IBR$] \label{Th:factIBR} Let $c \in \IBR[n,m]$ be a circuit.
 \begin{enumerate}[(i)]
   \item There exist $c_1$ in $\ABRop$ and $c_2$ in $\ABR$ such that $c = \varphi_2(c_1) \poi\varphi_1(c_2)$.
   \item  There exist $c_3$ in $\ABR$ and $c_4$ in $\ABRop$ such that $c = \varphi_1(c_3) \poi \varphi_2(c_4)$.
 \end{enumerate}
\end{theorem}
\begin{proof} The first statement follows by Corollary \ref{cor:factorisationIBRw}. Since $\IBRb$ has been shown to be isomorphic to $\Cospan{\ABR}$, then a result analogous to Corollary \ref{cor:factorisationIBRw} also holds for $\IBRb$, yielding the second statement.
\qed\end{proof}

\begin{remark} In the case in which the PID under consideration is actually a field, we can replace \eqref{eq:lcmIH} and \eqref{eq:lcmopIH} by the axiom $\lower10pt\hbox{$\includegraphics[height=1cm]{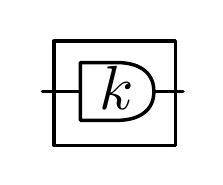}$} \!\! = \!\! \lower7pt\hbox{$\includegraphics[height=.8cm]{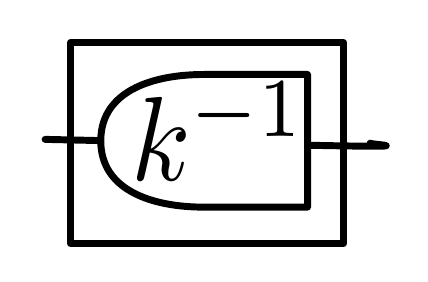}$}$ , for $k \neq 0$.
\end{remark}
The PROP $\SVR$ of linear relations over a field $\frPID$ is defined as follows:
\begin{itemize}
\item arrows $n\to m$ are subspaces of $\frPID^n \times \frPID^m$, considered as a $\frPID$-vector space
\item composition is relational: given $V \: n \to z$, $W \: z \to m$,
\[
(\xx,\zz)\in V\poi W \quad\Leftrightarrow\quad \exists \yy.\; (\xx,\yy)\in V \wedge (\yy,\zz)\in W
\]
\item the monoidal product is given by direct sum
\item the symmetry $n + m \to m+n$ is the subspace $\{(\matrixTwoROneC{\xx}{\yy},\matrixTwoROneC{\yy}{\xx}) \mid \xx \in \frPID^n \wedge {\yy \in \frPID^m}\}$.
\end{itemize}

We can now state our main result.
\begin{theorem}\label{th:IBR=SVR} Let $\frPID$ be the field of fractions of $\PID$. Then $\IBR \cong \SVR$.
\end{theorem}
The proof consists of the construction of the cube~\eqref{eq:cube} shown in the Introduction. We already noted that the top face~\eqref{eq:topface} is a pushout. We next prove that the bottom face is also a pushout (Section~\ref{sec:cubebottom}). Then, in Section~\ref{sec:cubeback}, we show commutativity of the rear faces, whose vertical arrows are isomorphisms.  The isomorphism $\IBR \to \SVR$ will then be given by universal properties of the top and bottom faces (Section~\ref{sec:cuberebuilt}).

%% file: source/cube2.tex
In this section we show that the following diagram, which is the bottom face of the cube \eqref{eq:cube}, is a pushout in $\PROP$.
\begin{equation}\tag{Bot}
\label{eq:bottomface}
\raise15pt\hbox{$
\xymatrix@C=40pt{
{\VectR + \VectRop} \ar[r]^-{[\kappa_1,\,\kappa_2]} \ar[d]_{[\iota_1,\,\iota_2]} & {\Span{\Mat \PID}} \ar[d]^{\Phi} \\
{\Cospan{\Mat \PID}} \ar[r]_-{\Psi} & {\SVR}
}$}
\end{equation}
In the diagram above, we define
\[\kappa_1(A \: n\to m) = (n \tl{\id} n \tr{A} m),\  \kappa_2(A \: n\to m) = (n \tl{A} m \tr{\id} m ),\]
\[\iota_1(A \: n\to m) = (n \tr{A} m \tl{\id} m)\text{ and }
\iota_2(A \: n\to m) = (n \tr{\id} n \tl{A} m).\]
For the definition of $\Phi$, we let $\Phi(n\tl{A} z \tr{B} m)$ be the subspace
\[
\{\,(\xx,\yy)\ |\ \xx\in \frPID^n,\, \yy\in \frPID^m,\, \exists \zz\in \frPID^z.\; A\zz=\xx \wedge B\zz = \yy \,\}.
\]
Instead, $\Psi(n \tr{A} z \tl{B} m)$ is defined to be the subspace \[
\{\, (\xx,\yy) \ | \ \xx\in \frPID^n,\, \yy\in \frPID^m,\ A\xx = B\yy\, \}
\]
In the sequel we verify that $\Phi$ and $\Psi$ are indeed functorial assignments. This requires some preliminary work. Let $\RMod{\PID}$ be the category of finite-dimensional $\PID$-modules and linear maps. We define $\RMod{\frPID}$ analogously. $\FRMod{\PID}$ and $\FRMod{\frPID}$ are the full subcategories of free modules of $\RMod{\PID}$ and $\RMod{\frPID}$ respectively (note that, of course, $\FRMod{\frPID} \cong \RMod{\frPID}$). There is an obvious PROP morphism $I \: \VectR \to \Mat{\frPID}$ interpreting a matrix with entries in $\PID$ as one with entries in $\frPID$. Similarly, we have an inclusion $J \: \FRMod{\PID} \to \FRMod{\frPID}$. This yields the following commutative diagram, where $\simeq$ denotes equivalence.
\[
\xymatrix{
{\Mat{\PID}} \ar[d]_I \ar[r]^-{\simeq} & {\FRMod{\PID}} \ar[d]^J \\
{\Mat{\frPID}} \ar[r]_-{\simeq} & {\FRMod{\frPID}}
}
\]

\begin{lemma} \label{lem:pullbacksinmat} $I: \Mat{\PID}\to \Mat{\frPID}$ preserves pullbacks and pushouts.
\end{lemma}
\begin{proof} Because the transpose operation induces a duality in both $\Mat{\PID}$ and $\Mat{\frPID}$, the morphism $\Mat{\PID}\to\Mat{\frPID}$ preserves pullbacks iff it preserves pushouts. It is thus enough to show that it preserves pullbacks. This can be easily be proved directly as follows. Suppose that the diagram
\begin{equation}
\label{eq:pbinmatr}
\tag{$\star$}
\vcenter{
\xymatrix@=15pt{
{r} \ar[d]_A \ar[r]^B& {m}\ar[d]^D\\
{n} \ar[r]_C  & {z}
}
}
\end{equation}
is a pullback in $\Mat{\PID}$. We need to show that it is also a pullback in $\Mat{\frPID}$.
Suppose that, for some $P \: q\to n$, $Q \: q\to m$ in $\Mat{\frPID}$ we have that
$CP = DQ$ in $\Mat{\frPID}$. Since $\PID$ is a PID we can find least common multiples:
thus let $d$ be a common multiple of all the denominators that appear in $P$ and $Q$.
Then $dP \: q \to n$, $dQ \: q\to m$ are in $\Mat{\PID}$ and we have
$C(dP)=d(CP)=d(DQ)=D(dQ)$. Since \eqref{eq:pbinmatr} is a pullback in $\Mat{\PID}$, there exists a unique $H \: q\to r$ with $AH=dP$ and $BH=dQ$. This means that we have found a mediating arrow, $H/d \: q \to r$, in $\Mat{\frPID}$ since $A(H/d)=AH/d=dP/d=P$ and similarly $B(H/d)=Q$. Uniqueness in $\Mat{\frPID}$ can also be translated in a straightforward way to uniqueness in $\Mat{\PID}$. Basically if $H'$ is another mediating morphism and $d'$ is the least common multiple of denominators in $H'$ then we must have $d'(H/d)=d'H'$ because of the universal property in $\Mat{\PID}$. Dividing both sides by $d'$ yields the required equality.
\qed\end{proof}

We are now able to show that
\begin{lemma}
$\Phi\: \Span{\VectR} \to \SVR$ is a PROP morphism.	
\end{lemma}
\begin{proof}
We must verify that $\Phi$ preserves composition. In the diagram below let the centre square be a pullback diagram in $\VectR$.
\[
\xymatrix{
& & {r} \ar[dl]_{F_2'} \ar[dr]^{G_1'} \\
& {z_1} \ar[dl]_{F_1} \ar[dr]^{G_1} & & {z_2} \ar[dl]_{F_2} \ar[dr]^{G_2} \\
{n} & & {z} & & {m}
}
\]
By definition of composition in $\Span{\Mat{\PID}}$, $(\tl{F1} \tr{G_1}) \poi (\tl{F_2}\tr{G_2})
=\ \tl{F_1F_2'}\tr{G_2G_1'}$.

Now, by definition, if $(\xx,\zz)\in \Phi(\tl{F_1F_2'}\tr{G_2G_1'})$
then there exist $\ww$ with $\xx= F_1F_2'\ww$ and $\zz=G_2G_1'\ww$.
Therefore
$(\xx,\zz)\in {\Phi(\tl{F1} \tr{G_1})}\poi
{\Phi(\tl{F_2}\tr{G_2})}$ by commutativity of the square.

Conversely, if $(\xx,\zz)\in \Phi(\tl{F_1} \tr{G_1})\poi
\Phi(\tl{F_2}\tr{G_2})$ then
for some $\yy$ we must have
$(\xx,\yy)\in \Phi(\tl{F_1} \tr{G_1})$ and
$(\yy,\zz)\in \Phi(\tl{F_2}\tr{G_2})$.
Thus there exists $\uu$ with
$\xx=F_1\uu$ and $\yy=G_1\uu$
and there exists $\vv$ with $\yy=F_2\vv$
and $\zz=G_2\vv$.
By Lemma~\ref{lem:pullbacksinmat}, the square is also a pullback in $\Mat{\frPID}$ and then it translates to a pullback diagram in $\VectSpFr$. It follows the existence of $\ww$ with $F_2'\ww=\uu$
and $G_1'\ww=\vv$: thus $(\xx,\zz)\in \Phi((\tl{F1} \tr{G_1})\poi (\tl{F_2}\tr{G_2}))$.
This completes the proof.
\qed\end{proof}

The proof that also $\Psi$ is a functor will rely on the following lemma.

\begin{lemma}\label{lemma:pushoutinvect}
Let the following be a pushout diagram in $\FRMod{\frPID}$.
\[
\xymatrix@=15pt{
 {U} \ar[d]_f \ar[r]^g &  {W} \ar[d]^q\\
{V}\ar[r]_p & {T}
}
\]	
Suppose that there exist $\vv\in V$, $\ww\in W$ such that $p\vv=q\ww$.
Then there exists $\uu\in U$ with $f\uu=\vv$ and $g\uu=\ww$.
\end{lemma}
\begin{proof}
Pushouts in $\FRMod{\frPID} \cong \RMod{\frPID}$ can be constructed by quotienting the vector space $V+W$ by the subspace generated by $\{\,(f\uu,g\uu)\,|\,\uu\in U\,\}$. Thus,
if $p(\vv) = q(\ww)$ then there exists a chain $\uu_1,\uu_2,\dots,\uu_k$
with $f(\uu_1)=\vv$, $g(\uu_1)=g(\uu_2)$, $f(\uu_2)=f(\uu_3)$, \dots, $f(\uu_{k-1})=f(\uu_{k-1})$ and $g(\uu_k)=\ww$.
If $k=1$ then we are finished. Otherwise, to construct an inductive argument we need to consider a chain $\uu_1, \uu_2, \uu_3$ with $f(\uu_1)=\vv$,
$g(\uu_1)=g(\uu_2)$, $f(\uu_2)=f(\uu_3)$ and $g(\uu_3)=\ww$.
Now $f(\uu_1-\uu_2+\uu_3)=f(\uu_1)-f(\uu_2)+f(\uu_3)=\vv$ and
$g(\uu_1-\uu_2+\uu_3)=g(\uu_1)-g(\uu_2)+g(\uu_3)=\ww$, so we have reduced the size of the chain to one.
\qed\end{proof}

\begin{lemma}
$\Psi\: \Cospan{\VectR} \to \SVR$ is a PROP morphism.
\end{lemma}
\begin{proof}
We must verify that $\Psi$ preserves composition. Let the square in the diagram below be a pushout in $\Mat{\PID}$. By definition of composition in $\Cospan{\Mat{\PID}}$
we have $(\tr{P_1}\tl{Q_1})\poi(\tr{P_2}\tl{Q_2})
=\ \tr{R_1P_1}\tl{R_2Q_2}$.
\[
\xymatrix{
{n} \ar[dr]_{P_1} & & {z} \ar[dl]^{Q_1} \ar[dr]_{P_2} & & {m} \ar[dl]^{Q_2} \\
& {z_1} \ar[dr]_{R_1} & & {z_2}\ar[dl]^{R_2} \\
& & {r}
}
\]
Consider $(\xx,\zz)\in \Psi(\tr{R_1P_1} \tl{R_2Q_2} )$.
Then $R_1P_1 \xx  = R_2 Q_2 \zz = \yy \in \frPID^r$.
Since the pushout diagram maps to a pushout diagram in $\VectSpFr$, we can use the conclusions of Lemma~\ref{lemma:pushoutinvect} to obtain $\yy\in \frPID^{z}$ such that $Q_1\yy=P_1\xx$ and $P_2\yy=Q_2\zz$. In other words, we have
$(\xx,\yy)\in\Psi(\tr{P_1}\tl{Q_1})$
and
$(\yy,\zz)\in\Psi(\tr{P_2}\tl{Q_1})$, meaning that $(\xx,\zz) \in \Psi(\tr{P_1}\tl{Q_1}) \poi\Psi(\tr{P_2}\tl{Q_1})$.

Conversely if $(\xx,\zz)\in\Psi(\tr{P_1}\tl{Q_1})\poi\Psi(\tr{P_2}\tl{Q_2})$ then $\exists \yy\in\frPID^{z}$ such that $(\xx,\yy)\in\Psi(\tr{P_1}\tl{Q_1})$
and $(\yy,\zz)\in\Psi(\tr{P_2}\tl{Q_2})$. It follows that $R_1P_1\xx=R_1Q_1\yy=R_2P_2\yy=R_2Q_2\zz$ and thus $(\xx,\zz)\in \Psi(\tr{R_1P_1} \tl{R_2Q_2} )$ as required.
\qed\end{proof}

\begin{remark} The proof of Lemma~\ref{lemma:pushoutinvect} relies on the fact that, for $\frPID$ a field, pushouts in $\FRMod{\frPID}$ coincide with those in $\RMod{\frPID}$. It would not work for an arbitrary PID $\PID$: $\FRMod{\PID}$ has pushouts for purely formal reasons, because it has pullbacks and is self-dual. However, differently from pullbacks (for which one can use, as we do in Section~\ref{sec:completeness}, that submodules of a free $\PID$-module are free), pushouts generally do not coincide with those calculated in $\RMod{\PID}$. This asymmetry is the reason why proving functoriality of $\Psi$ requires more work than for $\Phi$.
\end{remark}
We now verify some properties of~\eqref{eq:bottomface}.
\begin{lemma}
\label{lemma:bottomfacecommutes}
\eqref{eq:bottomface} commutes.
\end{lemma}
\begin{proof}
It suffices to show that it commutes on the two injections into
$\VectR + \VectRop$. This means that we have to show, for any
$A \: n\to m$ in $\Mat{\PID}$, that
\[
\Phi(\tl{\id}\tr{A}) = \Psi(\tr{A}\tl{\id})
\]
and
\[
\Phi(\tl{A}\tr{\id}) =
\Psi(\tr{\id}\tl{A}).
\]
These are clearly symmetric, so it is enough to check one.
But this follows directly from the definition of $\Phi$ and $\Psi$:
\[
\Phi(\tl{\id}\tr{A}) =
\{\, (\xx,\yy) \,|\, A \xx = \yy \,\} = \Psi(\tr{A}\tl{\id})
\]
\qed\end{proof}

%
%

\begin{lemma}
\label{lemma:arbitraryPROP}
Given an arbitrary PROP $\mathbb{X}$ and a commutative diagram
\begin{equation}
\label{eq:arbitrary}
\tag{$\dag$}
\raise15pt\hbox{$
\xymatrix@C=40pt{
{\VectR + \VectRop} \ar[r]^-{[\kappa_1,\,\kappa_2]} \ar[d]_{[\iota_1,\,\iota_2]} & {\Span{\Mat \PID}} \ar[d]^{\Gamma} \\
{\Cospan{\Mat \PID}} \ar[r]_-{\Delta} & {\mathbb{X}}
}$}
\end{equation}
consider the following diagram in $\Mat\PID$:
\begin{equation*}
\label{eq:square}
\tag{$\star$}
\raise10pt\hbox{$
\xymatrix{
{} \ar[r]^{G} \ar[d]_{F} & {} \ar[d]^Q \\
{} \ar[r]_{P} & {}
}$}
\end{equation*}
\begin{enumerate}[(i)]
\item if \eqref{eq:square} is a pushout diagram then
$\Gamma(\tl{F}\tr{G})=\Delta(\tr{P}\tl{Q})$.
\item if \eqref{eq:square} is a pullback diagram then
$\Gamma(\tl{F}\tr{G})=\Delta(\tr{P}\tl{Q})$.
\item if $\tl{F_1}\tr{G_1}$ and
$\tl{F_2}\tr{G_2}$ have the same pushout cospan
in $\Mat{\PID}$ then
$\Gamma(\tl{F_1}\tr{G_1})=\Gamma(\tl{F_2}\tr{G_2})$.
\item if $\tr{P_1}\tl{Q_1}$ and
$\tr{P_2}\tl{Q_2}$ have
the same pullback span in $\Mat{\PID}$ then
${\Delta(\tr{P_1}\tl{Q_1})}=
\Delta(\tr{P_2}\tl{Q_2})$.
\end{enumerate}
\end{lemma}
\begin{proof}~
\begin{enumerate}[(i)]
\item Suppose that $\tr{P}\tl{Q}$ is the cospan obtained by pushing out
$\tl{F}\tr{G}$ in $\Mat\PID$. Then
\begin{align*}
\Gamma(\tl{F}\tr{G}) &= \Gamma(\kappa_2 F \poi \kappa_1 G) \\
&= \Gamma(\kappa_2 F )\poi \Gamma(\kappa_1 G) \\
&= \Delta(\iota_2 F)\poi \Delta(\iota_1 G) \\
&= \Delta( \iota_2 F \poi \iota_1 G) \\
&= \Delta( \tr{P}\tl{Q} ).
\end{align*}
\item Suppose that $\tl{F}\tr{G}$ is the span obtained by pulling back
$\tr{P}\tl{Q}$. Then, reasoning in a similar way to (i), we get $\Delta(\tr{P}\tl{Q})= \Gamma(\tl{F}\tr{G})$.
\item Suppose that $\tr{P}\tl{Q}$ is the cospan obtained by pushing out
$\tl{F_1}\tr{G_1}$ and $\tl{F_2}\tr{G_2}$.
Using (i) we get $\Gamma(\tl{F_1}\tr{G_1})=\Delta(\tr{P}\tl{Q})=\Gamma(\tl{F_2}\tr{G_2})$.
\item The proof of (iv) is similar and uses (ii). \qed
\end{enumerate}
\end{proof}

\begin{lemma}
\label{lemma:pullbackpsi}
The following are equivalent:
\begin{enumerate}[(i)]
\item $n \tr{P_1} z_1 \tl{Q_1} m$ and $n \tr{P_2} z_2 \tl{Q_2}m$ have the same pullback in $\Mat{\PID}$.
\item $\Psi(\tr{P_1}\tl{Q_1})=
\Psi(\tr{P_2}\tl{Q_2})$.
\end{enumerate}
\end{lemma}
\begin{proof}
The conclusions of
Lemmas~\ref{lemma:bottomfacecommutes} and~\ref{lemma:arbitraryPROP}
give that (i) $\Rightarrow$ (ii). It thus suffices to show that
(ii) $\Rightarrow$ (i). Indeed, suppose that
$\Psi(\tr{P_1}\tl{Q_1})=
\Psi(\tr{P_2}\tl{Q_2})$. In particular
on elements $\xx\in\PID^n$, $\yy\in\PID^m$ we have $(\star)$
$P_1\xx=Q_1\yy$ if and only if $P_2\xx = Q_2\yy$. Compute the following
pullbacks in $\Mat{\PID}$:
\[
\xymatrix{
{r_1} \ar[d]_{G_1} \ar[r]^{F_1} & {m} \ar[d]^{Q_1} \\
{n} \ar[r]_{P_1} & {z_1}
}
\qquad
\xymatrix{
{r_2} \ar[d]_{G_2} \ar[r]^{F_2} & {m} \ar[d]^{Q_2} \\
{n} \ar[r]_{P_2} & {z_2}
}
\]
By $(\star)$ we can conclude that $P_1G_2=Q_1F_2$ and $P_2G_1=Q_2F_1$. This, using the universal property of pullbacks, implies that the spans $\tl{G_1}\tr{F_1}$ and
$\tl{G_2}\tr{F_2}$ are isomorphic. \qed
\end{proof}

\begin{lemma}
\label{lemma:pushoutphi}
The following are equivalent:
\begin{enumerate}[(i)]
\item $n \tl{F_1} z_1 \tr{G_1} m$ and
$n \tl{F_2} z_2 \tr{G_2} m$ have the same pushout
in $\Mat{\PID}$
\item $\Phi(\tl{F_1}\tr{G_1})=
\Phi(\tl{F_2}\tr{G_2})$.
\end{enumerate}
\end{lemma}
\begin{proof}
The conclusions of
Lemmas~\ref{lemma:bottomfacecommutes} and~\ref{lemma:arbitraryPROP} again
give us that (i) $\Rightarrow$ (ii). It thus suffices to show that
(ii) $\Rightarrow$ (i).
Assume $\Phi(\tl{F_1}\tr{G_1})=
\Phi(\tl{F_2}\tr{G_2})$.
Compute the following
pushouts in $\Mat{\PID}$:
\[
\xymatrix{
{z_1} \ar[d]_{G_1} \ar[r]^{F_1} & {n} \ar[d]^{Q_1} \\
{m} \ar[r]_{P_1} & {r_1}
}
\qquad
\xymatrix{
{z_2} \ar[d]_{G_2} \ar[r]^{F_2} & {n} \ar[d]^{Q_2} \\
{m} \ar[r]_{P_2} & {r_2}
}
\]
By the conclusion of Lemma~\ref{lemma:arbitraryPROP}, we have
$\Psi(\tr{P_1}\tl{Q_1})=
\Psi(\tr{P_2}\tl{Q_2})$.
Applying the conclusion of Lemma~\ref{lemma:pullbackpsi},
$\tr{P_1}\tl{Q_1}$ and
$\tr{P_2}\tl{Q_2}$ have the same pullback span.
Call this span $\tl{A}\tr{B}$. Then both
$\tr{P_1}\tl{Q_1}$ and
$\tr{P_2}\tl{Q_2}$ are the pushout cospan of
$\tl{A}\tr{B}$, thus they must be isomorphic.
\qed\end{proof}

\begin{lemma}
\label{lemma:phipsifull}
$\Phi \: \Span{\Mat{\PID}}\to\SVR$ and $\Psi \: \Cospan{\Mat{\PID}}\to\SVR$ are both full.	
\end{lemma}
\begin{proof}
Take any subspace
$S \: n\to m$ in $\SVR$. Picking any finite basis (say, of size $r$) for this subspace and multiplying out fractions gives us a finite set of elements in $\PID^{n+m}$.
In the obvious way, this yields
\[
n \tl{S_1} r \tr{S_2} m
\]	
in $\Span{\Mat{\PID}}$ with $\Phi(\tl{S_1}\tr{S_2})=S$. Thus $\Phi$ is full.
Let $\tr{R_1}\tl{R_2}$ be the cospan obtained from pushing out
$\tl{S_1}\tr{S_2}$ in $\Mat{\PID}$. By the conclusion of Lemma~\ref{lemma:arbitraryPROP}, $\Psi(\tr{R_1}\tl{R_2})=\Phi(\tl{S_1}\tr{S_2})=S$, which shows that $\Psi$ is full. \qed\end{proof}

\begin{theorem}
\eqref{eq:bottomface} is a pushout in $\PROP$.
\end{theorem}
\begin{proof}
Suppose that we have a commutative diagram of PROP morphisms as in~\eqref{eq:arbitrary}.
By the conclusions of Lemma~\ref{lemma:phipsifull} it suffices to show that there exists a PROP morphism $\Theta \: \SVR\to\mathbb{X}$ with $\Theta\Phi = \Gamma$ and $\Theta\Psi=\Delta$ -- uniqueness is automatic by fullness of $\Phi$ (or of $\Psi$).

Given a subspace $S \: n\to m$, by Lemma~\ref{lemma:phipsifull} there
exists a span $\tl{S_1}\tr{S_2}$ with
$\Phi(\tl{S_1}\tr{S_2})=S$.
We let $\Theta(S)=\Gamma(\tl{S_1}\tr{S_2})$.
This is well-defined: if $\tl{S_1'}\tr{S_2'}$ is another
span with $\Phi(\tl{S_1'}\tr{S_2'})=S$ then applying the conclusions of Lemma~\ref{lemma:pushoutphi} gives us that $\tl{S_1}\tr{S_2}$ and $\tl{S_1'}\tr{S_2'}$ have the same pushout in $\Mat{\PID}$. Now the conclusions of Lemma~\ref{lemma:arbitraryPROP} give us that $\Gamma(\tl{S_1}\tr{S_2})=\Gamma(\tl{S_1'}\tr{S_2'})$.
This  argument also shows that, generally, $\Theta\Phi=\Gamma$.
Finally, $\Theta$ preserves composition:
\begin{align*}
\Theta(R\poi S) &= \Theta(\Phi(\tl{R_1}\tr{R_2})\poi\Phi(\tl{S_1}\tr{S_2})) \\
&= \Theta(\Phi((\tl{R_1}\tr{R_2})\poi (\tl{S_1}\tr{S_2}))) \\
&= \Gamma((\tl{R_1}\tr{R_2})\poi (\tl{S_1}\tr{S_2})) \\
&= \Gamma(\tl{R_1}\tr{R_2})\poi\Gamma(\tl{S_1}\tr{S_2})\\
&= \Theta(R)\poi\Theta(S).
\end{align*}

It is also easy to show that
$\Theta\Psi = \Delta$: given a cospan $\tr{F}\tl{G}$ let
$\tl{P}\tr{Q}$ be its pullback span in $\Mat\PID$.
Using the conclusions of Lemma~\ref{lemma:arbitraryPROP},
${\Delta(\tr{F}\tl{G})} = \Gamma(\tl{P}\tr{Q})
=\Theta\Phi(\tl{P}\tr{Q})=\Theta\Psi(\tr{F}\tl{G})$.
\qed\end{proof}

\begin{remark} It is interesting to notice that, if one tries to glue in the same way spans and cospans of $\F$ (the PROP of functions, as considered in Section~\ref{sec:background}), the resulting pushout object is the terminal PROP:
$$\xymatrix{
\F + \Fop \ar[r] \ar[d] & \Span{\F}\ar[d] \\
\Cospan{\F} \ar[r] & \mathbbm{1}}$$
Syntactically, this corresponds to the observation that summing the SMTs of bialgebras and of separable Frobenius algebras (defined on the same monoid-comonoid pair) one obtains the trivial theory.
\end{remark}

%% file: source/7_cubebackwardfaces.tex
To complete the proof of Theorem~\ref{th:IBR=SVR}, it remains to show that the rear faces of the cube~\eqref{eq:cube} commute.
\begin{equation}\label{backwardcube}
\tag{Rear}
\raise20pt
\hbox{$
\xymatrix@C=35pt{
{\IBRb} \ar[d]_{\sem{\IBRb}} & {\ABR + \ABRop} \ar[l]_{[\tau_1,\tau_2]} \ar[d]|{\sem{\ABR} + \sem{\ABR}^{\op} } \ar[r]^-{[\sigma_1,\sigma_2]} & {\IBRw} \ar[d]^{\sem{\IBRw}} \\
\Cospan{\VectR} & \VectR + \VectRop \ar[r]_{[\kappa_1,\kappa_2]} \ar[l]^{[\iota_1,\iota_2]} & \Span{\VectR}
}$}
\end{equation}
For this purpose, we give an explicit description of the isomorphisms $\IBRw \to \Span{\VectR}$ and $\IBRb \to \Cospan{\VectR}$, whose existence has been shown in Sections~\ref{sec:completeness}-\ref{sec:IBRbCospan}, in the same inductive way as $\sem{\ABR}$ is defined.

The two isomorphisms are noted in~\eqref{backwardcube} with $\sem{\IBRw}$ and $\sem{\IBRw}$ respectively.
For the definition of $\kappa_1$, $\kappa_2$, $\iota_1$ and $\iota_2$ see the beginning of Section~\ref{sec:cubebottom}.
The PROP morphisms $\sigma_1 \: \ABR \to \IBRw$, $\sigma_2 \: \ABRop \to \IBRw$ and $\tau_1 \: \ABR \to \IBRb$, $\tau_2 \: \ABRop \to \IBRb$ have been introduced by Definition~\ref{def:IBRw} and \ref{def:IBRb} respectively.

\paragraph{An inductive presentation of $\sem{\IBRw}$} The PROP morphism $\sem{\IBRw} \: \IBRw \to \Span{\VectR}$ is defined by induction on circuits of $\IBRw$, where $c \in \Sigma_{\ABR}$ means that $c$ is a generator in the signature of $\ABR$, and similarly for $c \in \Sigma_{\ABRop}$.
\begin{align*}
   c  \mapsto \left\{
	\begin{array}{ll}
        \kappa_1(\sem{\AB}(c')) & \text{ if } c = \sigma_1(c') \text{ and }c'\in \Sigma_{\ABR} \\
        \kappa_2( \sem{\AB}^{\op} (c'))  & \text{ if } c = \sigma_2(c') \text{ and }c'\in \Sigma_{\ABRop} \\
	\sem{\IBRw}(c_1) \poi \sem{\IBRw}(c_2) & \text{ if } c=c_1 \poi c_2\\
	\sem{\IBRw}(c_1) \tns \sem{\IBRw}(c_2) & \text{ if } c=c_1 \tns c_2\\
        \end{array}
\right.
 \end{align*}
 The mapping is well-defined as all the equations of $\IBRw$ are sound w.r.t. $\sem{\IBRw}$. It is clear by definition that $\sem{\IBRw}$ makes the rightmost square in \eqref{backwardcube} commute. It remains to show the following result.

\begin{proposition}\label{prop:semanticsIBRwIso}  $\sem{\IBRw}$ is an isomorphism of PROPs.
\end{proposition}


\begin{proof}
For fullness, let $n \tl{A} z \tr{B} m$ be an arrow in $\Span{\VectR}$. By fullness of $\sem{\ABR}$ there are circuits $c_1 \in \ABR[z,n]$ and $c_2 \in \ABR[z,m]$ such that $\sem{\ABR}(c_1) = A$ and $\sem{\ABR}(c_2) = B$. The following derivation shows that $n \tl{A} z \tr{B} m$ is targeted by $\sigma_2(c_1^{\star}) \poi \sigma_1(c_2) \in \IBR[n,m]$.
\begin{align*}
\sem{\IBRw}(\sigma_2(c_1^{\star}) \poi \sigma_1(c_2)) &= \sem{\IBRw}(\sigma_2(c_1^{\star})) \poi \sem{\IBRw}(\sigma_1(c_2)) \\
&=
\kappa_2(\sem{\ABR}^{op}(c_1^{\star})) \poi \kappa_2(\sem{\ABR}(c_2)) \\
&=
\kappa_2(A \: n \to z) \poi \kappa_2(B \: z \to m) \\
&=
(n \tl{A}z \tr{\id} z) \poi (z \tl{\id}z \tr{B} m) \\
&=
n \tl{A}z \tr{B} m.
 \end{align*}
 It remains to show faithfulness. For this purpose, let $c \in \IBRw[n,m]$ and $c' \in \IBRw[n,m]$ be circuits and suppose that $\sem{\IBRw}(c)  = \sem{\IBRw}(c')$. By Corollary~\ref{cor:factorisationIBRw} it follows that
\begin{eqnarray*}
\sem{\IBRw}(c) = n \tl{\sem{\ABR}(c_1^{\star})} z \tr{\sem{\ABR}(c_2)} m \\
\sem{\IBRw}(c') = n \tl{\sem{\ABR}(c_1'^{\star})} z \tr{\sem{\ABR}(c_2')} m
\end{eqnarray*}
for circuits $c_1, c_1'$ of $\ABRop$ and $c_2,c_2'$ of $\ABR$ such that $c = \sigma_2(c_1);\sigma_1(c_2)$ and $c' = \sigma_2(c_1');\sigma_1(c_2')$. Since $\sem{\IBRw}(c)  = \sem{\IBRw}(c')$ are the same arrow of $\Span{\VectR}$, that means they are isomorphic spans: thus there is an invertible matrix $U \in \VectR[z,z]$ making the following diagram commute.
\begin{eqnarray*}
\vcenter{
\xymatrix@C=25pt@R=15pt{ && \ar[drr]^{\sem{\ABR}(c_2)} z  \ar[dll]_{\sem{\ABR}(c_1^{\star})}&& \\
 n && \ar[ll]^{\sem{\ABR}(c_1'^{\star})} z \ar[u]^>>>>U \ar[rr]_{\sem{\ABR}(c_2)}&& m }
 }
\end{eqnarray*}
Then by Lemma~\ref{lemma:mirror} we have that $c$ and $c'$ are equal as circuits of $\IBRw$.
 \qed
 \end{proof}

%

\paragraph{An inductive presentation of $\sem{\IBRb}$} Similarly to what we did for $\IBRw$, we define a PROP morphism $\sem{\IBRb} \: \IBRb \to \Cospan{\VectR}$ by induction on circuits of $\IBRb$ as follows:
\begin{align*}
   c  \mapsto \left\{
	\begin{array}{ll}
        \iota_1(\sem{\AB}(c')) & \text{ if } c = \tau_1(c') \text{ and }c'\in \Sigma_{\ABR} \\
        \iota_2( \sem{\AB}^{\op} (c'))  & \text{ if } c = \tau_2(c') \text{ and }c'\in \Sigma_{\ABRop} \\
	\sem{\IBRb}(c_1) \poi \sem{\IBRb}(c_2) & \text{ if } c=c_1 \poi c_2\\
	\sem{\IBRb}(c_1) \tns \sem{\IBRb}(c_2) & \text{ if } c=c_1 \tns c_2\\
        \end{array}
\right.
 \end{align*}
 The mapping is well-defined as all the equations of $\IBRb$ are sound w.r.t. $\sem{\IBRb}$. Also, $\sem{\IBRb}$ clearly makes the leftmost part of~\eqref{backwardcube} commute.

\begin{proposition} $\sem{\IBRb}$ is an isomorphism of PROPs. \end{proposition}
\begin{proof} Following~\eqref{eq:IsoCospan}, it suffices to show that $\sem{\IBRb} = \pn \poi \sem{\IBRw} \poi \tra$. This can be easily verified by induction on $c \in \IBRb$. For instance, $\sem{\IBRb}$ maps $\Bmult \: 2 \to 1$ into $2 \tr{\id} 2 \tl{\tiny \matrixOneOne} 1$. Instead $\pn \poi \sem{\IBRw} \poi \tra$ maps $\Bmult$ first to $\Wmult$, then to $2 \tl{\id} 2 \tr{\tiny \matrixOneOneFlat} 1$ and finally to $2 \tr{\id} 2 \tl{\tiny \matrixOneOne} 1$. \qed \end{proof}

%
%

%% file: source/11_cuberebuilt.tex
The results of the previous two sections conclude the proof of Theorem~\ref{th:IBR=SVR}. We are now in position to patch together all the faces of the cube \eqref{eq:cube}. This will also give us an inductive presentation of the isomorphism $\sem{\IBR} \: \IBR \to \SVR$.
\begin{equation*}
\raise30pt\hbox{$
\xymatrix@=35pt{
& {\ABR + \ABRop} \ar[dr]^{[\varphi_1,\varphi_2]} \ar[dd]_(.7){\sem{\ABR}+\sem{\ABR}^{\op}}|{\hole}
\ar[dl]_{[\tau_1,\tau_2]} \ar[rr]^{[\sigma_1,\sigma_2]} & & {\IBRw} \ar[dl]_{\Theta} \ar[dd]^{\sem{\IBRb}} \\
{\IBRb} \ar[rr]^(.7){\Lambda} \ar[dd]_{\sem{\IBRb}}  & & {\IBR} \ar@{.>}[dd]^(.3){\sem{\IBR}} \\
& {\VectR+ \VectRop}  \ar[dr]^{[\psi_1,\psi_2]}  \ar[dl]_{[\iota_1,\iota_2]} \ar[rr]^(.4){[\kappa_1,\kappa_2]}|(.53){\hole} & & {\Span {\VectR}} \ar[dl]_{\Phi} \\
{\Cospan {\VectR}} \ar[rr]^{\Psi} & & {\SVR}
}$}
\end{equation*}
Above we draw the PROP morphism $[\psi_1,\psi_2] \: \VectR + \VectRop \to \SVR$ defined by commutativity of the bottom face. Commutativity of all the faces yields commutativity of the ``section'':
 \begin{equation}\label{eq:cubesection}
 \tag{Sec}
 \vcenter{
 \xymatrix@=15pt{
 & {\ABR + \ABRop} \ar[dr]^{[\varphi_1,\varphi_2]} \ar[dd]_{\sem{\ABR}+\sem{\ABR}^{\op}}
& &  \\
& & {\IBR} \ar@{.>}[dd]^{\sem{\IBR}} \\
& {\VectR+ \VectRop}  \ar[dr]_{[\psi_1,\psi_2]} & & \\
 & & {\SVR}
 }
 }
 \end{equation}
Diagram~\eqref{eq:cubesection} provides us a recipe for an inductive presentation of $\sem{\IBR}$, for circuits of $\IBR$, similarly to what we previously did for $\sem{\IBRb}$ and $\sem{\IBRw}$:
\begin{align*}
   c  \mapsto \left\{
	\begin{array}{ll}
        \psi_1(\sem{\ABR}(c')) & \text{ if } c = \varphi_1(c') \text{ and }c'\in \Sigma_{\ABR} \\
        \psi_2( \sem{\ABR}^{\op} (c'))  & \text{ if } c = \varphi_2(c') \text{ and }c'\in \Sigma_{\ABRop} \\
	\sem{\IBR}(c_1) \poi \sem{\IBR}(c_2) & \text{ if } c=c_1 \poi c_2\\
	\sem{\IBR}(c_1) \tns \sem{\IBR}(c_2) & \text{ if } c=c_1 \tns c_2\\
        \end{array}
\right.
 \end{align*}

By observing the definition of $\sem{\ABR}$ and $[\iota_1,\iota_2] \poi \Psi$ (or, equivalently, $[\kappa_1,\kappa_2] \poi \Phi$), one can compute the value of $\sem{\IBR}$ on the generators in $\Sigma_{\ABR}$ as follows:

\begin{multicols}{2}\noindent
\begin{equation*}
\Bcomult \longmapsto [(%
\small{\matrixOne},\tiny{\left(\begin{array}{c}
                \!\!  1 \!\!\\
                 \!\! 1 \!\!
                \end{array}\right)})]
\end{equation*}
\begin{equation*}
\Wmult  \longmapsto  [\tiny{(\left(%
                \begin{array}{c}
                \!\!\!  0 \!\!\!\\
                \!\!\!  1 \!\!\!
                \end{array}\right)},\small{\matrixOne}),(\tiny{\left(%
                \begin{array}{c}
                \!\!\!  1 \!\!\!\\
                \!\!\!  0 \!\!\!
                \end{array}\right)},\small{\matrixOne})]
\end{equation*}
\end{multicols}
\smallskip
\begin{multicols}{3}\noindent
\begin{equation*}
\Bcounit \longmapsto [(\small{\matrixOne,\matrixNull})]
\end{equation*}
\begin{equation*}
\Wunit \longmapsto [(\small{\matrixNull,\matrixZero})]
\end{equation*}
\begin{equation*}
\scalar \longmapsto [(\small{\matrixOne},\small{\left(%
                \begin{array}{c}
                \!\!  k \!\!
                \end{array}\right)})]
\end{equation*}
\end{multicols}
\smallskip
A generator $c$ in $\Sigma_{\ABRop}$ is mapped to the inverse relation of $\sem{\IBR}(\coc{c})$. In the above definition, notation $[(\xx_1,\yy_1), \dots, (\xx_z,\yy_z)]$ for an arrow in $\SVR[n,m]$ indicates the subspace of $\frPID^n \times \frPID^m$ spanned by pairs $(\xx_1,\yy_1), \dots, (\xx_z,\yy_z)$ of vectors, where each $\xx_i$ is in $\PID^n$ and each $\yy_i$ is in $\PID^m$. Also, $\matrixNull$ denotes the unique element of the space of dimension $0$.

%% file: source/12_cubeinstances.tex
In this concluding section, we exhibit a simple, yet important, example of our construction: the axiomatisation $\IH{\Z}$ for the PROP of rational subspaces. As in general case, we begin by describing the sub-theory of integer matrices.

\paragraph{The theory of integer matrices} By Proposition~\ref{prop:ab=vect}, the PROP $\Mat{\Z}$ of integer matrices is presented by the axioms \eqref{eq:wmonunitlaw}-\eqref{eq:scalarsum} of $\HA{\Z}$. In fact, a finite axiomatisation is possible: let us denote by $\HA{}$ the PROP freely generated by the SMT with signature $\{\antipode ,\Bcounit , \Bcomult , \Wunit , \Wmult \}$ and equations:
 \begin{multicols}{3}\noindent
\begin{equation*}
\lower9pt\hbox{$\includegraphics[height=.9cm]{graffles/Wunitlaw.pdf}$}
\!\!\!
=\!
\lower5pt\hbox{$\includegraphics[height=.6cm]{graffles/idcircuit.pdf}$}
\end{equation*}
\begin{equation*}
\lower5pt\hbox{$\includegraphics[height=.6cm]{graffles/Wmult.pdf}$}
\!
=
\!\!\!\!
\lower11pt\hbox{$\includegraphics[height=1cm]{graffles/Wcomm.pdf}$}
\end{equation*}
\begin{equation*}
\lower12pt\hbox{$\includegraphics[height=1cm]{graffles/Wassocl.pdf}$}
\!\!\!
=
\!\!\!
\lower12pt\hbox{$\includegraphics[height=1cm]{graffles/Wassocr.pdf}$}
\end{equation*}
\end{multicols}
\begin{multicols}{3}\noindent
\begin{equation*}
\lower9pt\hbox{$\includegraphics[height=.9cm]{graffles/Bcounitlaw.pdf}$}
\!\!\!
=
\!
\lower5pt\hbox{$\includegraphics[height=.6cm]{graffles/idcircuit.pdf}$}
\end{equation*}
\begin{equation*}
\lower5pt\hbox{$\includegraphics[height=.6cm]{graffles/Bcomult.pdf}$}
\!
=
\!\!\!
\lower11pt\hbox{$\includegraphics[height=1cm]{graffles/Bcomm.pdf}$}
\end{equation*}
\begin{equation*}
\lower11pt\hbox{$\includegraphics[height=1cm]{graffles/Bcoassocl.pdf}$}
\!\!\!
=
\!\!\!
\lower11pt\hbox{$\includegraphics[height=1cm]{graffles/Bcoassocr.pdf}$}
\end{equation*}
\end{multicols}
\begin{multicols}{4}
\noindent
\begin{equation*}
\lower3pt\hbox{$
\lower5pt\hbox{$\includegraphics[height=.6cm]{graffles/lunitsl.pdf}$}
=
\lower5pt\hbox{$\includegraphics[height=.6cm]{graffles/lunitsr.pdf}$}
$}
\end{equation*}\noindent\begin{equation*}
\lower3pt\hbox{$
\lower5pt\hbox{$\includegraphics[height=.6cm]{graffles/runitsl.pdf}$}
=
\lower5pt\hbox{$\includegraphics[height=.6cm]{graffles/runitsr.pdf}$}
$}
\end{equation*}
\begin{equation*}
\lower5pt\hbox{$\includegraphics[height=.6cm]{graffles/bialgl.pdf}$}
=
\lower10pt\hbox{$\includegraphics[height=.9cm]{graffles/bialgr.pdf}$}
\end{equation*}
\begin{equation*}
\lower4pt\hbox{$
\lower4pt\hbox{$\includegraphics[height=.5cm]{graffles/unitsl.pdf}$}
= \lower4pt\hbox{$\includegraphics[height=.5cm]{graffles/idzerocircuit.pdf}$}
$}
\end{equation*}
\end{multicols}
\begin{multicols}{4}\noindent
\begin{equation*}
\lower11pt\hbox{$\includegraphics[height=1cm]{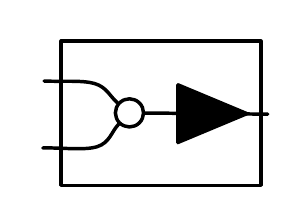}$}
\!\!
=
\!\!
\lower11pt\hbox{$\includegraphics[height=1cm]{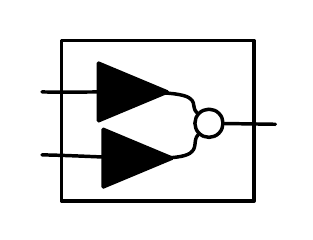}$}
\end{equation*}
\begin{equation*}
\lower10pt\hbox{$\includegraphics[height=.9cm]{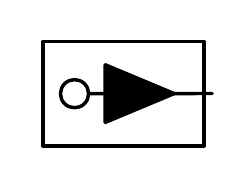}$}
\!\!
=
\!\!
\lower10pt\hbox{$\includegraphics[height=.9cm]{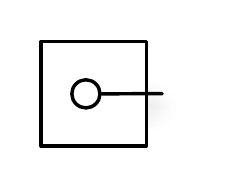}$}
\end{equation*}
\begin{equation*}
\lower11pt\hbox{$\includegraphics[height=1cm]{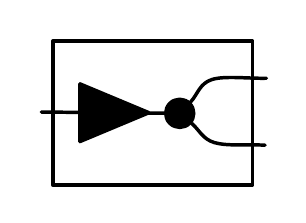}$}
\!\!
=
\!\!
\lower11pt\hbox{$\includegraphics[height=1cm]{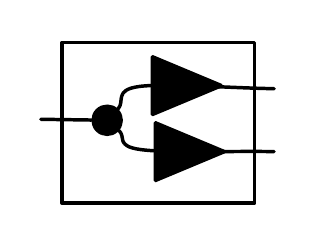}$}
\end{equation*}
\begin{equation*}
\lower10pt\hbox{$\includegraphics[height=.9cm]{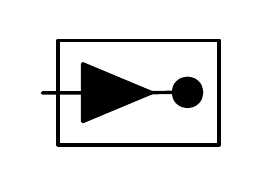}$}
\!\!
=
\!\!
\lower10pt\hbox{$\includegraphics[height=.9cm]{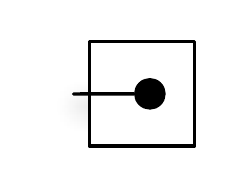}$}
\end{equation*}
\end{multicols}\vspace{-.3cm}
\begin{equation*}
  \lower11pt\hbox{$\includegraphics[height=1cm]{graffles/hopfr.pdf}$}
  = \lower8pt\hbox{$\includegraphics[height=.8cm]{graffles/hopfc.pdf}$}
   \end{equation*}

\begin{proposition}\label{prop:finitePresHAZ} $\HA{} \cong \HA{\Z}$.
\end{proposition}
\begin{proof}
We define a PROP morphism $\alpha \: \HA{\Z} \to \HA{}$ inductively as follows. It is the identity on $\Bcounit$, $\Bcomult$, $\Wunit$ and $\Wmult$. For $k \in \Z$, $\alpha(\scalar)$ is given by:
\begin{align*}
 \lower4pt\hbox{$\includegraphics[height=15pt]{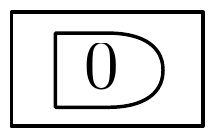}$} \mapsto \lower4pt\hbox{$\includegraphics[height=15pt]{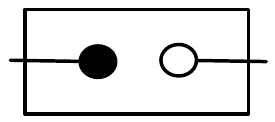}$} %
 \ \ \ \ &&
  \lower10pt\hbox{$\includegraphics[height=25pt]{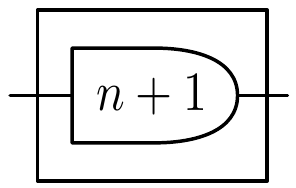}$} \mapsto \lower15pt\hbox{$\includegraphics[height=35pt]{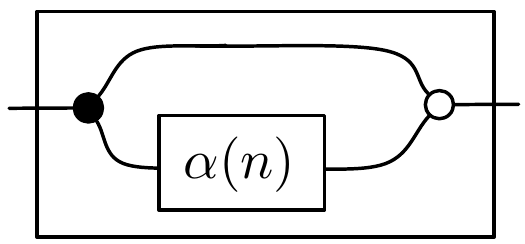}$} \ \ \ \  &&
  \lower10pt\hbox{$\includegraphics[height=25pt]{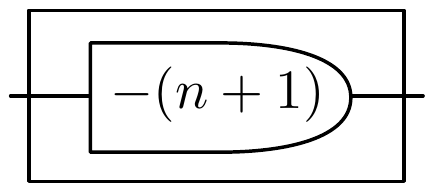}$} \mapsto \lower10pt\hbox{$\includegraphics[height=25pt]{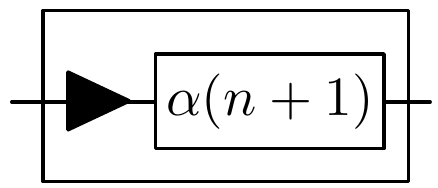}$}
\end{align*}
Finally, we put $\alpha(c_1 \tns c_2) = \alpha(c_1) \tns \alpha(c_2)$ and $\alpha(c_1 \poi c_2) = \alpha(c_1) \poi \alpha(c_2)$.
An inductive argument confirms that $\alpha$ is well-defined, in the sense that it preserves equality of circuits in $\HA{\Z}$. Fullness is clear by construction. For faithfulness, just observe that all axioms of $\HA{}$ are also axioms of $\HA{\Z}$. \qed
\end{proof}

A pleasant example of graphical reasoning in $\HA{}$ is the derivation showing that the antipode $\antipode$ is involutive:
$$\includegraphics[height=2.2cm]{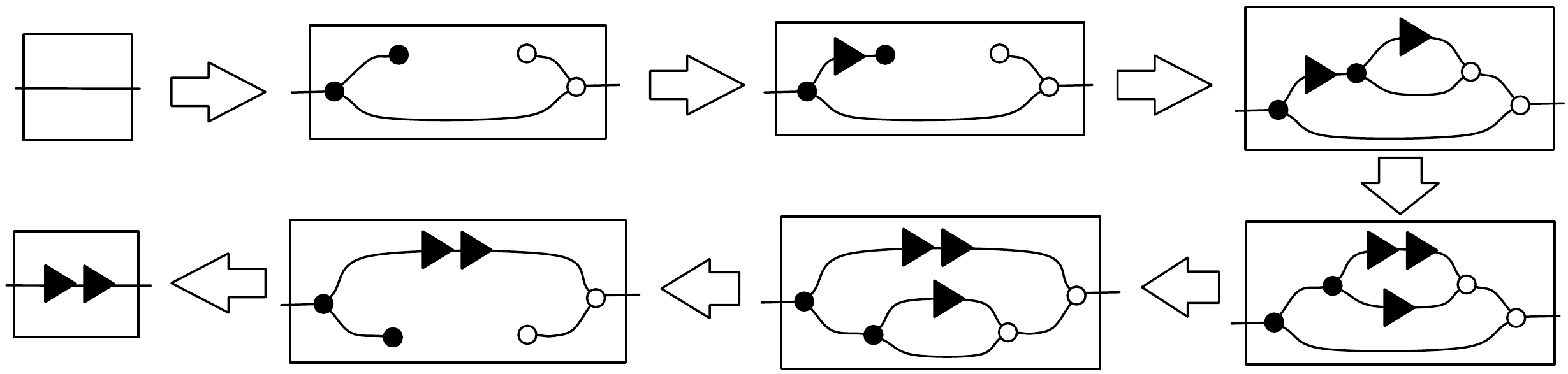}$$

\paragraph{The theory of rational subspaces} By Theorem~\ref{th:IBR=SVR}, $\IH{\Z}$ is isomorphic to the PROP $\SVH{\Q}$ of subspaces over the field $\Q$ of rational numbers. In view of Proposition~\ref{prop:finitePresHAZ}, we can give an alternative presentation of $\IH{\Z}$ based on the finite signature of $\HA{}+\HA{}^{\op}$: in axioms \eqref{eq:lcmIH}-\eqref{eq:lcmopIH}, $\scalar$ and $\coscalar$ become notational conventions for $\alpha(\scalar)$ and $\alpha^{\op}(\coscalar)$, respectively.

For a glimpse of the graphical reasoning in $\IH{\Z}$, we give a combinatorial circuit description of the subspaces of the 2-dimensional rational space (where $k_1,k_2$ are non-zero integers):
\begin{align}\label{eq:subspacesof2}
\spaceFull && \spaceZero && \spaceXaxis && \spaceYaxis && \spacekonektwo.
\end{align}
The circuit $\spaceFull$ denotes (via $\sem{\IH{\Z}}$) the full space $[ \tiny{\left(\begin{array}{c}
                \!\! 1 \!\! \\
                \!\! 0 \!\!
                \end{array}\right),\left(\begin{array}{c}
                \!\! 0 \!\! \\
                \!\! 1 \!\!
                \end{array}\right)}]$ and $\spaceZero$ the $0$-dimensional subspace $\{\tiny{\left(\begin{array}{c}
                \!\! 0 \!\! \\
                \!\! 0 \!\!
                \end{array}\right)}\}$. The remaining subspaces, all of dimension $1$, are conventionally represented as lines through the origin on the $2$-dimensional cartesian coordinate system. Three kinds of circuit suffice to represent all of them: $\spaceXaxis$ denotes the $x$-axis; $\spaceYaxis$ denotes the $y$-axis; for $k_1, k_2 \neq 0$, $\spacekonektwo$ denotes the line with slope $\frac{k_2}{k_1}$.

Conversely, using the modular structure of $\IH{\Z}$ it is easy to check that the above combinatorial analysis~\eqref{eq:subspacesof2} covers all the $1 \to 1$ circuits.

Notice that $\IH{\Z}[1,1]$ contains within its structure all of rational arithmetic:
$0$ can be identified with \spaceXaxis, and $\frac{k_2}{k_1}$, for $k_1\neq 0$, with \spacekonektwo. Multiplication
$\cdot\: \IH{\Z}[1,1]\times\IH{\Z}[1,1]\to \IH{\Z}[1,1]$
is composition $x\cdot y = x\poi y$, addition $+\: \IH{\Z}[1,1]\times\IH{\Z}[1,1]\to \IH{\Z}[1,1]$
is defined \[x+y = \Bcomult \poi (x \oplus y) \poi \Wmult.\]
Multiplication is associative
but not commutative in general: of course, it \emph{is} commutative when restricted to rationals.
Associativity and commutativity of addition follow from associativity and commutativity in $\bcom$ and $\wmon$.

%% file: source/A_AppendixFrobLaws.tex
The Frobenius axioms both for the white --- \eqref{eq:WFrob} --- and for the black structure --- \eqref{eq:BFrob} --- make valid any deformation of the internal topology of circuits of $\IBRw$, as long as the connections between boundaries are preserved. We list here some useful laws of that kind. In describing the various derivation steps, we occasionally use the notation $(n)^{op}$, which means the counterpart in $\ABRop$ of a valid equation $(n)$ in $\ABR$.
\begin{equation}\label{eq:Bfrobcomult}\tag{F1}
\lower11pt\hbox{$\includegraphics[height=1cm]{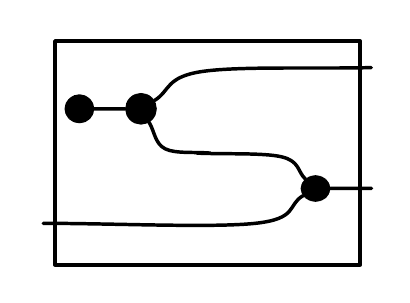}$}
\eql{\eqref{eq:BFrob}}
\lower11pt\hbox{$\includegraphics[height=1cm]{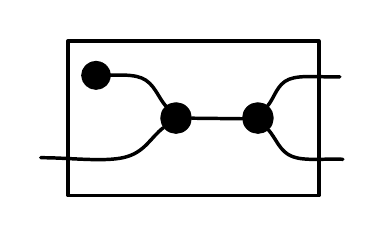}$}
\eql{\eqref{eq:bcomoncomm},\eqref{eq:bcomonunitlaw}}
\lower6pt\hbox{$\includegraphics[height=.6cm]{graffles/Bcomult.pdf}$}
\eql{\eqref{eq:bcomonunitlaw}}
\lower11pt\hbox{$\includegraphics[height=1cm]{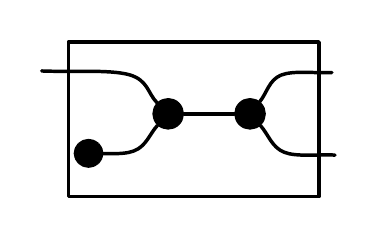}$}
\eql{\eqref{eq:BFrob}}
\lower11pt\hbox{$\includegraphics[height=1cm]{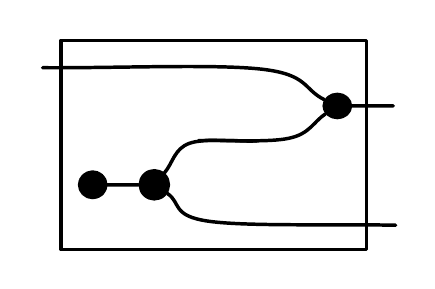}$}
\end{equation}
\begin{equation}\label{eq:Bsnake}\tag{F2}
\lower11pt\hbox{$\includegraphics[height=1cm]{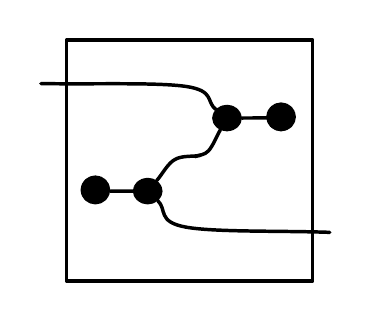}$}
\!\!\!\!\eql{\eqref{eq:BFrob}}\!\!\!\!
\lower10pt\hbox{$\includegraphics[height=1cm]{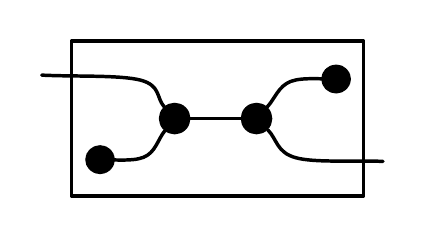}$}
\!\!\!\!\eql{\eqref{eq:bcomoncomm},\eqref{eq:bcomonunitlaw},\eqref{eq:bcomonunitlaw}$^{op}$}\!\!\!\!
\lower4pt\hbox{$\includegraphics[height=.6cm]{graffles/idcircuit.pdf}$}
\!\!\!\!\eql{\eqref{eq:bcomonunitlaw},\eqref{eq:bcomoncomm}$^{op}$,\eqref{eq:bcomonunitlaw}$^{op}$}\!\!\!\!
\lower10pt\hbox{$\includegraphics[height=1cm]{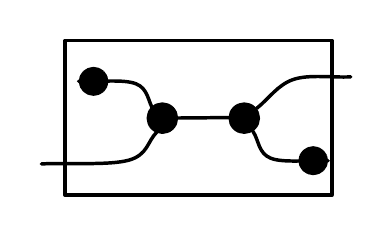}$}
\!\!\!\!\eql{\eqref{eq:BFrob}}\!\!\!\!
\lower11pt\hbox{$\includegraphics[height=1cm]{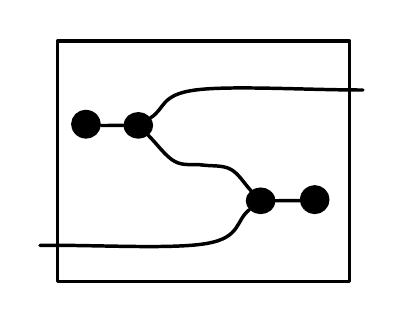}$}
\end{equation}

The following laws are derived analogously. The ones involving the white structure use the white Frobenius axiom~\eqref{eq:WFrob}.

\begin{multicols}{2}\noindent
\begin{equation}\label{eq:Bfrobmult}\tag{F3}
\lower11pt\hbox{$\includegraphics[height=1cm]{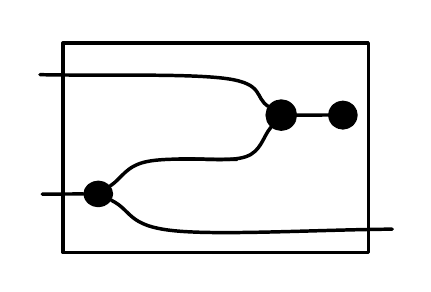}$} =
\lower6pt\hbox{$\includegraphics[height=.6cm]{graffles/Bmult.pdf}$} =
\lower11pt\hbox{$\includegraphics[height=1cm]{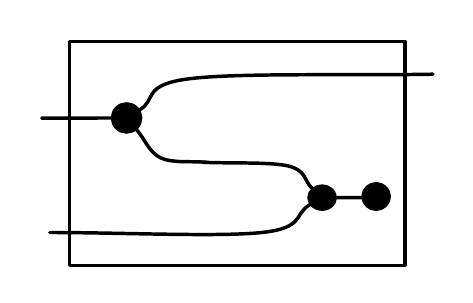}$}
\end{equation}
\begin{equation}\label{eq:Wsnake}\tag{F4}
\lower11pt\hbox{$\includegraphics[height=1cm]{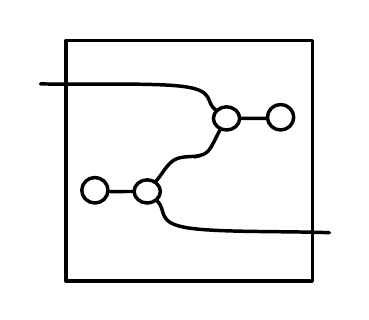}$} =
\lower6pt\hbox{$\includegraphics[height=.6cm]{graffles/idcircuit.pdf}$} =
\lower11pt\hbox{$\includegraphics[height=1cm]{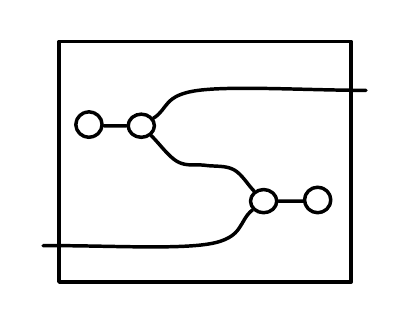}$}
\end{equation}
\end{multicols}
\begin{multicols}{2}\noindent
\begin{equation}\label{eq:Wfrobcomult}\tag{F5}
\lower11pt\hbox{$\includegraphics[height=1cm]{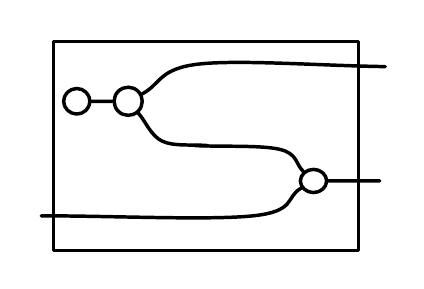}$} =
\lower6pt\hbox{$\includegraphics[height=.6cm]{graffles/Wcomult.pdf}$} =
\lower11pt\hbox{$\includegraphics[height=1cm]{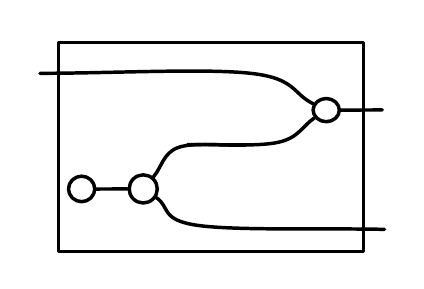}$}
\end{equation}
\begin{equation}\label{eq:Wfrobmult}\tag{F6}
\lower11pt\hbox{$\includegraphics[height=1cm]{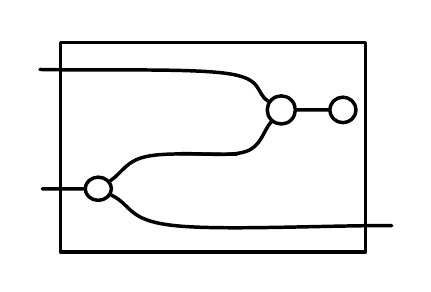}$} =
\lower11pt\hbox{$\includegraphics[height=1cm]{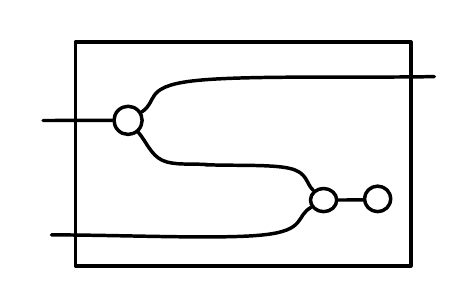}$} =
\lower6pt\hbox{$\includegraphics[height=.6cm]{graffles/Wmult.pdf}$}
\end{equation}
\end{multicols}

\noindent For later reference, we also record the following derivation.
\begin{equation}\label{eq:wccantipodesquare}\tag{F7}
\lower12pt\hbox{$\includegraphics[height=1.1cm]{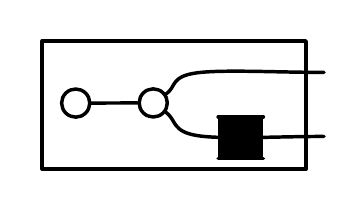}$} \eql{\eqref{eq:scalarwunit}}
\lower12pt\hbox{$\includegraphics[height=1.1cm]{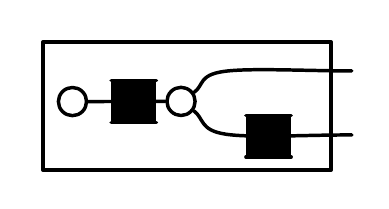}$} \eql{\eqref{eq:scalarwmult}$^{\op}$}
\lower12pt\hbox{$\includegraphics[height=1.1cm]{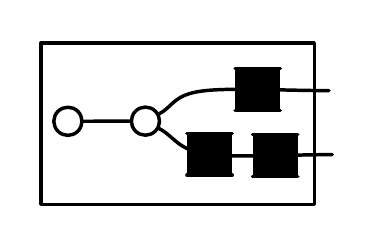}$}
\eql{\eqref{eq:scalarmult}}
\lower12pt\hbox{$\includegraphics[height=1.1cm]{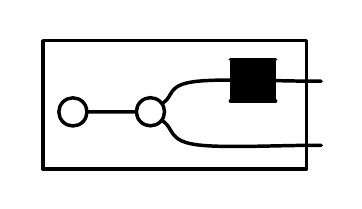}$}
\end{equation}
The same equation reflected about the $y$-axis and the black counterparts are proven analogously.
\begin{multicols}{3}\noindent
\begin{equation}\label{eq:lwccantipodesquare}\tag{F8}
\lower8pt\hbox{$\includegraphics[height=.8cm]{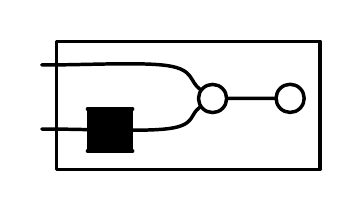}$} \!\!\!=\!\!\!
\lower8pt\hbox{$\includegraphics[height=.8cm]{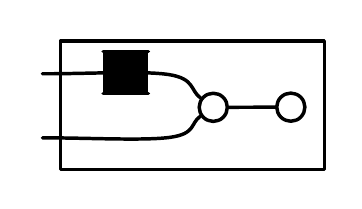}$}
\end{equation}
\begin{equation}\label{eq:bccantipodesquare}\tag{F9}
\lower8pt\hbox{$\includegraphics[height=.8cm]{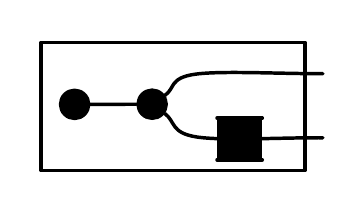}$} \!\!\!=\!\!\!
\lower8pt\hbox{$\includegraphics[height=.8cm]{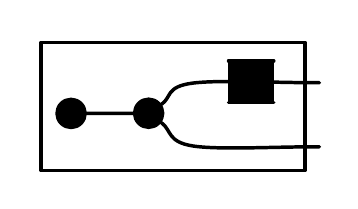}$}
\end{equation}
\begin{equation}\label{eq:lbccantipodesquare}\tag{F10}
\lower8pt\hbox{$\includegraphics[height=.8cm]{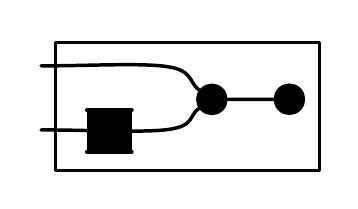}$} \!\!\!=\!\!\!
\lower8pt\hbox{$\includegraphics[height=.8cm]{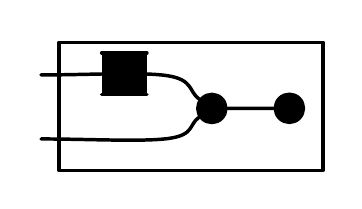}$}
\end{equation}
\end{multicols}

%% file: source/C_AppendixLaws.tex
In this section we supply the equational proofs of the laws stated in Section~\ref{sec:ibrw}. We begin with the derivations of \eqref{eq:lccb} and \eqref{eq:uniqueantipode}.
\begin{equation*}
\lower9pt\hbox{$\includegraphics[height=.8cm]{graffles/wccantipodel.pdf}$} \eql{\eqref{eq:rcc}}
\lower9pt\hbox{$\includegraphics[height=.8cm]{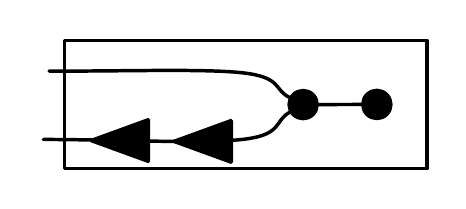}$} \eql{\eqref{eq:scalarmult}}
\lower8pt\hbox{$\includegraphics[height=.7cm]{graffles/bccl.pdf}$}
\end{equation*}
\begin{equation*}
\lower4pt\hbox{$\includegraphics[height=.5cm]{graffles/antipode.pdf}$} \eql{\eqref{eq:lcm}}
\lower9pt\hbox{$\includegraphics[height=.8cm]{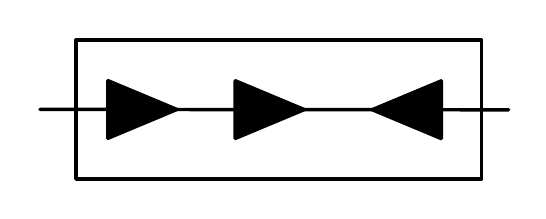}$} \eql{\eqref{eq:scalarmult}}
\lower4pt\hbox{$\includegraphics[height=.5cm]{graffles/antipodeop.pdf}$}
\end{equation*}
The derivation of \eqref{eq:rccb} is analogous to the one of \eqref{eq:lccb}, with \eqref{eq:lcc} used in place of \eqref{eq:rcc}. Now that \eqref{eq:uniqueantipode} has been proven, we follow the convention to write $\antipodesquare$ for both $\antipode$ and $\antipodeop$. We give next the derivation for~\eqref{eq:QFrob}:
\begin{equation*}
\lower60pt\hbox{$\includegraphics[height=4.8cm]{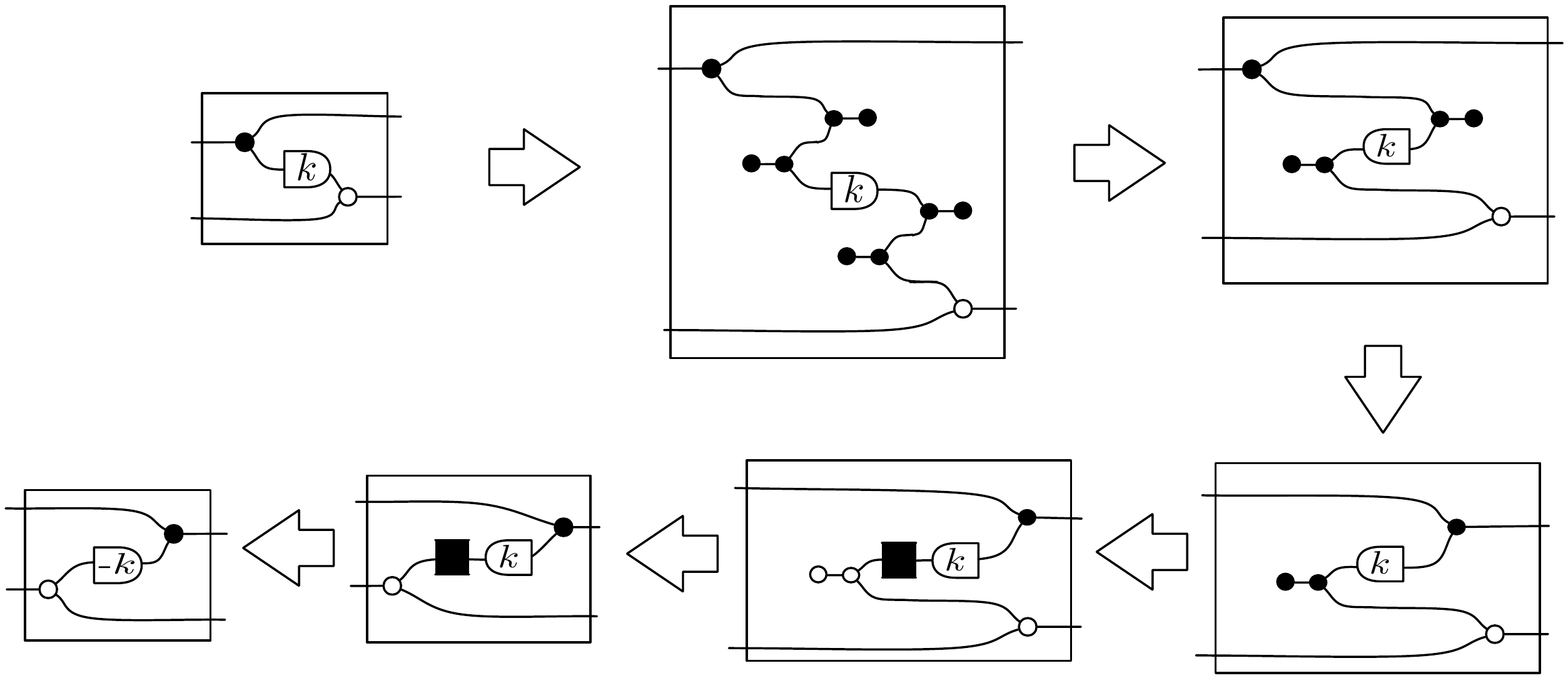}$}
\end{equation*}
The first step uses twice \eqref{eq:Bsnake}. The successive steps use: \eqref{eq:BccscalarAxiomOne}, \eqref{eq:scalarbcounit}, \eqref{eq:Bfrobmult}, \eqref{eq:lccb} and \eqref{eq:wccantipodesquare}, \eqref{eq:Wfrobcomult}, \eqref{eq:scalarmult}.

We show below the proof of \eqref{eq:coscalarwunit}, where $l \neq 0$. The ones for \eqref{eq:scalarwcounit} is symmetric.
\begin{equation*}
\lower6pt\hbox{$\includegraphics[height=.7cm]{graffles/Wunitcoscalarl.pdf}$} \eql{\eqref{eq:scalarwunit}}
\lower7pt\hbox{$\includegraphics[height=.8cm]{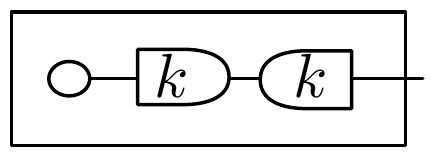}$}
 \eql{\eqref{eq:lcm}}
\lower6pt\hbox{$\includegraphics[height=.7cm]{graffles/Wunit.pdf}$}
\end{equation*}
Next we give the derivation of \eqref{eq:coscalarbcomult}, where $l \neq 0$. The one of \eqref{eq:scalarbmult} is analogous.
\begin{equation*}
\lower9pt\hbox{$\includegraphics[height=.9cm]{graffles/coscalarBcomult_der1.pdf}$} \!\!\!\eql{\eqref{eq:Bfrobcomult}}\!\!\!
\lower10pt\hbox{$\includegraphics[height=1cm]{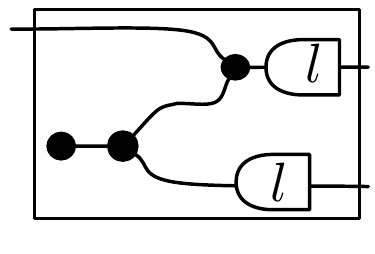}$} \!\!\!\eql{\eqref{eq:BccscalarAxiomTwo}}\!\!\!
\lower10pt\hbox{$\includegraphics[height=1cm]{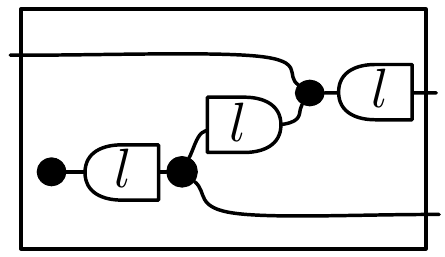}$} \!\!\!\eql{\eqref{eq:scalarbcomult}$^{\op}$}\!\!\!
\lower10pt\hbox{$\includegraphics[height=1cm]{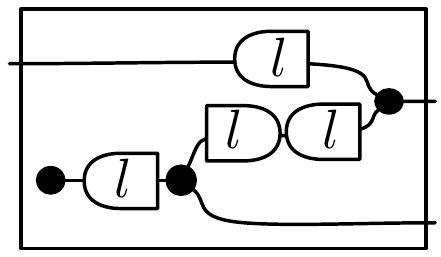}$} \!\!\eql{\eqref{eq:lcm}}\!\!
\lower10pt\hbox{$\includegraphics[height=1cm]{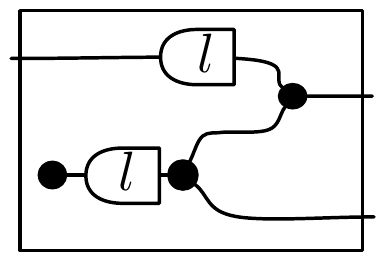}$} \!\!\!\eql{\eqref{eq:scalarbcounit}$^{\op}$}\!\!\!
\lower10pt\hbox{$\includegraphics[height=1cm]{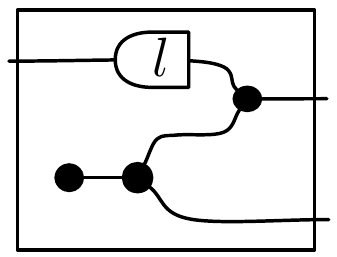}$} \!\!\!\eql{\eqref{eq:Bfrobcomult}}\!\!\!
\lower9pt\hbox{$\includegraphics[height=.9cm]{graffles/coscalarBcomult_der7.pdf}$}
\end{equation*}
We now consider the task of deriving law \eqref{eq:papillon}. For the first half:
\begin{center}
\includegraphics[height=4cm]{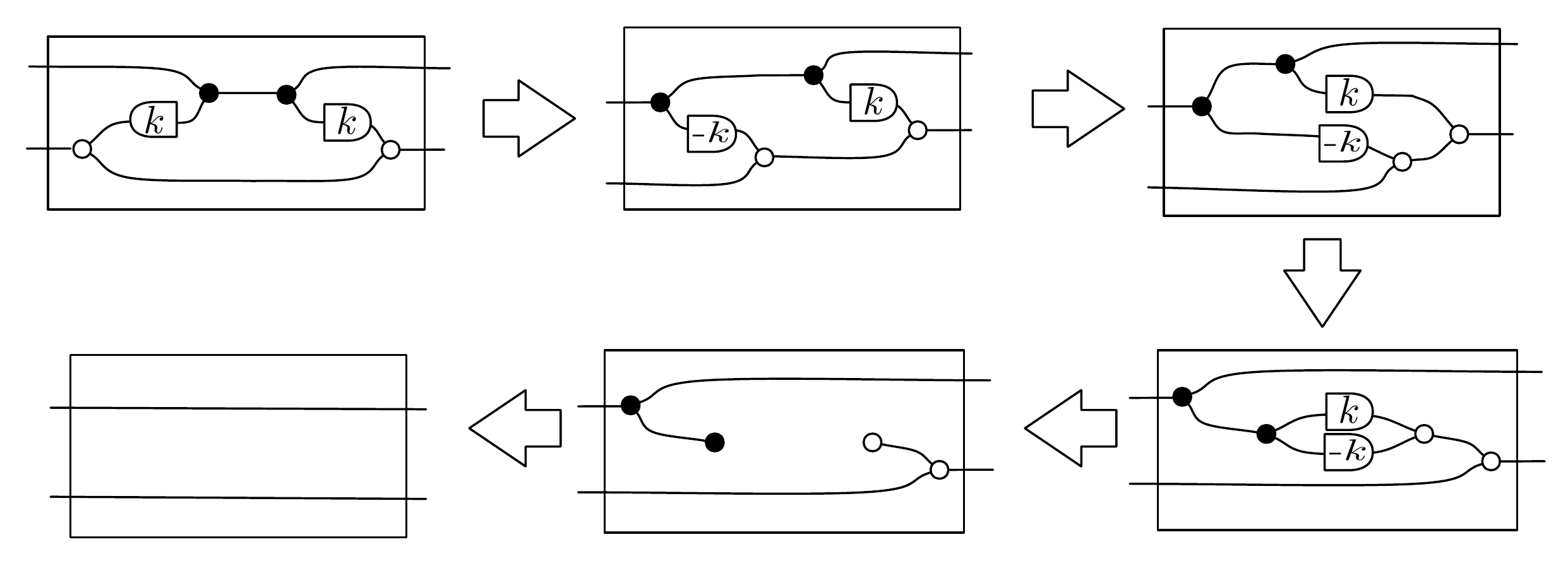}
\end{center}
The sequence of equations that are used is the following: \eqref{eq:QFrob}, axiom of SMCs, \eqref{eq:bcomonassoc} and \eqref{eq:wmonassoc}, \eqref{eq:scalarsum} and \eqref{eq:zeroscalar}, \eqref{eq:bcomonunitlaw} and \eqref{eq:wmonunitlaw}.
The second half of \eqref{eq:papillon} is derived analogously as follows.
\begin{center}
\includegraphics[height=4cm]{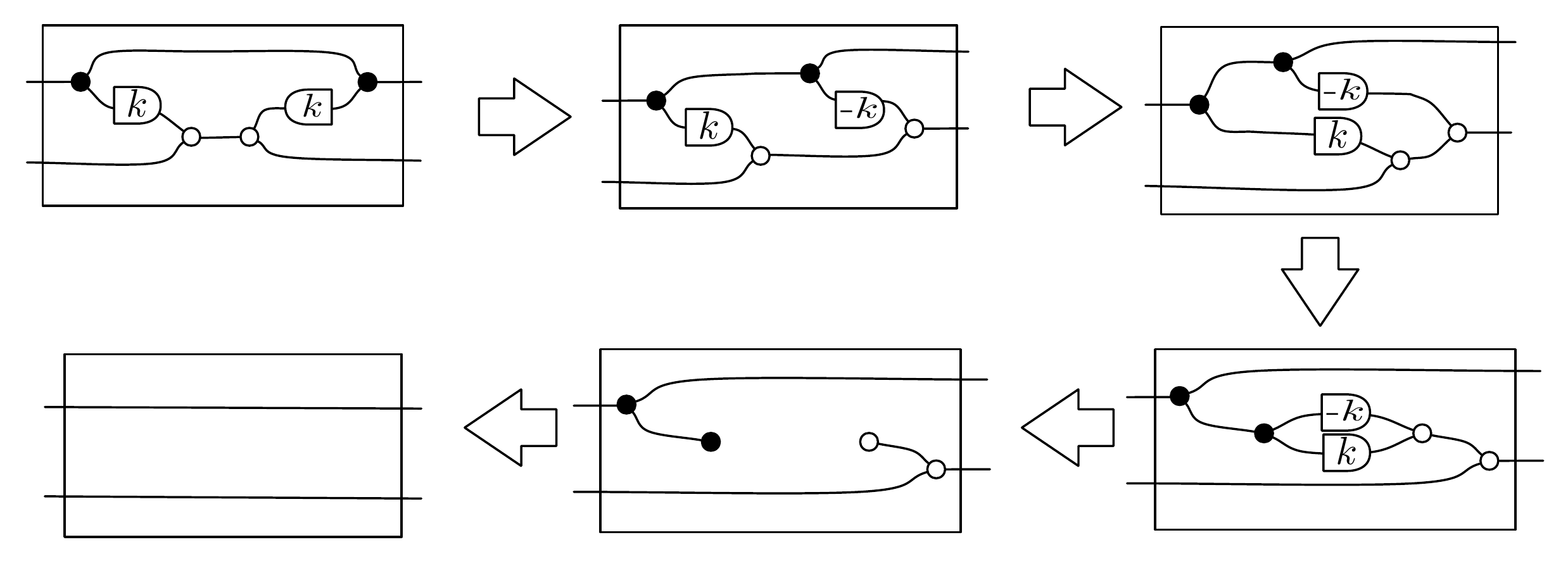}
\end{center}
In order to show the validity of \eqref{eq:wunitcancelbcomult}, we proceed by induction on the coarity $n \geq 1$ of the circuit, i.e., the number of gates on the right boundary. For the case $n = 1$, we have the following derivation, where $l \neq 0$.
\begin{equation}\label{eq:wcounitcancelbcomultder}
\lower40pt\hbox{$\includegraphics[height=3cm]{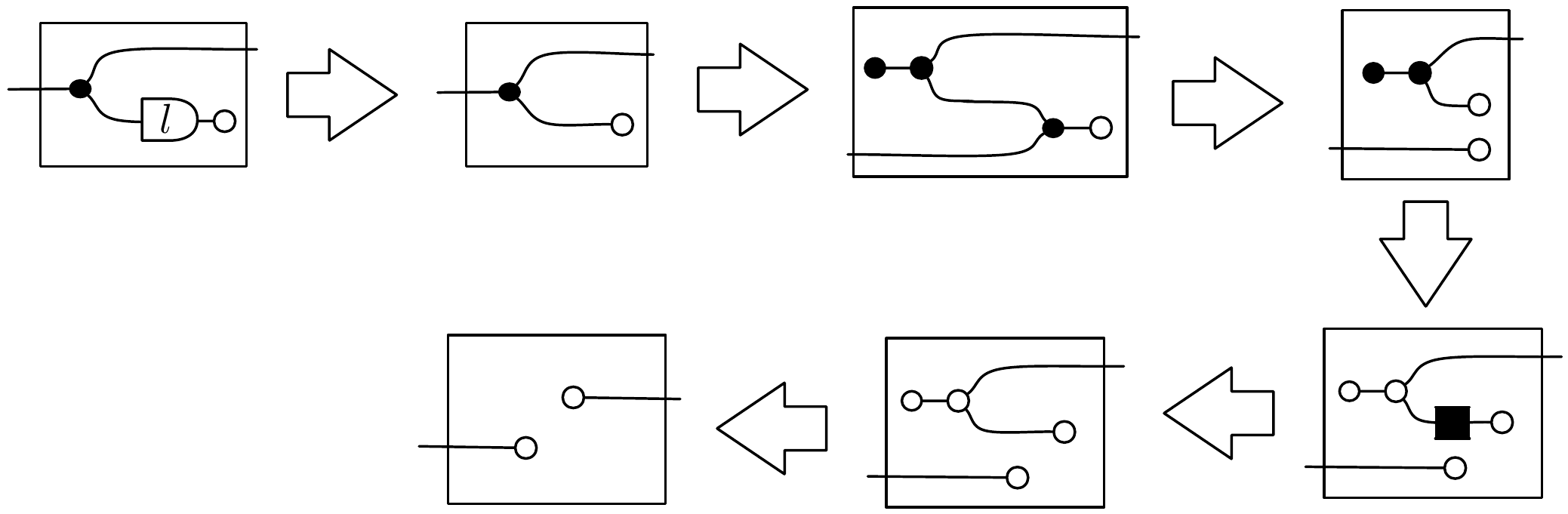}$}
\end{equation}
The sequence of applied laws is: \eqref{eq:scalarwcounit}, \eqref{eq:Bfrobcomult}, \eqref{eq:unitsr}$^{op}$, \eqref{eq:lccb}, \eqref{eq:scalarwunit}, \eqref{eq:wmonunitlaw}$^{op}$. The inductive case is handled as follows.
\begin{equation*}
\lower12pt\hbox{$\includegraphics[height=1.1cm]{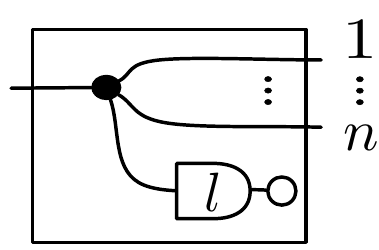}$} \eql{\eqref{eq:scalarwcounit}}
\lower12pt\hbox{$\includegraphics[height=1.1cm]{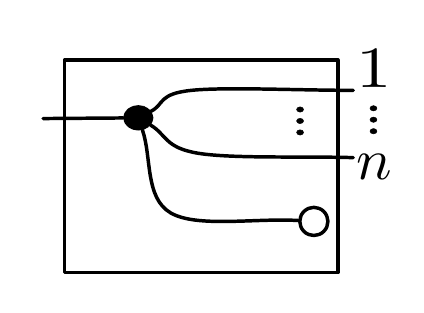}$} \eql{\eqref{eq:bcomonassoc}}
\lower12pt\hbox{$\includegraphics[height=1.1cm]{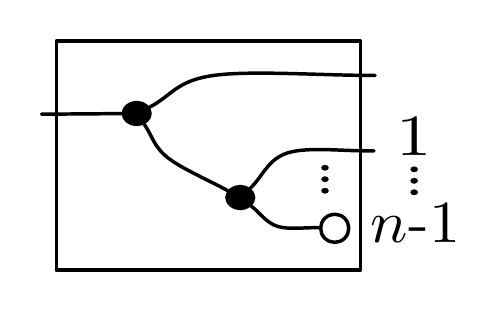}$}
\eql{Ind. hyp.}
\lower12pt\hbox{$\includegraphics[height=1.1cm]{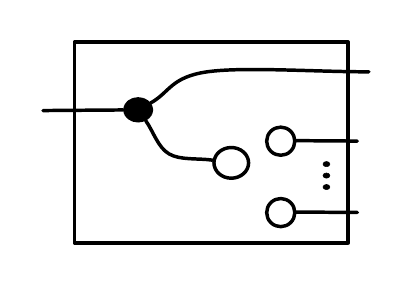}$}
\eql{\eqref{eq:wcounitcancelbcomultder}}
\lower10pt\hbox{$\includegraphics[height=1cm]{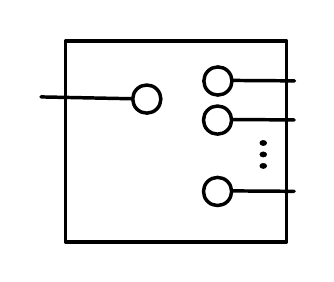}$}
\end{equation*}

\noindent Finally, we show the derivation for \eqref{eq:Bsep}. The sequence of applied laws is \eqref{eq:wbone}, \eqref{eq:bcomonunitlaw}+\eqref{eq:bcomonunitlaw}$^{\op}$, \eqref{eq:scalarsum}+\eqref{eq:scalarsum}$^{\op}$, \eqref{eq:bcomonassoc}+\eqref{eq:bcomonassoc}$^{\op}$, \eqref{eq:papillon}.
\begin{equation*}
\includegraphics[height=2cm]{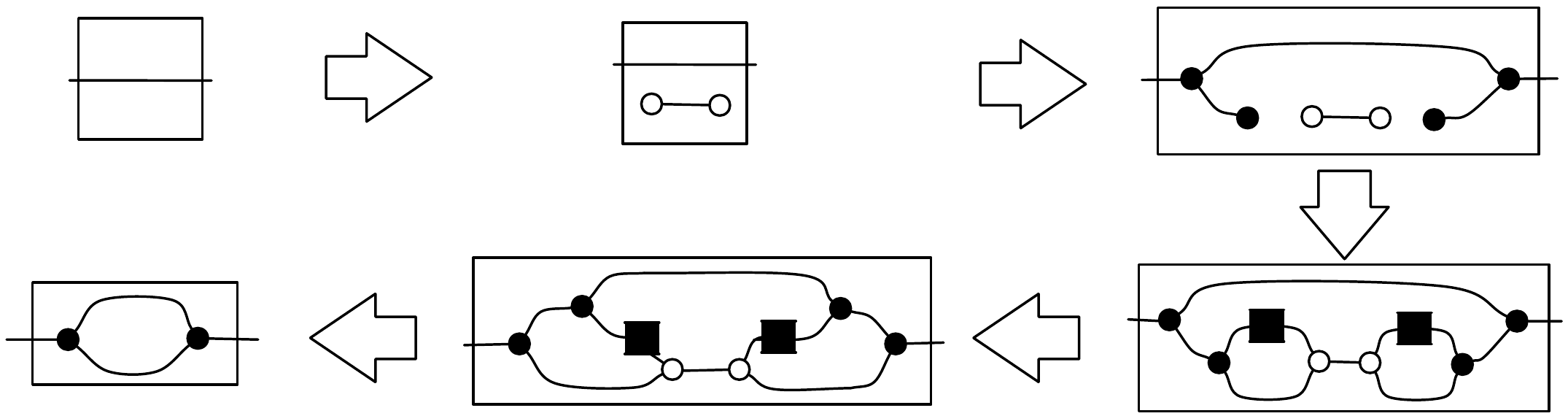}
\end{equation*} 

%% file: source/B_AppendixCC.tex
We give more detailed proofs to the statements of Section~\ref{sec:cc}.

\begin{proof}[Proposition \ref{prop:snakecc}] We give the argument proving the left side of \eqref{eq:gensnake} --- the proof for the right side is completely symmetric. We proceed by induction on $n$. For the case $n = 1$, the statement is given by \eqref{eq:Bsnake}. For the inductive step, let $n = i+1$. In the sequel we show the equality
\begin{eqnarray}\label{eq:ccsnakeInd}
\lower15pt\hbox{$\includegraphics[height=1.3cm]{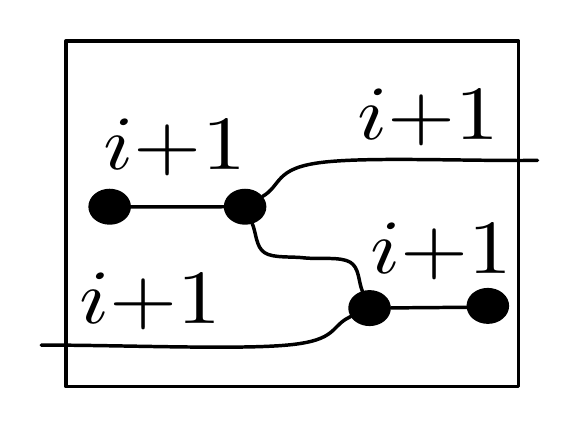}$}
&=&
\lower15pt\hbox{$\includegraphics[height=1.3cm]{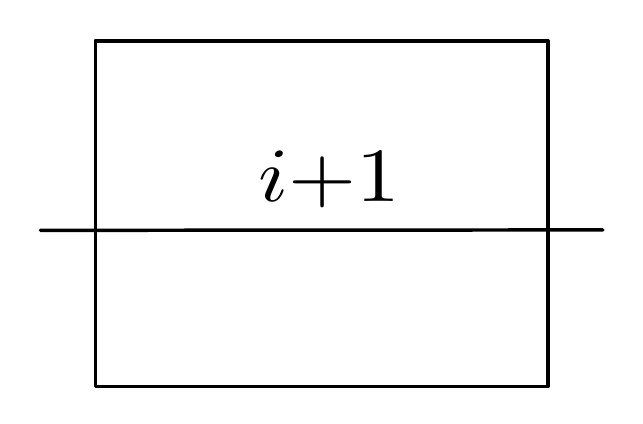}$} \end{eqnarray}
yielding the left side of \eqref{eq:gensnake}. For this purpose, it will be useful the following equation, allowing to ``move'' the compact closed structure past the symmetries of $\IBRw$.
\begin{eqnarray}\label{eq:moveccpastsym}
\lower11pt\hbox{$\includegraphics[height=1cm]{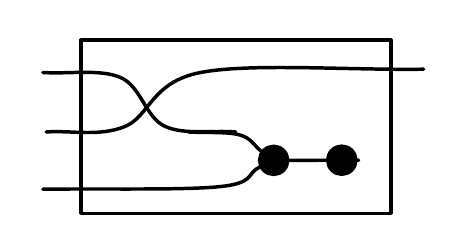}$}
&=&
\lower11pt\hbox{$\includegraphics[height=1cm]{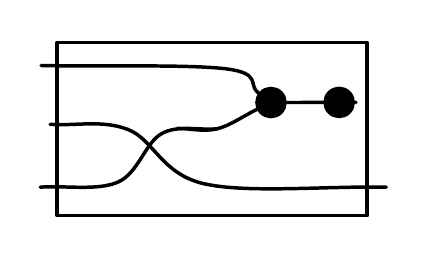}$} \end{eqnarray}
Its derivation in $\IBRw$ is the following.
\begin{eqnarray*}
\lower11pt\hbox{$\includegraphics[height=1cm]{graffles/swapccl.pdf}$}
=
\lower11pt\hbox{$\includegraphics[height=1cm]{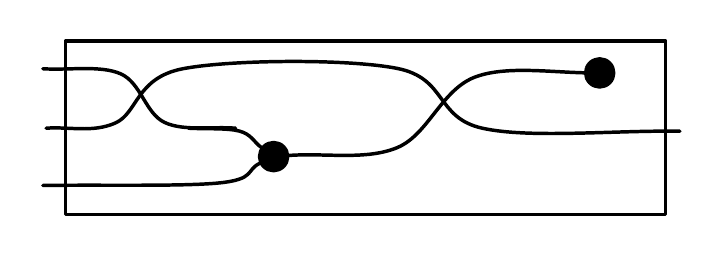}$}
=
\lower11pt\hbox{$\includegraphics[height=1cm]{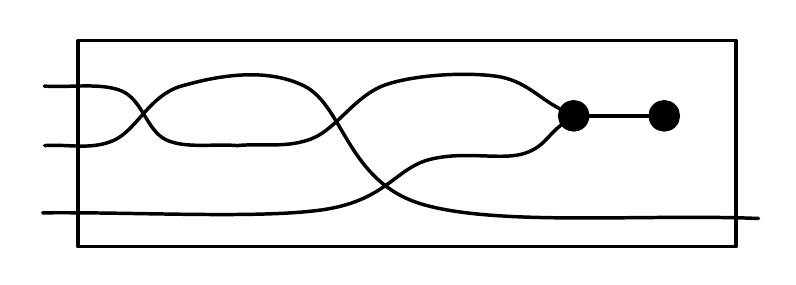}$}
=
\lower11pt\hbox{$\includegraphics[height=1cm]{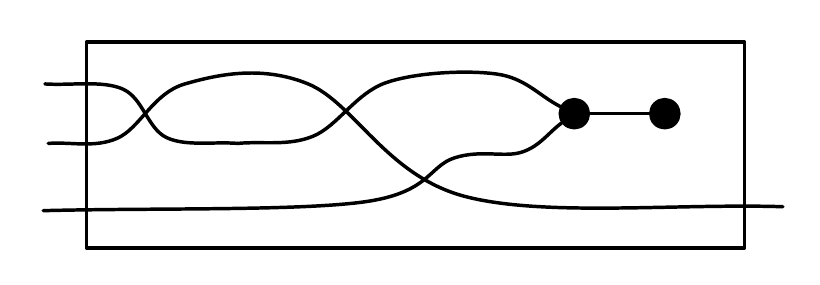}$}
=
\lower11pt\hbox{$\includegraphics[height=1cm]{graffles/swapccr.pdf}$} \end{eqnarray*}
The first and the second equality holds by naturality of symmetry, applied as on the left and on the right below, respectively.
\[
\xymatrix@=30pt{
1 \tns 1 \ar[d]_{\IdBcounitc} \ar[rr]^{\sigma_{1,1} = \symNet } && 1 \tns 1 \ar[d]^{\BcounitId} \\
1 \ar[rr]^{\sigma_{1,0} = \id_1} && 1
}
\qquad \qquad
\xymatrix@=30pt{
1 \tns 2 \ar[rr]^{\sigma_{1,2} = \symNetTwoOne} \ar[d]_{\id_1 \tns \Bmult} && 2 \tns 1 \ar[d]^{\Bmult \tns \id_1} \\
1 \tns 1 \ar[rr]^{\sigma_{1,1} = \symNet} && 1 \tns 1
}
\]
The third equality applies the axiom $\sigma_{1,2} = (\sigma_{1,1} \tns \id_1) \poi (\id_1 \tns \sigma_{1,1})$ of symmetric monoidal categories (SMCs). Finally, the fourth equality applies the axiom $\sigma_{1,1} \poi \sigma_{1,1} = \id_1$ of SMCs. We are now ready to show the derivation of \eqref{eq:ccsnakeInd}. The circuit on the left side of \eqref{eq:ccsnakeInd} has the following shape.
\begin{center}
\includegraphics[height=2.6cm]{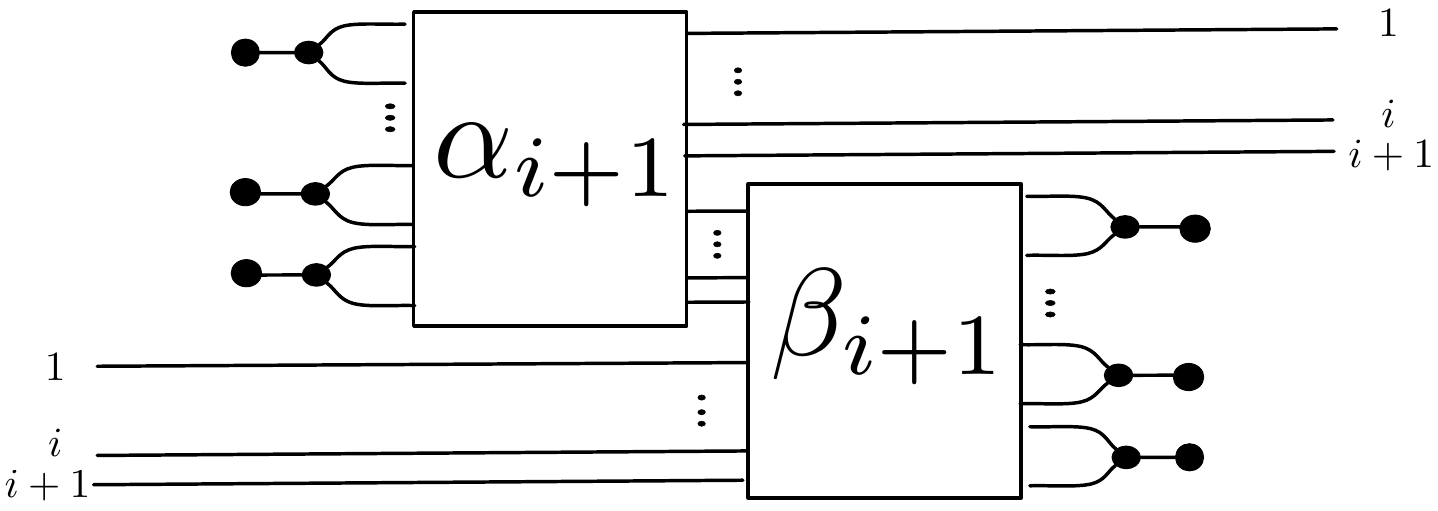}
\end{center}
By definition, port $1$ of the bottommost circuit $\rccB$ (call it $c_l$) connects to port $i+1$ on the right boundary and port $2$ connects to port $1$ of the bottommost circuit $\lccB$ (call it $c_r$). The other port of $c_r$ connects instead to port $i+1$ on the left boundary. By iteratively applying \eqref{eq:moveccpastsym} to $c_r$, we can move it towards the middle of the circuit, past all the symmetries in $\beta_{i+1}$. The resulting circuit is the following:
\begin{center}
\includegraphics[height=3cm]{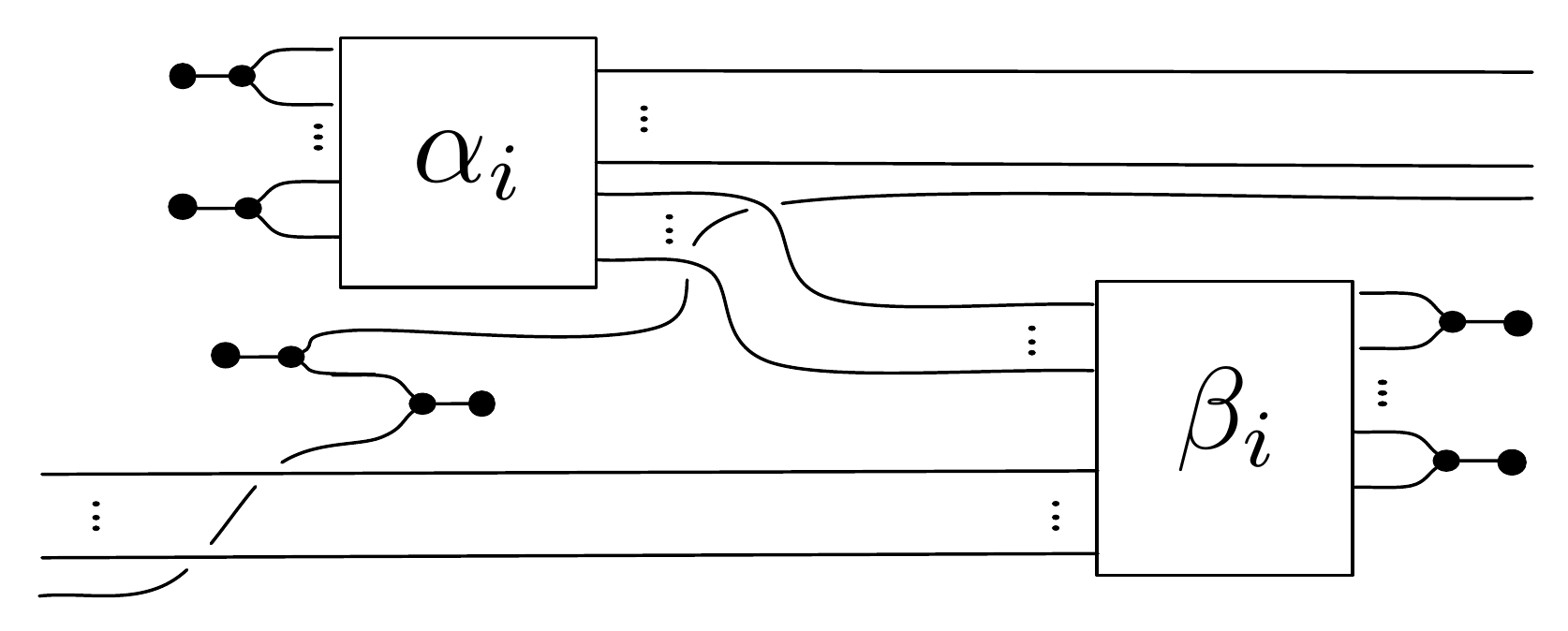}
\end{center}
Note that, now that we isolated $c_l$ and $c_r$, the circuits $\alpha_{i+1}$ and $\beta_{i+1}$ become by definition $\alpha_{i}$ and $\beta_i$ --- observe that the application of \eqref{eq:moveccpastsym} does not affect the arity of the symmetries in the circuit. We are now in position to apply \eqref{eq:Bsnake}:
\begin{center}
\includegraphics[height=3cm]{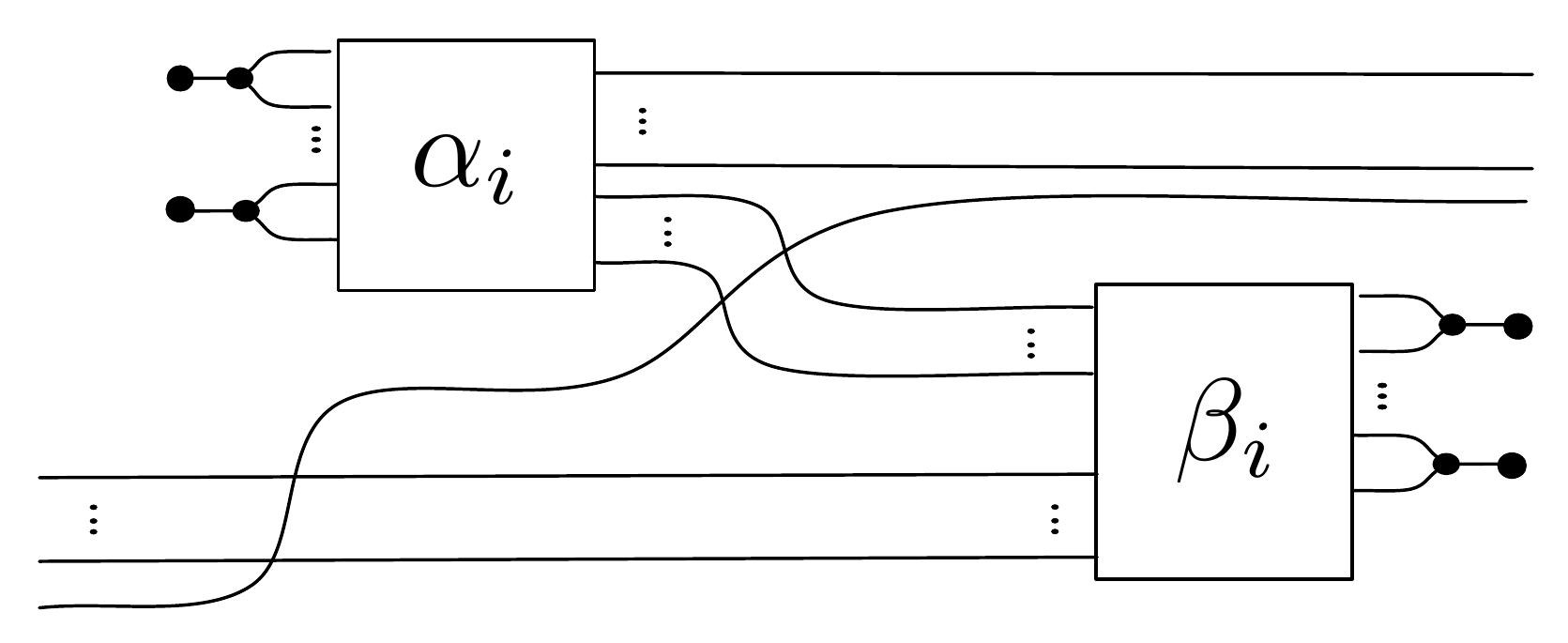}
\end{center}
We can then use again \eqref{eq:moveccpastsym} to move the identity circuit in the middle towards the bottom.
\begin{center}
\includegraphics[height=3cm]{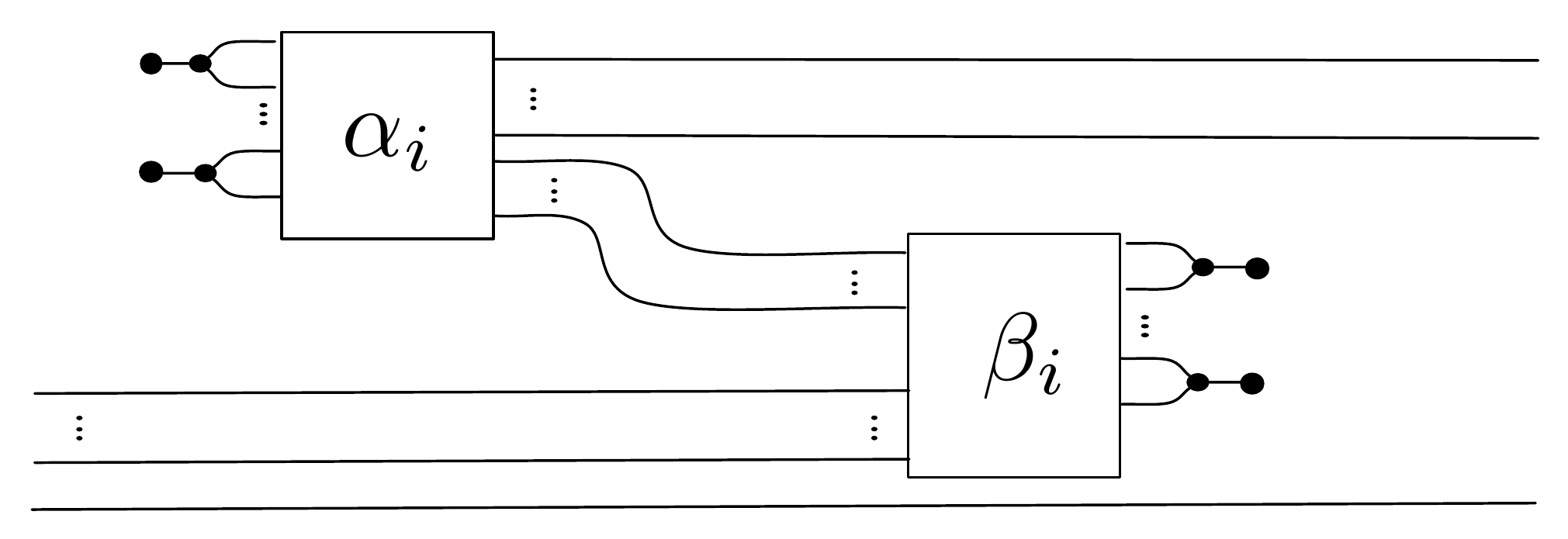}
\end{center}
It is now possible to apply the inductive hypothesis on $i$, obtaining as a result the desired identity circuit as on the right side of \eqref{eq:ccsnakeInd}.
\qed\end{proof}

\begin{proof}[Proposition \ref{prop:star=refl}]
The proof is by induction on $c \in \IBRw$. First we give the derivations for the four base cases of white/black unit/counit.
\begin{equation*}
\lower16pt\hbox{$\includegraphics[height=1.3cm]{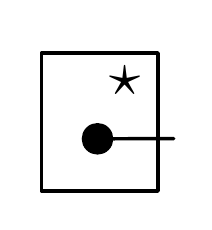}$}
\eql{Def. $\coc{(\cdot)}$}
\lower15pt\hbox{$\includegraphics[height=1.2cm]{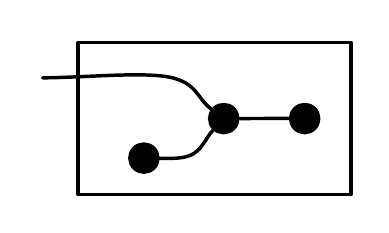}$}
\eql{\eqref{eq:bcomonunitlaw}$^{\op}$}
\ \lower8pt\hbox{$\includegraphics[height=.7cm]{graffles/Bcounit.pdf}$} \end{equation*}
\begin{equation*}
\lower15pt\hbox{$\includegraphics[height=1.3cm]{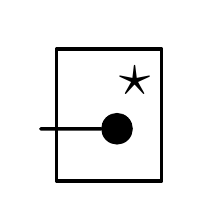}$}
\eql{Def.$\coc{(\cdot)}$}
\lower15pt\hbox{$\includegraphics[height=1.2cm]{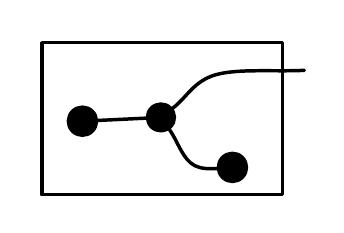}$}
\eql{\eqref{eq:bcomonunitlaw}}
\ \lower8pt\hbox{$\includegraphics[height=.7cm]{graffles/Bunit.pdf}$} \end{equation*}
\begin{equation*}
\lower15pt\hbox{$\includegraphics[height=1.3cm]{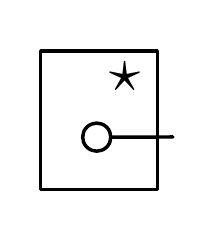}$}
\!\eql{Def.$\coc{(\cdot)}$}\!
\lower15pt\hbox{$\includegraphics[height=1.2cm]{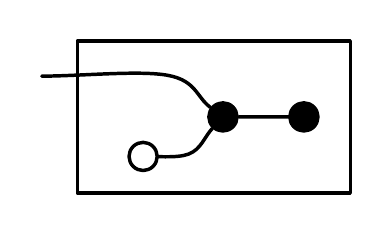}$}
\!\eql{\eqref{eq:lccb},\eqref{eq:lwccantipodesquare}}\!
\lower15pt\hbox{$\includegraphics[height=1.2cm]{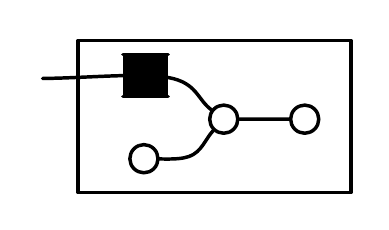}$}
\!\eql{\eqref{eq:wmonunitlaw}}\!
\lower15pt\hbox{$\includegraphics[height=1.2cm]{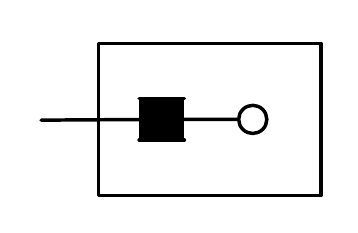}$}
\!\eql{\eqref{eq:scalarwunit}}\!
\ \lower8pt\hbox{$\includegraphics[height=.7cm]{graffles/Wcounit.pdf}$} \end{equation*}
\begin{equation*}
\lower15pt\hbox{$\includegraphics[height=1.3cm]{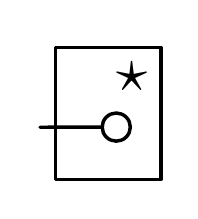}$}
\!\eql{Def.$\coc{(\cdot)}$}\!
\lower15pt\hbox{$\includegraphics[height=1.2cm]{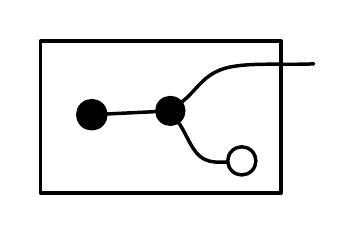}$}
\!\eql{\eqref{eq:rccb}}\!
\lower15pt\hbox{$\includegraphics[height=1.2cm]{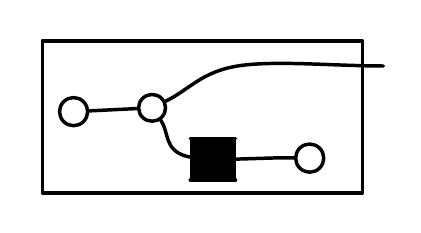}$}
\!\eql{\eqref{eq:scalarwunit}$^{\op}$}\!
\lower15pt\hbox{$\includegraphics[height=1.2cm]{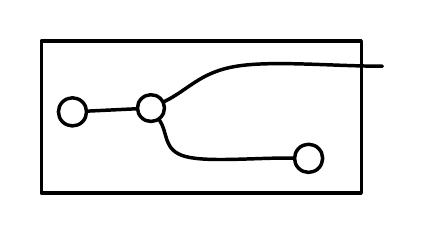}$}
\!\eql{\eqref{eq:wmonunitlaw}$^{\op}$}\!
\ \lower8pt\hbox{$\includegraphics[height=.7cm]{graffles/Wunit.pdf}$} \end{equation*}

We now consider the base cases $\scalar$ and $\coscalar$, for $k \in \PID$.
\begin{equation*}
\lower9pt\hbox{$\includegraphics[height=.8cm]{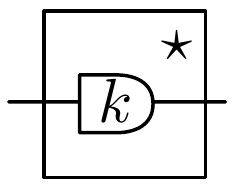}$}
\!\eql{Def. $\coc{\cdot}$}\!
\lower14pt\hbox{$\includegraphics[height=1.3cm]{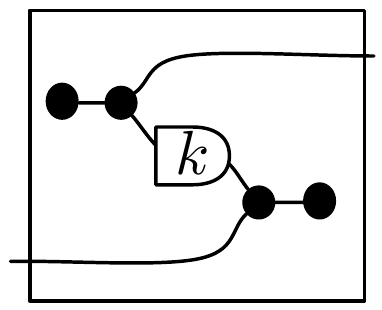}$}
\!\eql{\eqref{eq:BccscalarAxiomOne}}\!
\lower14pt\hbox{$\includegraphics[height=1.3cm]{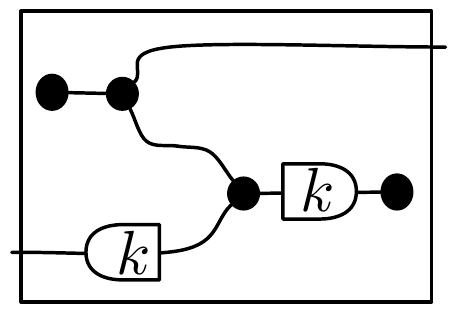}$}
\!\eql{\eqref{eq:scalarbcounit}}\!
\lower14pt\hbox{$\includegraphics[height=1.3cm]{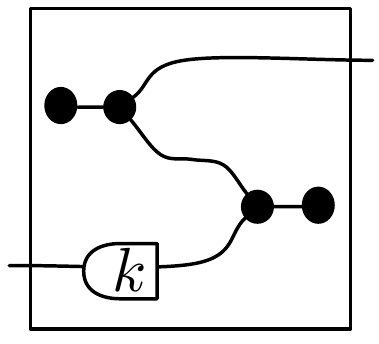}$}
\!\eql{\eqref{eq:Bsnake}}\!
\ \lower9pt\hbox{$\includegraphics[height=.9cm]{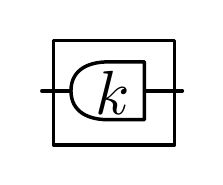}$}
\end{equation*}
\begin{equation*}
\lower9pt\hbox{$\includegraphics[height=.8cm]{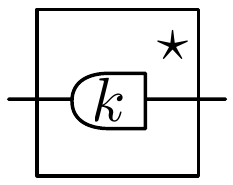}$}
\!\eql{Def. $\coc{\cdot}$}\!
\lower14pt\hbox{$\includegraphics[height=1.3cm]{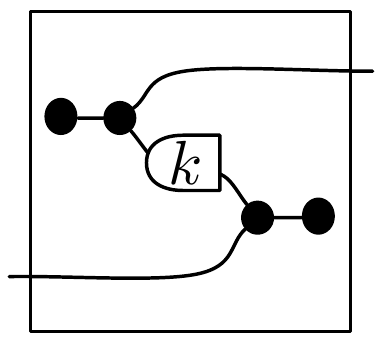}$}
\!\eql{\eqref{eq:BccscalarAxiomTwo}}\!
\lower14pt\hbox{$\includegraphics[height=1.3cm]{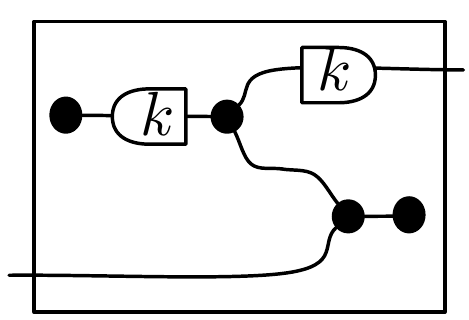}$}
\!\eql{\eqref{eq:scalarbcounit}$^{\op}$}\!
\lower14pt\hbox{$\includegraphics[height=1.3cm]{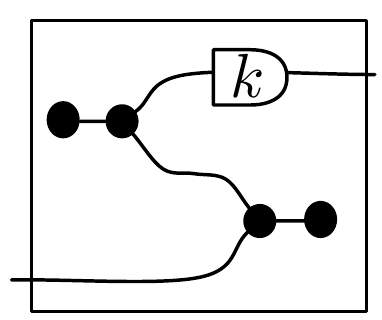}$}
\!\eql{\eqref{eq:Bsnake}}\!
\ \lower9pt\hbox{$\includegraphics[height=.9cm]{graffles/scalar.pdf}$}
\end{equation*}

We also provide the derivation for the base case $\Bcomult$.
\begin{center}
\includegraphics[height=5.5cm]{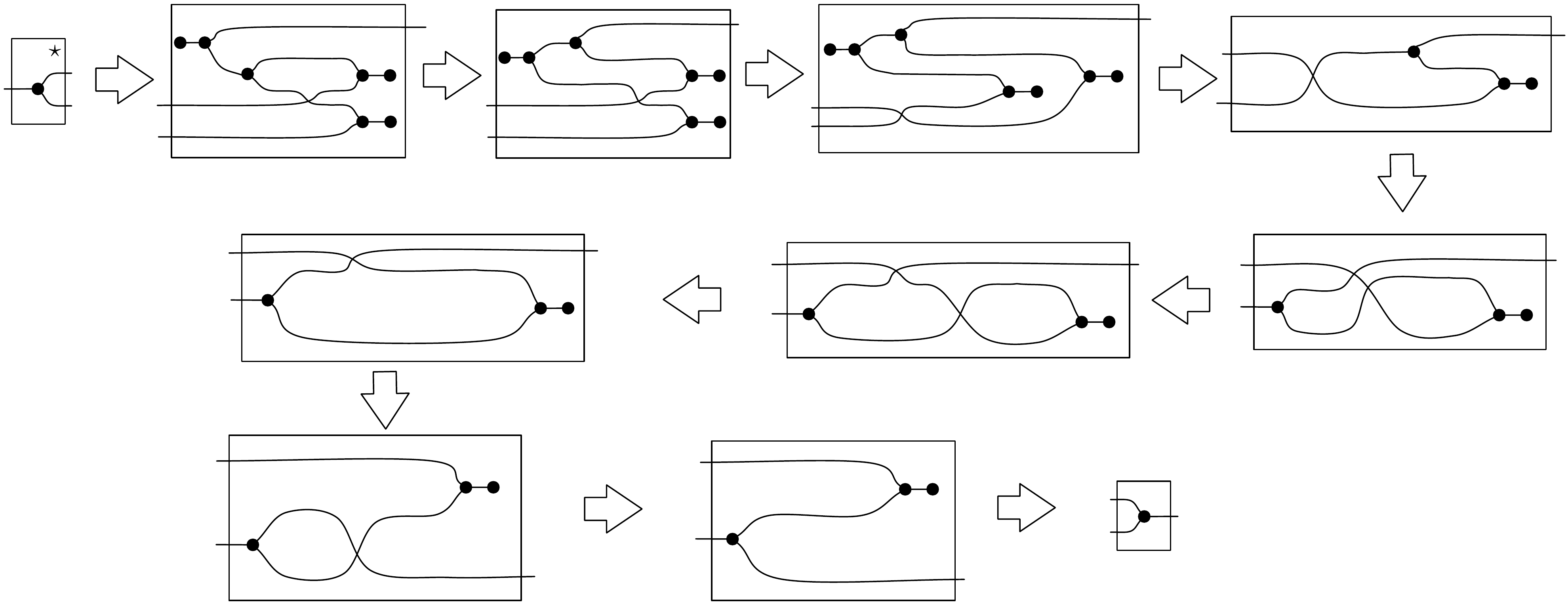}
\end{center}
The sequence of applied laws is: definition of $\coc{(\cdot)}$, \eqref{eq:bcomonassoc}, \eqref{eq:moveccpastsym}, \eqref{eq:Bsnake}, naturality of symmetry, axiom of SMCs, \eqref{eq:bcomoncomm}$^{\op}$, \eqref{eq:moveccpastsym}, \eqref{eq:bcomoncomm}, \eqref{eq:Bfrobmult}.

The remaining base cases of generators $\Bcomult$, $\Wmult$ and $\Wcomult$ are handled in an analogous way by using the Frobenius laws derived in~\ref{AppFrob}. The proof is concluded by examining the two inductive cases. For sequential composition:
   \begin{equation*}    \lower10pt\hbox{$\includegraphics[height=1cm]{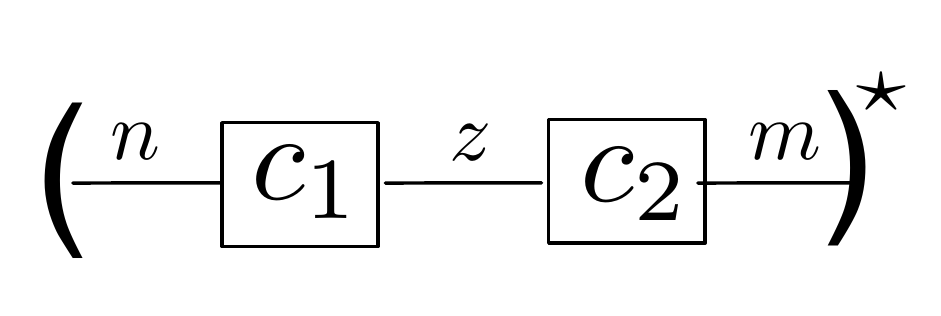}$}
  \!\eql{}\!
   \lower10pt\hbox{$\includegraphics[height=.9cm]{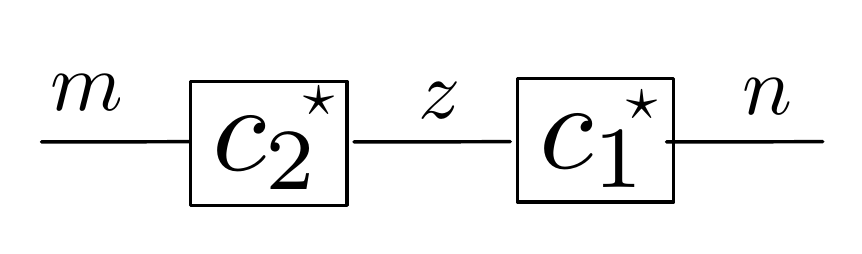}$}
  \!\eql{Ind. hyp.}\!
   \lower10pt\hbox{$\includegraphics[height=.9cm]{graffles/reflcompr.pdf}$}
  \!\eql{}\!
   \lower10pt\hbox{$\includegraphics[height=1cm]{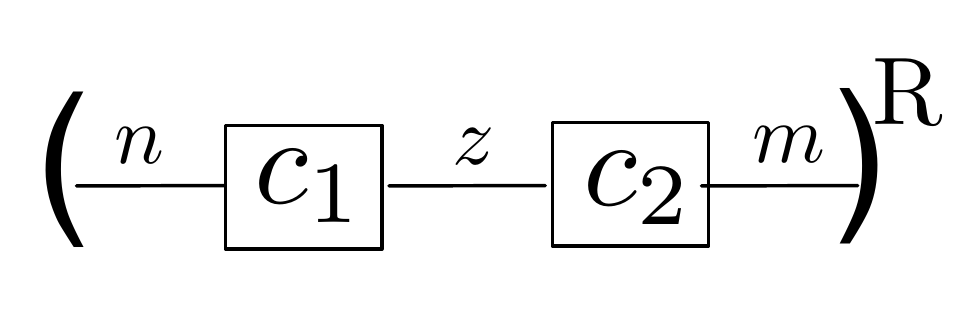}$}
   \end{equation*}
The derivation for the case of parallel composition $\tns$ is analogous.
\qed\end{proof}

%% file: source/D_AppendixLawsIB.tex
We verify the claim of Section~\ref{sec:cubetop}, by verifying that \eqref{eq:wbone}, \eqref{eq:BccscalarAxiomOne}, \eqref{eq:BccscalarAxiomTwo}, \eqref{eq:bbone}, \eqref{eq:WcccoscalarAxiomOne} and \eqref{eq:WcccoscalarAxiomTwo} are all derivable in $\IBR$. The following is the derivation of \eqref{eq:wbone}.
\begin{equation*}
\lower7pt\hbox{$\includegraphics[height=.7cm]{graffles/idzerocircuit.pdf}$}
\eql{\eqref{eq:bwbone}}
\lower8pt\hbox{$\includegraphics[height=.8cm]{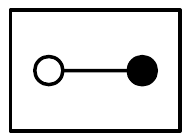}$}
\eql{\eqref{eq:BSepIBR}}
\lower9pt\hbox{$\includegraphics[height=.9cm]{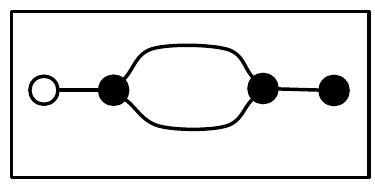}$}
\eql{\eqref{eq:lccb}}
\lower9pt\hbox{$\includegraphics[height=.9cm]{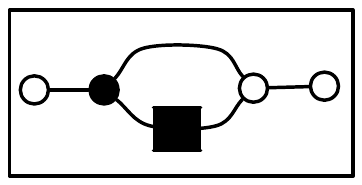}$}
\eql{\eqref{eq:scalarsum}}
  \lower8pt\hbox{$\includegraphics[height=.8cm]{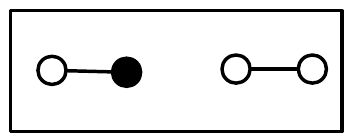}$}
\eql{\eqref{eq:bwbone}}
  \lower6pt\hbox{$\includegraphics[height=.6cm]{graffles/WBone.pdf}$}
\end{equation*}
The derivation of~\eqref{eq:bbone} is the ``photografic negative'' of the one of \eqref{eq:wbone}. We now show the derivations for \eqref{eq:BccscalarAxiomOne} and \eqref{eq:WcccoscalarAxiomOne}. For $l \neq 0$:
\begin{equation*}
\lower9pt\hbox{$\includegraphics[height=.8cm]{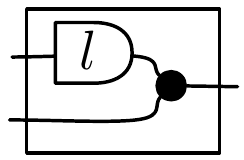}$}
\eql{\eqref{eq:lcmopIH}}
\lower9pt\hbox{$\includegraphics[height=.8cm]{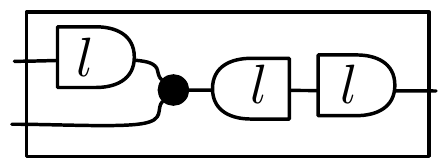}$}
\eql{\eqref{eq:scalarbcomult}$^{\op}$}
\lower9pt\hbox{$\includegraphics[height=.8cm]{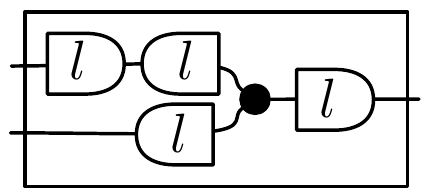}$}
\eql{\eqref{eq:lcmIH}}
\lower9pt\hbox{$\includegraphics[height=.8cm]{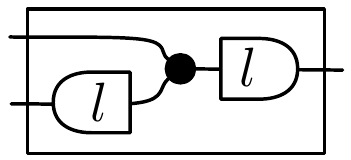}$}
\end{equation*}\noindent
\begin{equation*}
\lower9pt\hbox{$\includegraphics[height=.8cm]{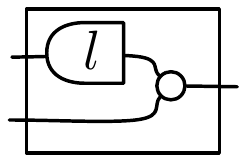}$}
\eql{\eqref{eq:lcmopIH}}
\lower9pt\hbox{$\includegraphics[height=.8cm]{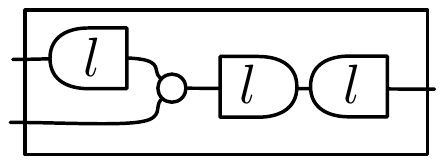}$}
\eql{\eqref{eq:scalarwmult}}
\lower9pt\hbox{$\includegraphics[height=.8cm]{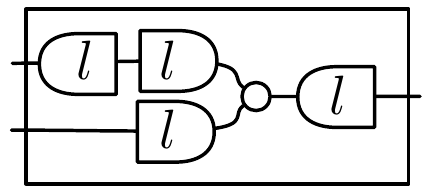}$}
\eql{\eqref{eq:lcmopIH}}
\lower9pt\hbox{$\includegraphics[height=.8cm]{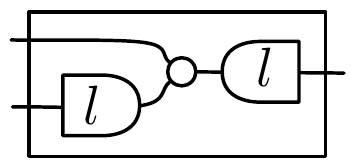}$}.
\end{equation*}
The zero cases:
\begin{equation*}
\lower9pt\hbox{$\includegraphics[height=.8cm]{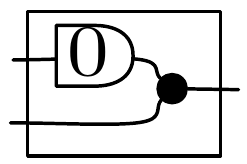}$}
\!\!\!\!\eql{\eqref{eq:zeroscalar}}\!\!\!\!
\lower9pt\hbox{$\includegraphics[height=.8cm]{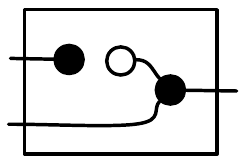}$}
\!\!\!\!\eql{\eqref{eq:Bfrobmult},\eqref{eq:lccb}}\!\!\!\!
\lower12pt\hbox{$\includegraphics[height=1.1cm]{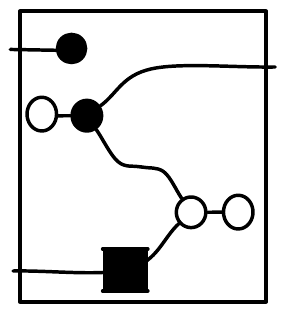}$}
\!\!\!\!\eql{\eqref{eq:unitsr}}\!\!\!\!
\lower10pt\hbox{$\includegraphics[height=.9cm]{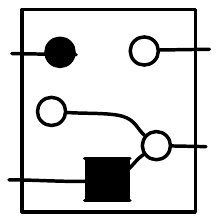}$}
\!\!\!\!\eql{\eqref{eq:wmonunitlaw},\eqref{eq:scalarwunit}}\!\!\!\!
\lower9pt\hbox{$\includegraphics[height=.8cm]{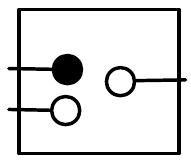}$}
\!\!\!\!\eql{\eqref{eq:bcomonunitlaw}}\!\!\!\!
\lower8pt\hbox{$\includegraphics[height=.7cm]{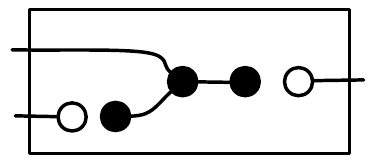}$}
\!\!\!\!\eql{\eqref{eq:zeroscalar},\eqref{eq:zeroscalar}$^{\op}$}\!\!\!\!
\lower9pt\hbox{$\includegraphics[height=.8cm]{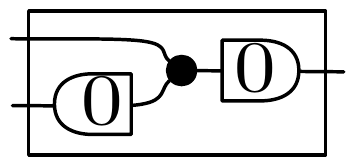}$}
\end{equation*}\noindent
\begin{equation*}
\lower9pt\hbox{$\includegraphics[height=.8cm]{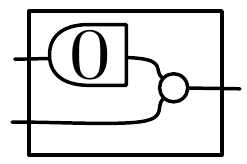}$}
\!\!\!\!\eql{\eqref{eq:zeroscalar}$^{\op}$}\!\!\!\!
\lower9pt\hbox{$\includegraphics[height=.8cm]{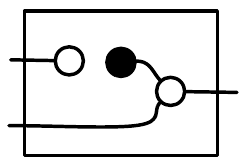}$}
\!\!\!\!\eql{\eqref{eq:Wfrobmult},\eqref{eq:rcc}}\!\!\!\!
\lower12pt\hbox{$\includegraphics[height=1.1cm]{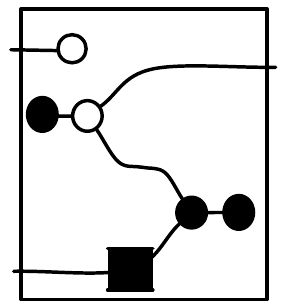}$}
\!\!\!\!\eql{\eqref{eq:unitsl}$^{\op}$}\!\!\!\!
\lower10pt\hbox{$\includegraphics[height=.9cm]{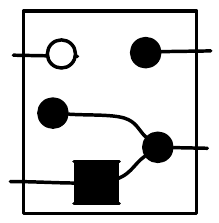}$}
\!\!\!\!\eql{\eqref{eq:bcomonunitlaw},\eqref{eq:scalarbcounit}$^{\op}$}\!\!\!\!
\lower9pt\hbox{$\includegraphics[height=.8cm]{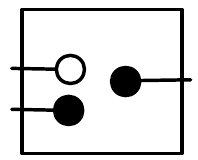}$}
\!\!\!\!\eql{\eqref{eq:wmonunitlaw},\eqref{eq:zeroscalar},\eqref{eq:zeroscalar}$^{\op}$}\!\!\!\!
\lower9pt\hbox{$\includegraphics[height=.8cm]{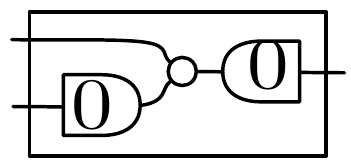}$}
\end{equation*}
The other two equations \eqref{eq:BccscalarAxiomTwo} and \eqref{eq:WcccoscalarAxiomTwo} are proven symmetrically. 